\pdfoutput=1
\def\arXiv{1} 
\documentclass[11pt]{article}
\usepackage{enumitem}
\usepackage{amssymb}
\usepackage{amsbsy}
\usepackage{amsmath}
\usepackage{amsfonts}
\usepackage{latexsym}
\usepackage{graphicx}
\usepackage{color}
\usepackage{xifthen}
\usepackage{xspace}
\usepackage{mathtools}
\usepackage{bbm}
\usepackage{multirow}
\usepackage[normalem]{ulem}
\usepackage{array}
\usepackage{booktabs}
\usepackage[utf8]{inputenc}
\usepackage[T1]{fontenc}
\usepackage{etoolbox}
\usepackage{threeparttable}
\usepackage{comment}
\usepackage{pifont}%
\newtoggle{restatements}

\newtoggle{heavyplots}

\usepackage{xfrac}
\newcommand{\notarxiv}[1]{foo}
\newcommand{\arxiv}[1]{ba}
\ifdefined \arXiv
	\renewcommand{\arxiv}[1]{#1}%
	\renewcommand{\notarxiv}[1]{}%
\else%
	\renewcommand{\arxiv}[1]{}%
	\renewcommand{\notarxiv}[1]{#1}%
\fi%

\arxiv{
	\usepackage{titlesec}
	\usepackage[page,titletoc,title]{appendix}
}

\usepackage{amsthm}
\usepackage{thmtools}
\usepackage{thm-restate}
\usepackage[linesnumbered,ruled,vlined]{algorithm2e}
\usepackage{setspace}
\usepackage[dvipsnames]{xcolor}
\usepackage{makecell}

\notarxiv{
\usepackage[utf8]{inputenc} %
\usepackage[T1]{fontenc}    %
\usepackage[hidelinks]{hyperref}       %
\usepackage{url}            %
\usepackage{booktabs}       %
\usepackage{amsfonts}       %
\usepackage{nicefrac}       %
\usepackage{microtype}      %
\usepackage{xcolor}
\usepackage{tabularx}
\usepackage[noend]{algpseudocode}
\usepackage{stackengine}
\usepackage{comment}
\usepackage[numbers,sort&compress,square]{natbib}
\bibliographystyle{IEEEtranN}
}

\makeatletter
\notarxiv{
\def\@IEEEsectpunct{~~}
}
\makeatother

\arxiv{
	\usepackage[top=1in, right=1in, left=1in, bottom=1in]{geometry}
	\usepackage[numbers,square]{natbib}
	\usepackage{url}
	\usepackage[hidelinks]{hyperref}
	\hypersetup{
		colorlinks=true,
		linkcolor=blue!70!black,
		citecolor=blue!70!black,
		urlcolor=blue!70!black}
}

\theoremstyle{plain}

\newtheorem{theorem}{Theorem}
\newtheorem{lemma}{Lemma}

\newtheorem{fact}{Fact}
\newtheorem{proposition}{Proposition}

\newtheorem{definition}{Definition}
\newtheorem*{definition*}{Definition}

\theoremstyle{definition}
\newtheorem{remark}{Remark}
\newenvironment{psketch}{%
	\proof}{\endproof}

\arxiv{
}

\makeatletter
\newcommand*{\propenum}[1]{%
	\expandafter\@propenum\csname c@#1\endcsname%
}

\newcommand*{\@propenum}[1]{%
	$\ifcase#1\or1\or1'\or2\or3\or42%
	\else\@ctrerr\fi$%

}
\AddEnumerateCounter{\propenum}{\@propenum}{1}
\makeatother

\newcommand{\mc}[1]{\mathcal{#1}}

\newcommand{\norm}[1]{\left\|{#1}\right\|} %
\newcommand{\lone}[1]{\norm{#1}_1} %
\newcommand{\ltwo}[1]{\norm{#1}_2} %
\newcommand{\linf}[1]{\norm{#1}_\infty} %
\newcommand{\lfro}[1]{\left\|{#1}\right\|_{\rm F}} %
\newcommand{\norms}[1]{\|{#1}\|} %

\newcommand{\opnorm}[1]{\norm{#1}_{\mathrm{op}}}  %
\newcommand{\lones}[1]{\norms{#1}_1} %
\newcommand{\ltwos}[1]{{\norms{#1}}_2} %
\newcommand{\linfs}[1]{\norms{#1}_\infty} %
\newcommand{\R}{\mathbb{R}} %
\newcommand{\N}{\mathbb{N}} %
\newcommand{\E}{\mathbb{E}\,} %
\newcommand{\I}{\mathbb{I}} %

\newcommand{\ai}[1][i]{A_{#1:}}
\newcommand{\aj}[1][j]{A_{:#1}}
\newcommand{\nnz}{\mathsf{nnz}}

\usepackage{xcolor}

\SetCommentSty{mycommfont}
\SetKwInput{KwInput}{Input}                %
\SetKwInput{KwOutput}{Output}              %
\SetKwInput{KwReturn}{Return}              %

\newcommand{\InnerLoop}{\mathtt{InnerLoop}}
\newcommand{\InnerLoopApprox}{\mathtt{InnerLoop}}

\newcommand{\OuterLoop}{\mathtt{OuterLoop}}

\newcommand{\AddDense}{\mathtt{AddDense}}
\newcommand{\DenseStep}{\mathtt{DenseStep}}
\newcommand{\AddSparse}{\mathtt{AddSparse}}
\newcommand{\MultSparse}{\mathtt{MultSparse}}
\newcommand{\Scale}{\mathtt{Scale}}

\newcommand{\Norm}{\mathtt{Norm}}
\newcommand{\Get}{\mathtt{Get}}

\newcommand{\SumUp}{\mathtt{UpdateSum}}
\newcommand{\Sample}{\mathtt{Sample}}

\newcommand{\GetSum}{\mathtt{GetSum}}

\newcommand{\GetNorm}{\mathtt{GetNorm}}

\newcommand{\Initialize}{\mathtt{Init}}

\newcommand{\AEM}{\mathtt{ApproxExpMaintainer}}
\newcommand{\AEMS}{\mathtt{AEM}}

\newcommand{\IMS}{\mathtt{IM}}
\newcommand{\IM}{\mathtt{IterateMaintainer}}
\newcommand{\WIMS}{\mathtt{WIM}}
\newcommand{\WIM}{\mathtt{WeightedIterateMaintainer}}
\newcommand{\CIMS}{\mathtt{CIM}}
\newcommand{\CIM}{\mathtt{CenteredIterateMaintainer}}

\makeatletter
\let\oldnl\nl%
\newcommand{\nonl}{\renewcommand{\nl}{\let\nl\oldnl}}%
\makeatother

\DeclareMathOperator*{\argmax}{arg\,max}

\DeclareMathOperator*{\argmin}{arg\,min}

\providecommand{\diag}{\mathop{\rm diag}}

\providecommand{\sign}{\mathop{\rm sign}}

\newcommand{\dist}{\mathop{\rm dist}}

\newcommand{\ceil}[1]{\left\lceil{#1}\right\rceil}
\newcommand{\floor}[1]{\left\lfloor{#1}\right\rfloor}

\newcommand{\half}{\frac{1}{2}}
\newcommand{\defeq}{\coloneqq}

\newcommand{\grad}{\nabla}

\newcommand{\xset}{\mathcal{X}}
\newcommand{\yset}{\mathcal{Y}}
\newcommand{\zset}{\mathcal{Z}}

\newcommand{\eps}{\epsilon}

\renewcommand{\O}[1]{O\left( #1 \right)}
\newcommand{\Otils}[1]{\widetilde{O}( #1 )}
\newcommand{\Otil}[1]{\widetilde{O}\left( #1 \right)}
\newcommand{\Otilb}[1]{\widetilde{O}\left( #1 \right)}
\newcommand{\Otilbig}[1]{\widetilde{O}\Big( #1 \Big)}
\renewcommand{\floor}[1]{\lfloor #1 \rfloor}

\newcommand{\innermid}{\nonscript\;\delimsize\vert\nonscript\;}
\newcommand{\activatebar}{%
	\begingroup\lccode`\~=`\|
	\lowercase{\endgroup\let~}\innermid 
	\mathcode`|=\string"8000
}

\newcommand{\inner}[2]{\left<#1,#2\right>}
\newcommand{\inners}[2]{\big<#1,#2\big>}

\NewDocumentCommand{\prox}{ m m 
O{\alpha}}{\mathrm{Prox}_{#1}^{#3}(#2)}

\newcommand{\nA}{\norm{A}_{2\rightarrow\infty} }
\newcommand{\ball}{\mathbb{B}}

\newcommand{\x}{^\mathsf{x}}
\newcommand{\y}{^\mathsf{y}}

\newcommand{\Ex}[1][\xi]{\E}

\newcommand{\ones}{\boldsymbol{1}}

\newcommand{\indic}[1]{\I_{\{#1\}}}

\newcommand{\tg}{\tilde{g}}
\newcommand{\td}{\tilde{\delta}}

\newcommand{\rcs}{\mathsf{rcs}}

\newcommand{\cs}{\mathrm{cs}}

\newcommand{\1}{\mathbf{1}}

\newcommand{\gap}{\mathrm{Gap}}
\newcommand{\vepsi}{\varepsilon_{\textup{inner}}}
\newcommand{\vepsout}{\varepsilon_{\textup{outer}}}

\newcommand{\veps}{\varepsilon}
\newcommand{\tveps}{\tilde{\veps}}
 \SetKwInOut{Parameter}{Parameters}
\SetKwComment{Comment}{$\triangleright$}{}
\SetCommentSty{color{black}}

\newcommand{\truncate}{\mathsf{truncate}}

\newcommand{\Lmv}{L_{\mathsf{mv}}}
\newcommand{\Lrc}{L_{\mathsf{rc}}}
\newcommand{\Lco}{L_{\mathsf{co}}}

\newcommand{\ellone}{\ell_1\text{-}\ell_1}
\newcommand{\elltwoone}{\ell_2\text{-}\ell_1}
\newcommand{\elltwo}{\ell_2\text{-}\ell_2}

\newcommand{\lttco}{L_{\mathsf{co}}^{2, 2}}
\newcommand{\looco}{L_{\mathsf{co}}^{1, 1}}
\newcommand{\ltoco}{L_{\mathsf{co}}^{2, 1}}
\newcommand{\ltooco}{L_{\mathsf{co}}^{2, 1, (1)}}
\newcommand{\ltotco}{L_{\mathsf{co}}^{2, 1, (2)}}
\newcommand{\ltohco}{L_{\mathsf{co}}^{2, 1, (3)}}
\newcommand{\lttcop}{\left(L_{\mathsf{co}}^{2, 2}\right)}
\newcommand{\loocop}{\left(L_{\mathsf{co}}^{1, 1}\right)}
\newcommand{\ltocop}{\left(L_{\mathsf{co}}^{2, 1}\right)}
\newcommand{\ltoocop}{\left(L_{\mathsf{co}}^{2, 1, (1)}\right)}

\newcommand{\ScM}[0]{\mathtt{ScaleMaintainer}}
\newcommand{\GlobM}[0]{\AEM}
\newcommand{\UpdateSum}{\mathtt{UpdateSum}}

\newcommand{\Del}{\mathtt{Del}}
\newcommand{\tx}[0]{\tilde{x}}
\newcommand{\bg}[0]{\bar{\delta}}
\newcommand{\bx}[0]{\bar{x}}

\newcommand{\tnnz}{\nnz'}
\newcommand{\lam}{\lambda}

\newcommand{\ltnfrac}{\log(\tfrac{n}{\veps\lam})}

\newcommand{\epsaprx}{\varphi}

\newcommand{\epsscm}{\veps_{\textup{scm}}}
\newcommand{\lamscm}{\lambda_{\textup{scm}}}
\newcommand{\omscm}{\omega_{\textup{scm}}}

\newcommand{\clip}{\mathrm{clip}}
\newcommand{\setup}{($\zset$, $\norm{\cdot}_{\cdot}$, $r$, $\Theta$, 
	$\clip$) } 
\title{Coordinate Methods for Matrix Games}

\arxiv{
\author{Yair Carmon ~~~ Yujia Jin ~~~  Aaron Sidford ~~~ Kevin 
Tian\\
	\texttt{\{\href{mailto:yairc@stanford.edu}{yairc},%
		\href{mailto:yujiajin@stanford.edu}{yujiajin},%
		\href{mailto:sidford@stanford.edu}{sidford},%
		\href{mailto:kjtian@stanford.edu}{kjtian}\}@stanford.edu}}
\date{}
}

\notarxiv{
\author{
 \IEEEauthorblockN{Yair Carmon\IEEEauthorrefmark{1}\IEEEauthorrefmark{2}, Yujia Jin\IEEEauthorrefmark{2}, Aaron Sidford\IEEEauthorrefmark{2}, Kevin Tian\IEEEauthorrefmark{2}}
\IEEEauthorblockA{
	\IEEEauthorrefmark{1}Tel Aviv University
	\IEEEauthorrefmark{2}Stanford University
	\\
	\texttt{\{yairc, yujiajin, sidford, kjtian\}@stanford.edu}}
}
}

\begin{document}

\pagenumbering{gobble}
\maketitle

\begin{abstract}%
We develop primal-dual coordinate methods for solving bilinear 
saddle-point problems of the form $\min_{x \in \mathcal{X}} 
\max_{y\in\mathcal{Y}} y^\top A x$ which contain linear programming, 
classification, and regression as special cases. Our methods push existing 
fully stochastic sublinear methods and variance-reduced methods towards 
their limits in terms of per-iteration complexity and sample complexity. We 
obtain nearly-constant per-iteration complexity by designing efficient data 
structures 
leveraging Taylor approximations to the exponential and a 
binomial heap. We improve sample complexity via low-variance gradient 
estimators using dynamic sampling distributions that depend on both the 
iterates and the magnitude of the matrix entries. 

Our runtime bounds improve upon those of existing primal-dual methods 
by a factor depending on \emph{sparsity measures} of the $m$ by $n$ 
matrix $A$. For example, when rows and columns have constant 
$\ell_1/\ell_2$ norm ratios, we offer improvements by a factor of $m+n$ 
in the fully stochastic setting and $\sqrt{m+n}$ in the variance-reduced 
setting. We apply our methods to computational geometry problems, i.e. minimum 
enclosing ball, maximum inscribed ball, and linear regression, and obtain improved complexity bounds.
For linear regression with an elementwise nonnegative 
matrix, our guarantees improve on exact gradient methods by a factor of 
$\sqrt{\nnz(A)/(m+n)}$. %
\end{abstract}

\arxiv{\newpage
\setcounter{tocdepth}{2}
\tableofcontents}

\arxiv{\newpage}
\pagenumbering{arabic}

\section{Introduction}

Bilinear minimax problems of the form 
\begin{equation}\label{eq:problem}
\min_{x\in\xset}\max_{y\in\yset}y^\top A x
~\text{where}~
A\in\R^{m\times n},
\end{equation}
are fundamental to machine learning, economics and 
theoretical 
computer 
science~\citep{MinskyP87,VonNeumannM44,Dantzig53}. We focus on three 
important settings characterized by 
different domain geometries. 
When $\xset$ and $\yset$ are probability simplices---which we call the 
$\ellone$ setting---the problem~\eqref{eq:problem} 
corresponds to a zero-sum matrix game and also to a linear 
program in canonical feasibility form. The $\elltwoone$ setting, 
where  $\xset$ is a Euclidean ball and $\yset$ is a simplex, is useful for 
linear classification (hard-margin support vector machines) as well as 
problems in computational geometry~\cite{ClarksonHW10}. Further, the $\elltwo$ 
setting, where both $\xset$ and $\yset$ are Euclidean balls (with general center), includes linear regression.

Many  problems of practical interest are \emph{sparse}, i.e., the number of nonzero elements in $A$, which we denote by $\nnz$, satisfies $\nnz \ll m 
n$. Examples include: linear programs with constraints involving few 
variables, linear classification with 1-hot-encoded features, and linear 
systems that arise from physical models with local interactions. The 
problem description size $\nnz$ plays a central role in several runtime 
analyses of algorithms for solving the problem~\eqref{eq:problem}. 

However, sparsity is not an entirely satisfactory measure of instance 
complexity: it 
is not continuous in the elements of $A$ and consequently it cannot 
accurately reflect the simplicity of ``nearly sparse'' instances with many 
small (but nonzero) elements. Measures of \emph{numerical sparsity}, such 
as the $\ell_1$ to $\ell_2$ norm ratio, can fill this gap~\cite{GuptaS18}. 
Indeed, many 
problems encountered in practice are numerically sparse. Examples include: 
linear programming constraints of the form $x_1 \ge \frac{1}{n}\sum_i x_i 
$, linear 
classification with neural network activations as features, and linear 
systems arising from physical models with interaction whose strength 
decays with distance. 

Existing bilinear minimax solvers do not exploit the numerical sparsity of 
$A$ and their runtime guarantees do not depend on it---the basic limitation 
of these methods is that they do not directly access the large matrix 
entries, and instead sample the full columns and rows in which they 
occur. 
To overcome this limitation, we propose methods that access $A$ a single 
entry at a time, leverage  
numerical sparsity by accessing larger coordinates more frequently, 
and enjoy runtime guarantees that depend explicitly on 
numerical sparsity measures. For numerically sparse large-scale instances 
our runtimes are substantially better than the previous state-of-the-art. 
Moreover, our runtimes subsume the previous state-of-the-art dependence 
on $\nnz$ and $\rcs$,  the maximum number 
of nonzeros in any row or column.

In addition to proposing algorithms with improved runtimes, we develop two 
techniques that may be of broader interest. First, we design non-uniform 
sampling schemes that minimize regret bounds; we use a general framework 
that unifies the Euclidean and (local norms) simplex geometries, possibly 
facilitating future extension. Second, we build a data structure capable of  
efficiently maintaining and sampling from multiplicative weights iterations (i.e.\ entropic projection) \emph{with a fixed dense component}. This 
data structure overcomes limitations of existing techniques for maintaining 
entropic projections and we believe it may prove effective in other settings  
where such projections appear.

\notarxiv{
\begin{table*}[!ht]
	\centering
	\renewcommand{\arraystretch}{1.75}
	\caption{Dependence on $A$ for different methods in different  
			geometries.} \label{table:L}
	\begin{tabular}{c c c c}
		\toprule
		{\bf } & { $\Lmv$ (matrix-vector) } & { $\Lrc$ (row-column) } 
		& { 
		$\Lco$ 
		(coordinate) } \\
		\midrule
		{\bf $\ell_1$-$\ell_1$} & $\max_{i,j}|A_{ij}|$ & $\max_{i,j}|A_{ij}|$ & 
		$\max\Big\{ \max_{i} \ltwo{\ai}, \max_j \ltwo{\aj} \Big\}$ \\ 
		{\bf $\ell_2$-$\ell_1$} & $\max_{i}\ltwo{\ai}$ & $\max_{i}\ltwo{\ai}$ 
		& 
		$\max\Big\{ \max_{i} \lone{\ai} , \lfro{A}\Big\}^\dagger$ 
	\\
		{\bf $\ell_2$-$\ell_2$} & $\opnorm{A}$ & $\norm{A}_\mathrm{F}$ &  
		$\max\Big\{ \sqrt{\sum_i \lone{\ai}^2 }, \sqrt{\sum_j \lone{\aj}^2 } 
		\Big\}$\\
		\bottomrule
	\end{tabular}

	\vspace{2ex}
	\parbox[left]{0.8\textwidth}{\small 
		\emph{Comments:} $\ai$ and $\aj$ are the $i$th row and $j$th column of $A$, respectively. 
		Numerically sparse instances satisfy $\Lco = 
		O(\Lrc)$. $^\dagger\,$In the 
		$\elltwoone$ setting 
		we can also achieve, via alternative sampling schemes, 
		$\Lco=\Lrc\sqrt{\rcs}$ and $\Lco = 
		\max\{ \max_{i} \lone{\ai} , \sqrt{\max_{i} \lone{\ai}\max_{j} 
			\lone{\aj}}\}$.
	}

\vspace{12pt}

	\centering
	\renewcommand{\arraystretch}{1.25}
	\bgroup
	\everymath{\displaystyle}
	\caption{Comparison of iterative methods for 
			bilinear problems}\label{table:runtimes}
	\begin{tabular}{c c c}	
		\toprule
		{ Method } & { Iteration cost } & {Total runtime} \\
		\midrule
		Exact gradient~\cite{Nemirovski04,Nesterov07} & $O(\nnz)$ & $ 
		\Otilbig{\nnz\cdot{\Lmv}\cdot{\eps^{-1}}}$ 
		\\ 
		Row-column~\cite{GrigoriadisK95,ClarksonHW10,BalamuruganB16} & 
		$O(n+m)$ & %
		$\Otilbig{(m + n) \cdot {\Lrc^2} \cdot{\epsilon^{-2}}}$ \\
		Row-column VR~\cite{BalamuruganB16,CarmonJST19} &  $O(n+m)$
		 & $\Otilbig{\nnz+\sqrt{\nnz\cdot 
		 (m+n)}\cdot{\Lrc}\cdot{\eps^{-1}}}$ 
		 \\
		\midrule
		Sparse row-col (folklore) & 
		$\Otil{\rcs}$ & %
		$\Otilbig{\rcs \cdot {\Lrc^2}\cdot{\epsilon^{-2}}}$ \\
		Sparse row-col VR (Appendix~\ref{app:rcs-vr}) & $\Otil{\rcs}$ & 
		$\Otilbig{\nnz + \sqrt{\nnz\cdot 
				\rcs}\cdot{\Lrc}\cdot{\eps^{-1}}}$  \\
		\midrule
		\textbf{Coordinate (Section~\ref{sec:sublinear})} & $\Otil{1}$ & $ 
		\Otilbig{\nnz+{\Lco^2}\cdot{\eps^{-2}}}$ 
		 \\
		\textbf{Coordinate VR (Section~\ref{sec:vr})} & $\Otil{1}$ &
		$\Otilbig{\nnz %
		+\sqrt{\nnz}\cdot{\Lco}\cdot{\eps^{-1}}} $%
		\\
		\bottomrule
	\end{tabular}
\egroup

	\vspace{2ex}
	\parbox[left]{0.8\textwidth}{\small 
	\emph{Comments:} 	 
	$\nnz$ denotes the number of nonzeros in $A\in\R^{m\times n}$ and 
	$\rcs\le \max\{m,n\}$ denotes the maximum number of nonzeros in any 
	row and 
	column of $A$. The quantities $\Lmv,\Lco$ and $\Lrc$ depend on 
	problem 
	geometry (see Table~\ref{table:L}). 
	}

	\vspace{12pt}
	
	\centering
	\renewcommand{\arraystretch}{1.4}
	\bgroup
	\caption{Comparison of complexity for different applications.}\label{table:applications}
	\begin{tabular}{c c c}
		\toprule
		Task                        & Method                    & Runtime $(m\ge n)$ \\ 
		\midrule
		\multirow{2}{*}{\makecell{MaxIB}}      & 
		\citet{ZhuLY16} &    $ 
		\Otilb{mn+{\rho m\sqrt{n}}\cdot{\epsilon^{-1}}}$        \\ %
		& Our method (Theorem~\ref{thm:MaxIB})                &    $ 
		\Otilb{\nnz + 
		{\rho\sqrt{\nnz\cdot\rcs}}\cdot{\epsilon^{-1}}}$      
		  \\ 
		\hline
		\multirow{2}{*}{\makecell{MinEB}}      & 
		\citet{ZhuLY16} &   $ 
		\Otilb{mn+{m\sqrt{n}}\cdot{{\eps^{-1/2}}}}$          \\ %
		& Our method (Theorem~\ref{thm:MinEB})               &  $ 
		\Otilb{\nnz+\sqrt{\nnz\cdot\rcs}\cdot {\eps^{-1/2}}}^\dagger$          
		\\ 
		\midrule
		\multirow{3}{*}{\makecell{Regression\\($A^\top A\succeq\mu I$)}} & 
		AGD~\citep{Nesterov83} &  $ 
		\Otilb{\nnz\cdot{\opnorm{A}}{\frac{1}{\sqrt{\mu}}}}$ \\ 
		& \citet{GuptaS18}&  $ 
		\Otilbig{\nnz+\nnz^{2/3}\cdot{\Big(\sum_{i\in[n]}\|A\|_\mathrm{F}\cdot
		 \|\ai\|_1\cdot \|\ai\|_2\Big)^{1/3}}\frac{1}{\sqrt{\mu}}}$ \\ 
		& Our method (Theorem~\ref{thm:reg})               & 
		$\Otilbig{\nnz+\sqrt{\nnz}\cdot 
			{\max\Big\{\sqrt{\sum_i\lones{\ai}^2}, 
				\sqrt{\sum_j\lones{\aj}^2}\Big\}}{\frac{1}{\sqrt{\mu}}}}$ \\ 
			\bottomrule
			
	\end{tabular}
	
		\vspace{2ex}
		\parbox[left]{0.8\textwidth}{\small 
		\emph{Comments:} 
		 $\rho$ denotes the radii ratio of the minimum ball enclosing the 
		rows of $A$ and maximum ball inscribed in 
		them. $^\dagger\,$For MinEB, 
		we state an upper bound here and refer readers to Section~\ref{app:minEB} for a more fine-grained runtime.
	}
	\egroup
\end{table*} }

\subsection{Our results}\label{sec:our_results}
Table~\ref{table:runtimes} summarizes our runtime guarantees and puts 
them in the context of the best existing results. We consider methods that 
output (expected) $\epsilon$-accurate solutions of the saddle-point 
problem~\eqref{eq:problem}, namely a pair $x,y$ satisfying 
\begin{equation*}
\E\left[ \max_{v\in\yset}v^\top Ax-\min_{u\in\xset}y^\top A u\right]\le \eps.
\end{equation*}
The algorithms in Table~\ref{table:runtimes} are all 
iterative solvers for the general problem $\min_{x\in\xset}\max_{y\in\yset} 
f(x,y)$, specialized to $f(x,y)=y^\top A x$. Each algorithm presents a 
different tradeoff between per-iteration complexity and the required 
iteration count, corresponding to the matrix access modality: exact 
gradient methods compute matrix-vector products in each iteration, 
row-column stochastic gradient methods sample a row and a column in 
each 
iteration, and our proposed coordinate methods take this tradeoff to an  
extreme by sampling a single coordinate of the matrix per 
iteration.\footnote{
	Interior point methods offer 
	an alternative tradeoff between iteration cost and iteration count: the 
	number of required iterations depends on $1/\epsilon$ only 
	logarithmically, but every iteration is costly, requiring a linear system 
	solution which at present takes time $\Omega(\min\{m,n\}^2)$. In the $\ellone$ geometry, 
	the best known runtimes for interior point methods are  
	$\Otils{(\nnz + \min\{m,n\}^2) \sqrt{\min\{m,n\}}}$~\cite{LeeS15}, 
	$\Otils{\max\{m,n\}^\omega}$~\cite{CohenLS18}, and $\Otils{mn + 
	\min\{m,n\}^3}$~\cite{BrandKSS20}. In this 
	paper we are mainly interested in the large-scale low-accuracy regime with $L/\eps<\min(m,n)$ where the runtimes described in 
	Table~\ref{table:runtimes} are favorable (with the exception of \cite{BrandKSS20} in certain cases). 
	Our methods take only few 
	passes over the data, which are not the case for many interior-point methods~\cite{LeeS15,CohenLS18}.  Also, our methods do not rely on a general (ill-conditioned) linear system solver, which is a key ingredient in interior point methods. 
}
In addition, variance reduction (VR) schemes combine both fast stochastic 
gradient computations and infrequent exact gradient computations, 
maintaining the amortized per-iteration cost of the stochastic scheme and 
reducing the total iteration count for sufficiently small $\epsilon$. 

The runtimes in Table~\ref{table:runtimes} depend on the numerical range of $A$ through a matrix norm $L$ 
that changes with both the problem geometry and the type of matrix 
access; 
we use $\Lmv$, $\Lrc$ and $\Lco$ to denote the constants corresponding 
to matrix-vector products, row-column queries and coordinated queries, 
respectively. 
Below, we describe these runtimes in detail. 
In the settings we study, our results are the first theoretical 
demonstration of runtime gains arising from sampling a single coordinate of $A$ at a time, as opposed to entire rows and columns.

\arxiv{
\begin{table}
	\centering
	\renewcommand{\arraystretch}{1.25}
	\begin{tabular}{c c c c}
		\toprule
		{\bf } & { $\Lmv$ (matrix-vector) } & { $\Lrc$ (row-column) } 
		& { 
		$\Lco$ 
		(coordinate) } \\
		\midrule
		{\bf $\ell_1$-$\ell_1$} & $\max_{i,j}|A_{ij}|$ & $\max_{i,j}|A_{ij}|$ & 
		$\max\Big\{ \max_{i} \ltwo{\ai}, \max_j \ltwo{\aj} \Big\}$ \\ 
		{\bf $\ell_2$-$\ell_1$} & $\max_{i}\ltwo{\ai}$ & $\max_{i}\ltwo{\ai}$ 
		& 
		$\max\Big\{ \max_{i} \lone{\ai} , \lfro{A}\Big\}^\dagger$ 
	\\
		{\bf $\ell_2$-$\ell_2$} & $\opnorm{A}$ & $\norm{A}_\mathrm{F}$ &  
		$\max\Big\{ \sqrt{\sum_i \lone{\ai}^2 }, \sqrt{\sum_j \lone{\aj}^2 } 
		\Big\}$\\
		\bottomrule
	\end{tabular}
	\caption{\small\textbf{Dependence on $A$ for different methods in different  
		geometries.} Comments: $\ai$ and $\aj$ denote the $i$th row and $j$th column of $A$, respectively. 
		Numerically sparse instances satisfy $\Lco = 
		O(\Lrc)$. $^\dagger\,$In the 
		$\elltwoone$ setting 
		we can also achieve, via alternative sampling schemes, 
		$\Lco=\Lrc\sqrt{\rcs}$ and $\Lco = 
		\max\{ \max_{i} \lone{\ai} , \sqrt{\max_{i} \lone{\ai}\max_{j} 
			\lone{\aj}}\}$. 
	}\label{table:L}

\vspace{-0.3cm}

\vspace{1cm}

	\centering
	\renewcommand{\arraystretch}{1.25}
	\bgroup
	\everymath{\displaystyle}
	\begin{tabular}{c c c}	
		\toprule
		{ Method } & { Iteration cost } & {Total runtime} \\
		\midrule
		Exact gradient~\cite{Nemirovski04,Nesterov07} & $O(\nnz)$ & $ 
		\Otilbig{\nnz\cdot{\Lmv}\cdot{\eps^{-1}}}$ 
		\\ 
		Row-column~\cite{GrigoriadisK95,ClarksonHW10,BalamuruganB16} & 
		$O(n+m)$ & %
		$\Otilbig{(m + n) \cdot {\Lrc^2}\cdot{\epsilon^{-2}}}$ \\
		Row-column VR~\cite{BalamuruganB16,CarmonJST19} &  $O(n+m)$
		 & $\Otilbig{\nnz+\sqrt{\nnz\cdot 
		 (m+n)}\cdot{\Lrc}\cdot{\eps^{-1}}}$ 
		 \\
		\midrule
		Sparse row-col (folklore) & 
		$\Otil{\rcs}$ & %
		$\Otilbig{\rcs \cdot {\Lrc^2}\cdot{\epsilon^{-2}}}$ \\
		Sparse row-col VR (Appendix~\ref{app:rcs-vr}) & $\Otil{\rcs}$ & 
		$\Otilbig{\nnz + \sqrt{\nnz\cdot 
				\rcs}\cdot{\Lrc}\cdot{\eps^{-1}}}$  \\
		\midrule
		\textbf{Coordinate (Section~\ref{sec:sublinear})} & $\Otil{1}$ & $ 
		\Otilbig{\nnz+{\Lco^2}\cdot{\eps^{-2}}}$ 
		 \\
		\textbf{Coordinate VR (Section~\ref{sec:vr})} & $\Otil{1}$ &
		$\Otilbig{\nnz %
		+\sqrt{\nnz}\cdot{\Lco}\cdot{\eps^{-1}}} $%
		\\
		\bottomrule
	\end{tabular}
\egroup
	\caption{\small\textbf{Comparison of iterative methods for 
	bilinear problems.} Comments: $\nnz$ denotes the number of nonzeros in $A\in\R^{m\times n}$ and 
	$\rcs\le \max\{m,n\}$ denotes the maximum number of nonzeros in any 
	row and 
	column of $A$. The quantities $\Lmv,\Lco$ and $\Lrc$ depend on 
	problem 
	geometry (see Table~\ref{table:L}).
	}\label{table:runtimes}
	
	\vspace{16pt}
	
	\centering
	\renewcommand{\arraystretch}{1.4}
	\bgroup
	\begin{tabular}{c c c}
		\toprule
		Task                        & Method                    & Runtime \\ 
		\midrule
		\multirow{2}{*}{\makecell{MaxIB}}      & 
		\citet{ZhuLY16} &    $ 
		\Otilb{mn+{\rho m\sqrt{n}}\cdot{\epsilon^{-1}}}$        \\ %
		& Our method (Theorem~\ref{thm:MaxIB})                &    $ 
		\Otilb{\nnz + 
		{\rho\sqrt{\nnz\cdot\rcs}}\cdot{\epsilon^{-1}}}^\dagger$      
		  \\ 
		\hline
		\multirow{2}{*}{\makecell{MinEB\\(when $m\ge n$)}}      & 
		\citet{ZhuLY16} &   $ 
		\Otilb{mn+{m\sqrt{n}}\cdot{{\eps^{-1/2}}}}$          \\ %
		& Our method (Theorem~\ref{thm:MinEB})               &  $ 
		\Otilb{\nnz+\sqrt{\nnz\cdot\rcs}\cdot {\eps^{-1/2}}}^\dagger$          
		\\ 
		\midrule
		\multirow{3}{*}{\makecell{Regression\\($A^\top A\succeq\mu I$)}} & 
		AGD~\citep{Nesterov83} &  $ 
		\Otilb{\nnz\cdot{\opnorm{A}}{\frac{1}{\sqrt{\mu}}}}$ \\ 
		& \citet{GuptaS18} &  $ 
		\Otilbig{\nnz+\nnz^{2/3}\cdot{\Big(\sum_{i\in[n]}\|A\|_\mathrm{F}\cdot
		 \|\ai\|_1\cdot \|\ai\|_2\Big)^{1/3}}\frac{1}{\sqrt{\mu}}}$ \\ 
		& Our method (Theorem~\ref{thm:reg})               & 
		$\Otilbig{\nnz+\sqrt{\nnz}\cdot 
			{\max\Big\{\sqrt{\sum_i\lones{\ai}^2}, 
				\sqrt{\sum_j\lones{\aj}^2}\Big\}}{\frac{1}{\sqrt{\mu}}}}$ \\ 
			\bottomrule
	\end{tabular}
	\caption{\small\textbf{Comparison of complexity for different applications.}
		Comments: $\rho$ denotes the radii ratio of the minimum ball enclosing the 
		rows of $A$ and maximum ball inscribed in 
		them. $^\dagger\,$For MaxIB and MinEB, 
		we refer the reader to Section~\ref{app:minEB} for a more fine-grained runtime bound.}
		\label{table:applications}
	\egroup
\end{table}
 }

\paragraph{Coordinate stochastic gradient methods.}
We develop coordinate stochastic 
gradient estimators which allow  per-iteration cost  $\Otil{1}$ and  
iteration count  
$\Otil{n+m+(\frac{\Lco}{\epsilon})^2}$. We define $\Lco$  in 
Table~\ref{table:L}; for each domain geometry, the quantity 
$\frac{\Lco}{\Lrc}$ is a measure of the numerical sparsity of $A$, satisfying 
 \begin{equation*}
 1 \le \frac{\Lco^2}{\Lrc^2} \le \rcs.
 \end{equation*}
Every iteration of our method requires sampling an element in a row or a 
column with probability proportional to its entries. Assuming a matrix 
access model that allows such sampling in time $\Otil{1}$ (similarly 
to~\cite{AzarGRK13,SidfordWWY18,GilyenLT18}), the total runtime of our 
method is 
$\Otil{n+m+(\frac{\Lco}{\epsilon})^2}$. In this case, for numerically sparse 
problems such that $\Lco = O(\Lrc)$, the proposed coordinate 
methods
outperform row-column sampling 
by a factor of $m+n$. Moreover, 
the bound $\Lco^2 \le \Lrc^2(m+n)$ implies that our runtime is never 
worse than that of row-column methods. When only coordinate access 
to the matrix $A$ is initially available, we may implement the required sampling 
access via 
preprocessing in time $O(\nnz)$. This changes the runtime to $\Otil{\nnz 
+   (\frac{\Lco}{\epsilon})^2}$, so that the 
 comparison above holds only when $(\frac{\Lco}{\epsilon})^2 = 
\tilde{\Omega}(\nnz)$. In that regime, the variance  
reduction technique we describe below provides even stronger guarantees. 

\paragraph{Coordinate methods with variance reduction.}
Using our recently proposed 
framework~\citep{CarmonJST19} we design a 
variance reduction algorithm with amortized per-iteration cost $\Otil{1}$, 
required iteration count of $\Otil{\sqrt{\nnz}\cdot \frac{\Lco}{\epsilon}}$ 
and total running time $\Otil{\nnz + \sqrt{\nnz}\cdot 
	\frac{\Lco}{\epsilon}}$. 
In the numerically sparse regime $\Lco = O(\Lrc)$, 
our runtime improves on row-column VR by a factor of $\sqrt{\nnz/ 
(m+n)}$, 
and in general the  
bound $\Lco \le \Lrc\sqrt{m+n}$ guarantees it is never worse. 
Since variance reduction methods always require a single pass over the data 
to compute an exact gradient, this comparison holds regardless of the 
matrix access model. 
In the $\elltwo$ setting we note that for \emph{elementwise 
non-negative} matrices, $\Lco = 
\max\{\ltwos{A\ones},\ltwos{A^\top\ones}\} 
\le \Lmv\sqrt{m+n}$, and consequently our 
method outperforms exact gradient methods by a factor of 
$\sqrt{\nnz/(m+n)}$, even without any numerical or spectral sparsity in 
$A$. Notably, this is the same factor of improvement that row-column VR 
achieves over exact gradient methods in the $\ellone$ and $\elltwoone$ 
regimes.

\paragraph{Optimality of the constant $\Lco$.}
For the $\ellone$ and $\elltwo$ settings, we argue that the constant
$\Lco$ in Table~\ref{table:L} is optimal in the restricted sense that no 
alternative sampling distribution for coordinate gradient estimation can 
have a better variance bound than $\Lco$ (a similar sense of optimality also 
holds for $\Lrc$ in each  
geometry). 
In the $\elltwoone$ setting, a different sampling distribution produces an 
improved (and optimal) constant $\max\{ \max_i \lone{\ai}, 
\opnorm{|A|}\}$, where $\ai$ is the $i$th row of $A$, and $|A|_{ij} = |A_{ij}|$ is the elementwise absolute value of $A$.
However, it is unclear how to efficiently 
sample from this distribution.

\paragraph{Row-column sparse instances.} 
Some problem instances admit a structured form of sparsity where every 
row and column has at most $\rcs$ nonzero elements. In all settings we have $\Lco \le 
\Lrc\sqrt{\rcs}$ and so our coordinate methods naturally improve 
when $\rcs$ is small. Specifically, the sampling distributions and data 
structures we develop in this paper allow us to modify previous methods 
for 
row-column VR \citep{CarmonJST19} to leverage row-column sparsity, 
reducing 
the amortized per-iteration cost from $\O{m+n}$ to $\Otil{\rcs}$.

\paragraph{Applications.}
We illustrate the implications of our results for two problems in 
computational geometry, minimum enclosing ball (Min-EB) and maximum 
inscribed ball (Max-IB), as well as linear regression. For Min-EB and 
Max-IB in the non-degenerate case $m\ge n$, we apply our $\elltwoone$ results to obtain algorithms 
whose runtime bounds coincide with the 
state-of-the-art~\cite{ZhuLY16} for dense problems, but can be 
significantly better for sparse or row-column sparse instances. For linear 
regression we focus on accelerated linearly converging algorithms, i.e., 
those that find $x$ such that $\ltwo{Ax-b}\le \epsilon$ in time 
proportional  to $\mu^{-\half}\log 
\frac{1}{\epsilon}$ where $\mu$ is the smallest eigenvalue of $A^\top A$. 
Within this class and in a number of settings, our reduced variance 
coordinate method offers  
improvement over the state-of-the-art: for instances where 
$\lone{\ai}=O(\ltwo{\ai})$ and $\lone{\aj}=O(\ltwo{\aj})$ for all $i,j$ it 
outperforms~\cite{GuptaS18} by a factor of $\nnz^{1/6}$, and for 
elementwise nonnegative instances it outperforms 
accelerated 
gradient descent by a factor of $\sqrt{\nnz/(m+n)}$. See 
Table~\ref{table:applications} for a detailed runtime comparison. 

\subsection{Our approach}
\label{ssec:intro_approach}

We now provide a detailed overview of our algorithm design and analysis 
techniques, highlighting our main technical insights. We focus on the 
$\ellone$ geometry, since it showcases all of our developments. Our 
technical contributions have two central themes:

\begin{enumerate}[leftmargin=\parindent]
	\item \textbf{Sampling schemes design.} The key to obtaining efficient 
	coordinate methods is carefully choosing the sampling distribution. Here,
	local norms analysis of stochastic mirror descent~\citep{Shalev-Shwartz12} on the one hand enables tight regret bounds, and 
	on the other hand imposes an additional design constraint since the 
	stochastic estimators must be bounded for the analysis to apply. We achieve estimators with improved variance bounds meeting this boundedness constraint by leveraging a ``clipping'' operation introduced by \citet{ClarksonHW10}. Specifically, in the simplex geometry, we truncate large coordinates of our estimators, and show that our method is robust to the resulting distortion. 
	\item \textbf{Data structure design.} Our goal is to perform iterations in 
	$\Otil{1}$ time, but our mirror descent procedures call for updates that 
	change $m + n$ variables in each step. We resolve this tension via data 
	structures that implicitly maintain the iterates. Variance reduction poses 
	a considerable challenge here, because every reduced-variance 
	stochastic gradient contains a dense component that changes all 
	coordinates in a complicated way. In particular, existing data structures 
	cannot efficiently compute the normalization factor necessary for 
	projection to the simplex. 
	We design a data structure that overcomes this hurdle via Taylor 
	expansions, coordinate binning, and a binomial heap-like construction. 
	The data structure computes approximate mirror projections, and we 
	modify the standard mirror descent analysis to show it is stable under the 
	particular structure of the resulting approximation errors.
\end{enumerate}
At the intersection of these two themes is a novel sampling technique we 
call ``sampling from the sum,'' which addresses the same variance 
challenges as the ``sampling from the difference'' technique 
of~\cite{CarmonJST19}, but is more amenable to efficient 
implementation with a data structure. 

\subsubsection{Coordinate stochastic gradient method}

Our algorithm is an instance of stochastic mirror 
descent~\cite{NemirovskiJLS09}, which in 
the $\ellone$ setting produces a sequence of iterates $(x_1, y_1), (x_2, 
y_2), \ldots$ according to
\arxiv{
\begin{equation}\label{eq:l1-smd}
x_{t+1} = \Pi_{\Delta}\left(x_t \circ \exp\{-\eta \tg\x(x_t, y_t)\}\right)
~~\mbox{and}~~
y_{t+1} = \Pi_{\Delta}\left(y_t \circ \exp\{ - \eta \tg\y(x_t, y_t)\}\right),
\end{equation}
}
\notarxiv{
\begin{equation}\label{eq:l1-smd}
\begin{aligned}
x_{t+1} & = \Pi_{\Delta}\left(x_t \circ \exp\{-\eta \tg\x(x_t, y_t)\}\right)\\
~~\mbox{and}~~
y_{t+1} & = \Pi_{\Delta}\left(y_t \circ \exp\{ - \eta \tg\y(x_t, y_t)\}\right),
\end{aligned}
\end{equation}
}
where $\Pi_\Delta(v) = \frac{v}{\lone{v}}$ is the projection onto the simplex ($\exp$ and 
$\log$ are applied to vectors elementwise, and elementwise multiplication is denoted by
$\circ$), $\eta$ is a step size, and $\tg\x, \tg\y$ 
are stochastic gradient estimators for $f(x,y)=y^\top A x$ satisfying
\arxiv{
\begin{equation*}
\E \tg\x(x,y) = \grad_x f(x,y)=A^\top y
~~\mbox{and}~~
\E \tg\y(x,y) = -\grad_y f(x,y)=-A x.
\end{equation*}}
\notarxiv{
\begin{equation*}
\begin{aligned}
\E \tg\x(x,y) & = \grad_x f(x,y)=A^\top y\\
~~\mbox{and}~~
\E \tg\y(x,y) & = -\grad_y f(x,y)=-A x.
\end{aligned}
\end{equation*}}

We describe the computation and analysis of $\tg\x$;  the treatment of 
$\tg\y$ is analogous. To compute $\tg\x(x,y)$, we 
sample $i,j$ from a distribution $p(x,y)$ on $[m]\times[n]$ and let
\begin{equation}\label{eq:l1-tg}
\tg \x(x,y) = \frac{y_i A_{ij}}{p_{ij}(x,y)} e_j,
\end{equation}
where $p_{ij}(x,y)$ denotes the probability of drawing $i,j$ from $p(x,y)$  
and $e_j$ is the $j$th standard basis vector---a simple calculation shows 
that $\E\tg\x = A^\top y$ for any $p$. We first design $p(x,y)$ to 
guarantee an   
$\Otil{(\frac{\Lco}{\epsilon})^2}$ iteration complexity for finding an 
$\epsilon$-accurate solution, and then briefly touch on how to compute 
the resulting iterations in $\Otil{1}$ time.

\paragraph{Local norms-informed distribution design.}
The standard stochastic mirror descent 
analysis~\cite{NemirovskiJLS09} shows that if $\E \linf{\tg\x(x,y)}^2 
\le L^2$ for all $x,y$ (and similarly for $\tg\y$), taking $\eta = 
\frac{\epsilon}{L^2}$ and a choice of $T=\Otil{(\frac{L}{\epsilon})^2}$ 
suffices to ensure that the iterate average $\frac{1}{T} \sum_{t=1}^T (x_t, y_t)$ is 
an $\epsilon$-accurate solution in expectation. Unfortunately, this 
analysis demonstrably fails to yield sufficiently tight bounds for our 
coordinate estimator: there exist instances for which any distribution $p$ produces 
$L\ge n\Lrc$. We tighten the analysis using a \emph{local norms} 
argument~\cite[cf.][Section 2.8]{Shalev-Shwartz12}, showing that 
$\Otil{(\frac{L}{\epsilon})^2}$ iterations suffice whenever $\linf{\eta \tg\x} \le 1$ with probability 1 and for all $x,y$
\begin{equation*}
\E \norm{\tg\x(x,y)}^2_x \le L^2,
~~\mbox{where}~~\norm{\gamma}_x^2 = \sum_j x_j  \gamma_j^2 
\end{equation*}
is the local norm at $x\in\xset$. We take 
\begin{equation}\label{eq:l1-p}
p_{ij} = y_i 
\frac{A_{ij}^2}{\ltwo{\ai}^2}
\end{equation}
 (recalling that $x,y$ are both probability 
vectors). Substituting 
into~\eqref{eq:l1-tg} gives
\arxiv{
\begin{equation*}
\E \norm{\tg\x(x,y)}^2_x =  \sum_{i,j} \frac{y_i^2 A_{ij}^2 x_j}{p_{ij}}
= \sum_{i,j} {y_i \ltwo{\ai}^2 x_j } =  \sum_{i} y_i\ltwo{\ai}^2 \le \max_i 
\ltwo{\ai}^2 \le \Lco^2,
\end{equation*}}
\notarxiv{
\begin{equation*}
\begin{aligned}
\E \norm{\tg\x(x,y)}^2_x & =  \sum_{i,j} \frac{y_i^2 A_{ij}^2 x_j}{p_{ij}}
= \sum_{i,j} {y_i \ltwo{\ai}^2 x_j } \\
& =  \sum_{i} y_i\ltwo{\ai}^2 \le \max_i 
\ltwo{\ai}^2 \le \Lco^2,
\end{aligned}
\end{equation*}
}
with $\Lco = \max\{ \max_i \ltwo{\ai}, \max_j \ltwo{\aj}\}$ as in 
Table~\ref{table:L}.

While this is the 
desired bound on $\E \norm{\tg\x(x,y)}^2_x$, the requirement  
$\linf{\eta \tg\x} \le 1$ does not hold when $A$ has 
sufficiently small elements. We address this by clipping $\tg$: we replace $\eta \tg\x$ with $\clip(\eta \tx\x)$, where
\begin{equation*}\label{eq:clipping}
[\clip(v)]_i \defeq \min\{ |v_i|, 1\} \sign(v_i),
\end{equation*}
the Euclidean projection to the unit box. The clipped gradient estimator clearly satisfies the desired bounds on infinity norm and local norm second moment, but is biased for the true gradient. Following the analysis of~\citet{ClarksonHW10}, we account for the bias by relating it to the second moment via
\[|\inner{\gamma - \text{clip}(\gamma)}{x}| \le \norm{\gamma}_x^2,\]
which allows to absorb the effect of the bias into existing terms in our error bounds. Putting together these pieces yields the desired bound on the iteration count.

\paragraph{Efficient implementation.}
Data structures for performing the update~\eqref{eq:l1-smd} and 
sampling from the resulting iterates in $\Otil{1}$ time are standard in the 
literature \cite[e.g.,][]{ShalevW16}.
We add to these the somewhat non-standard ability to also efficiently track 
the running sum of the iterates. To efficiently sample $i,j\sim p$ according 
to~\eqref{eq:l1-p} we first use the data structure to sample $i\sim y$ in 
$\Otil{1}$ time and then draw $j\in[n]$ with 
probability proportional to $A_{ij}^2$ in time $O(1)$, via either
$O(\nnz)$ preprocessing or an appropriate assumption about the 
matrix access model. The ``heavy lifting'' of our data structure design is dedicated for supporting variance reduction, which we describe in the next section.

\notarxiv{
\begin{table*}[!ht]
	\centering
	\renewcommand{\arraystretch}{2}
	\bgroup
			\everymath{\displaystyle}
			\caption{The distributions used in our coordinate gradient estimator}\label{table:dist-sublinear}
			\begin{tabular}{ccc }
				\toprule
				{Setting} & { $p_{ij}$  } & { $q_{ij}$ } \\
				\midrule
				{\bf $\ell_1$-$\ell_1$} & $ 
				y_i \cdot \frac{A_{ij}^2}{\norm{\ai}_2^2}$ & $ 
				x_j\cdot\frac{A_{ij}^2}{\norm{\aj}_2^2}$ 
				\\[12pt]
				{\bf $\ell_2$-$\ell_1$} & $  
				y_i \cdot \frac{|A_{ij}|}{\lones{\ai}}$ & $  
				\frac{A_{ij}^2}{\norm{A}_{\mathrm{F}}^2}$ 
				\\[12pt]
				{\bf $\ell_2$-$\ell_1$} & $  
				y_i \cdot \frac{|A_{ij}|}{\lones{\ai}}$ & $  
				\propto x_j^2 \cdot \1_{A_{ij}\neq 0}$
				\\[12pt]
				{\bf $\ell_2$-$\ell_1$} & $  
				y_i \cdot \frac{|A_{ij}|}{\lones{\ai}}$ & $  
				\frac{A_{ij} \cdot x_j^2}{\sum_{k \in [n]} \norm{A_{:k}}_1\cdot  x_k^2}$
				\\[12pt]
				{\bf $\ell_2$-$\ell_2$} & $  
				\frac{\lones{\ai}^2}{\sum_{k\in[m]}\lones{\ai[k]}^2}\cdot\frac{|A_{ij}|}{\lones{\ai}}$
				& $  
				\frac{\lones{\aj}^2}{\sum_{k\in[n]}\lones{\aj[k]}^2}\cdot\frac{|A_{ij}|}{\lones{\aj}}$
				\\[12pt]
				{\bf $\ell_2$-$\ell_2$} & $  
				\frac{{ y_i}^2}{\ltwo{y}^2}\cdot\frac{|A_{ij}|}{\lones{\ai}}$ & $  
				\frac{{x_j}^2}{\ltwo{x}^2}\cdot\frac{|A_{ij}|}{\lones{\aj}}$  \\
				\bottomrule
			\end{tabular}
\egroup
		
	\vspace{2ex}
	\parbox[left]{0.8\textwidth}{\small Recall that the estimator is of the form
		$ {\tg(x,y) = \big(\frac{1}{p_{ij}} {y_{i} A_{i j}}\cdot e_{j}, - 
			\frac{1}{q_{lk}}{A_{lk}x_{k}}\cdot e_{l}\big)}$ where $i,j\sim 
		p$ and 
		$l,k\sim q$. }

\vspace{32pt}
	\centering
	\renewcommand{\arraystretch}{2}
	\bgroup
	\everymath{\displaystyle}
	\caption{The distributions $p,q$ used for  our reduced variance 
		coordinate 
		gradient 
		estimator}\label{table:dist-vr}
	\begin{tabular}{ccc}
		\toprule
		{Setting} & { $p_{ij}$  } & { $q_{ij}$ } \\
		\midrule
		{\bf $\ell_1$-$\ell_1$} & $  
		\frac{y_i+2[y_0]_i}{3}\cdot\frac{A_{ij}^2}{\ltwo{\ai}^2}$
		& $  \frac{x_j+2[x_0]_j}{3} 
		\cdot\frac{A_{ij}^2}{\ltwo{\aj}^2}$ \\[12pt]
		{\bf $\ell_2$-$\ell_1$} & $  
		\frac{y_i+2[y_0]_i}{3}\cdot\frac{|A_{ij}|}{\lones{\ai}}$ & 
		$  \frac{A_{ij}^2}{\norm{A}_\mathrm{F}^2}$
		\\[12pt]
		{\bf $\ell_2$-$\ell_1$} & $  
		\frac{y_i+2[y_0]_i}{3}\cdot\frac{|A_{ij}|}{\lones{\ai}}$ & 
		$  \propto [x - x_0]_j^2 \cdot \1_{A_{ij} \neq 0}$
		\\[12pt]
		{\bf $\ell_2$-$\ell_1$} & $  
		\frac{y_i+2[y_0]_i}{3}\cdot\frac{|A_{ij}|}{\lones{\ai}}$ & 
		$ \frac{|A_{ij}| \cdot [x - x_0]_j^2}{\sum_{k \in [n]} \norm{A_{:k}}_1 \cdot [x - x_0]_k^2}$
		\\[12pt]
		{\bf $\ell_2$-$\ell_2$} & $ 
		\frac{\lones{\ai}^2}{\sum_{k\in[m]}\lones{\ai[k]}^2}\cdot\frac{|A_{ij}|}{\lones{\ai}}$
		& $
		\frac{\lones{\aj}^2}{\sum_{k\in[n]}\lones{\aj[k]}^2}\cdot\frac{|A_{ij}|}{\lones{\aj}}$
		\\[12pt]
		{\bf $\ell_2$-$\ell_2$} & $ 
		\frac{[y - y_0]_i^2}{\norm{y - y_0}_2^2} \cdot\frac{|A_{ij}|}{\lones{\ai}}$
		& $
		\frac{[x - x_0]_j^2}{\norm{x - x_0}_2^2}\cdot\frac{|A_{ij}|}{\lones{\aj}}$
		\\[12pt]
		\bottomrule
	\end{tabular}
	
	\vspace{2ex}
	\parbox[left]{0.8\textwidth}{\small Recall that the estimator is of the form
		$ {\tg(x,y) = \big(A^\top y +\frac{1}{p_{ij}} {(y_{i}-y_{0,i}) A_{i 
					j}} \cdot
			e_{j}, - Ax - 
			\frac{1}{q_{lk}}{A_{lk}(x_{k}-x_{0,k})} \cdot e_{l}\big)}$ 
		where 
		$i,j\sim p$ 
		and $l,k\sim q$ and $x_0, y_0$ is a reference point.}
	
	\egroup
\end{table*}
 }

\paragraph{Sampling distributions beyond $\ellone$.} Table~\ref{table:dist-sublinear} lists the sampling 
distributions we develop for the various problem geometries. Note that for 
the $\elltwoone$ setting we give three different distributions for 
sampling the simplex block of the gradient (i.e., $\tg\y$); each distribution 
corresponds to a different parameter $\Lco$ (see comments following  
Table~\ref{table:L}). The distribution $q_{ij} \propto \sqrt{y_i} \,|A_{ij} x_j|$ 
 yields a stronger bound $L$ in the $\elltwoone$ setting, but we do 
not know how to efficiently sample from it.

\subsubsection{Coordinate variance reduction}

\newcommand{\rp}[1][]{_{0#1}}

To accelerate the stochastic coordinate method we apply our recently 
proposed variance reduction 
framework~\cite{CarmonJST19}. This framework operates in 
$\frac{\alpha}{\epsilon}$ epochs, where $\alpha$ is a design parameter 
that trades between full and stochastic gradient computations. Each epoch 
consists of three parts: (i) computing the exact gradient at a reference point 
$(x\rp,y\rp)$, (ii) performing 
$T$ iterations of regularized stochastic 
mirror descent to produce the sequence $(x_1,y_1),\ldots,(x_T, y_T)$ and 
(iii) taking an extra-gradient step from the average of the iterates in (ii). Setting 
$\kappa = 1/(1 + \eta \alpha/2)$, the iterates $x_t$ follow the recursion
\arxiv{
\begin{equation}\label{eq:l1-reg-smd}
x_{t+1} = \Pi_{\Delta}\big(x_t^{\kappa} \circ  x_0^{1-\kappa}
\circ \exp\{ - \eta \kappa [ g\x\rp + 
\td\x(x_t,y_t)]\}\big),~\mbox{where}~\Pi_\Delta(v) = \frac{v}{\lone{v}},
\end{equation}}
\notarxiv
{
\begin{flalign}
& x_{t+1}  = \Pi_{\Delta}\big(x_t^{\kappa} \circ  x_0^{1-\kappa}
\circ \exp\{ - \eta \kappa [ g\x\rp + 
\td\x(x_t,y_t)]\}\big), \nonumber\\
&\quad\mbox{where}~\Pi_\Delta(v) = \frac{v}{\lone{v}},~\mbox{and}~
g\rp\x = A^\top y\rp
\label{eq:l1-reg-smd}
\end{flalign}
}
\arxiv{and $g\rp\x = A^\top y\rp$ }is the exact gradient at the reference point, and  
$\td\x$ is a stochastic gradient \emph{difference} estimator satisfying
\begin{equation*}
\E \td\x(x,y) = \grad_x f(x,y) - \grad_x f(x\rp, y\rp) = A^\top (y-y\rp).
\end{equation*}
The iteration for $y_t$ is similar. In~\cite{CarmonJST19} we show that if 
$\td\x$ satisfies
\begin{equation}\label{eq:l1-centered-basic}
\E \linfs{\td\x(x,y)}^2 \le L^2 \left( \lone{x-x\rp}^2 + 
\lone{y-y\rp}^2\right)~~\forall x,y
\end{equation}
and a similar bound holds on $\E \linfs{\td\y(x,y)}^2$, then 
$T=O(\frac{L^2}{\alpha^2})$ iterations per epoch with step size 
$\eta = \frac{\alpha}{L^2}$ suffice for the overall 
algorithm to return a point with expected error below $\epsilon$.

We would like to design a coordinate-based estimator $\td$ such that the 
bound~\eqref{eq:l1-centered-basic} holds for $L=\Lco$  as in 
Table~\ref{table:L} and 
each iteration~\eqref{eq:l1-reg-smd} takes $\Otil{1}$ time. Since every 
epoch also requires $O(\nnz)$ time for matrix-vector product (exact 
gradient) computations, the overall runtime would be $\Otils{(\nnz + 
\frac{\Lco^2}{\alpha^2})\cdot \frac{\alpha}{\epsilon}}$. Choosing $\alpha 
= 
\Lco/\sqrt{\nnz}$ then gives the desired runtime $\Otils{\nnz + 
\sqrt{\nnz}\cdot \frac{\Lco}{\epsilon}}$.  

\paragraph{Distribution design (sampling from the difference).}
We start with a straightforward adaptation of the general estimator 
form~\eqref{eq:l1-tg}.  To compute $\td\x(x,y)$, we 
sample $i,j\sim p$, where $p$  may depend on 
$x,x\rp,y$ and $y\rp$, and let
\begin{equation}\label{eq:l1-td}
\td \x(x,y) = \frac{(y_i -[y_0]_i) A_{ij}}{p_{ij}} e_j,
\end{equation}
where $e_j$ is the $j$th standard basis vector. As in the previous section, 
we 
find that the requirement~\eqref{eq:l1-centered-basic} is too stringent for 
coordinate-based estimators. Here too, we address this challenge with a 
local norms argument and clipping of the difference estimate. Using the ``sampling from the 
difference'' technique from~\cite{CarmonJST19}, we arrive at
\begin{equation}\label{eq:l1-sampling-difference}
p_{ij} = \frac{\left| y_i - [y_0]_i\right|}{\lone{y-y\rp}} \cdot 
\frac{A_{ij}^2 }{\ltwo{\ai}^2}.
\end{equation}
This distribution satisfies the local norm relaxation 
of~\eqref{eq:l1-centered-basic} with $L^2 = \Lco^2$.

\paragraph{Data structure design.}
Efficiently computing~\eqref{eq:l1-reg-smd} is significantly more 
challenging than its counterpart~\eqref{eq:l1-smd}. To clarify the difficulty 
and describe our solution, we write 
\begin{equation*}
x_t = \Pi_\Delta(\hat{x}_t) = \hat{x}_t / \lone{\hat{x}_t}
\end{equation*}
and break the recursion for the unnormalized iterates
$\hat{x}_t$  into two steps
\begin{flalign}
\hat{x}_{t}' &= \hat{x}_t ^{\kappa} \circ \exp\{v\},~\mbox{and}
\label{eq:l1smd-dense}\\
\hat{x}_{t+1} &= \hat{x}_{t}' \circ \exp\{s_t\},\label{eq:l1smd-sparse}
\end{flalign}
where $v=(1-\kappa)\log x\rp - \eta\kappa g\x\rp$ is a fixed dense 
vector, and $s_t = -\eta \td\x(x_t,y_t)$ is a varying 1-sparse vector. The key task of the data structures is maintaining 
the normalization factor $\lone{\hat{x}_t}$ in near-constant time. 
Standard data structures do not suffice because they lack support for the dense step~\eqref{eq:l1smd-dense}. 

Our high-level strategy is to handle the two steps~\eqref{eq:l1smd-dense} 
and~\eqref{eq:l1smd-sparse} separately. To handle the dense 
step~\eqref{eq:l1smd-dense}, we propose the data structure $\ScM$ 
that efficiently approximates $\lone{\hat{x}_t}$ in the ``homogeneous'' 
case of no sparse updates (i.e.\ $s_t=0$ for all $t$). We then add support for  the 
sparse step~\eqref{eq:l1smd-sparse} using a binomial heap-like 
construction involving $O(\log n)$ instances of $\ScM$.

\renewcommand{\hom}{^\mathrm{hom}}
\newcommand{\bv}{\bar{v}}
\newcommand{\bn}{\bar{n}}
\paragraph{The $\ScM$ data structure.}
When $s_t=0$ for all $t$ the iterates $\hat{x}_t$ admit closed forms
\arxiv{
\begin{equation*}
\hat{x}_{t+\tau} = \hat{x}_t^{\kappa^\tau} \circ \exp\left\{ 
v\sum_{t'=0}^{\tau-1} \kappa^{t'}\right\}
=\hat{x}_t^{\kappa^\tau} \circ \exp\Big\{ \frac{1-\kappa^{\tau}}{1-\kappa} 
v\Big\}
=\hat{x}_t\circ \exp\left\{[1-\kappa^{\tau}]\bv\right\},
\end{equation*}}
\notarxiv{
\begin{equation*}
\begin{aligned}
\hat{x}_{t+\tau} & = \hat{x}_t^{\kappa^\tau} \circ \exp\left\{ 
v\sum_{t'=0}^{\tau-1} \kappa^{t'}\right\}
=\hat{x}_t^{\kappa^\tau} \circ \exp\Big\{ \frac{1-\kappa^{\tau}}{1-\kappa} 
v\Big\}\\
& =\hat{x}_t\circ \exp\left\{[1-\kappa^{\tau}]\bv\right\},
\end{aligned}
\end{equation*}
}
where $\bv = \frac{v}{1-\kappa} - \log x_t$. %
Consequently, we design $\ScM$ to take as initialization $\bn$-dimensional 
vectors $\bx \in \R^{\bn}_{\ge 0}$, and $\bv\in\R^{\bn}$ and provide  approximations of the normalization factor
\begin{equation}\label{eq:exact-partition-function}
Z_\tau (\bx, \bv) = \lone{\bx \circ \exp\{ (1-\kappa^\tau) \bv\})}
\end{equation}
for arbitrary values of $\tau\ge1$. We show how to implement each 
query 
of $Z_\tau (\bx, \bv)$ in amortized time $\Otil{1}$. 
The data 
structure also supports initialization in time $\Otil{\bn}$ and deletions (i.e., 
setting elements of $\bx$ to zero) in amortized time 
$\Otil{1}$.

To efficiently approximate the quantity $Z_\tau (\bx, \bv)$ we replace the 
exponential with its order $p=O(\log n)$  Taylor series. That is, 
we would like to write
\arxiv{
\begin{equation*}
Z_\tau (\bx, \bv) =
\sum_{i\in[\bn]} [\bx]_i e^{(1-\kappa^\tau) [\bv]_i}
\approx \sum_{i\in[\bn]} [\bx]_i \sum_{q=0}^p
\frac{1}{q!} (1-\kappa^\tau)^q [\bv]_i^q
= \sum_{q=0}^p \frac{(1-\kappa^\tau)^q}{q!}
\inner{\bx}{\bv^q}.
\end{equation*}}
\notarxiv{
\begin{equation*}
\begin{aligned}
Z_\tau (\bx, \bv) & =
\sum_{i\in[\bn]} [\bx]_i e^{(1-\kappa^\tau) [\bv]_i}
\\
& \approx \sum_{i\in[\bn]} [\bx]_i \sum_{q=0}^p
\frac{1}{q!} (1-\kappa^\tau)^q [\bv]_i^q\\
& = \sum_{q=0}^p \frac{(1-\kappa^\tau)^q}{q!}
\inner{\bx}{\bv^q}.
\end{aligned}
\end{equation*}
}
The approximation $\sum_{q=0}^p \frac{(1-\kappa^\tau)^q}{q!}
\inner{\bx}{\bv^q}$ is cheap to compute, since for every $\tau$ it is a linear 
combination of the $p=\Otil{1}$ numbers $\{\inner{\bx}{\bv^q}\}_{q\in [p]}$ which we 
can compute once at initialization. However, the Taylor series approximation 
has low multiplicative error only when $|(1-\kappa^\tau) [\bv]_i| = O(p)$, 
which may fail to hold, as we may have $\linf{\bv}=\mathsf{poly}(n)$ in general. 
To handle this, suppose that for a fixed $\tau$ we have an offset 
$\mu\in\R$ and ``active set'' $A\subseteq [\bn]$ such that the following conditions hold for a 
threshold 
$R=O(p)$: (a) the Taylor approximation is valid in $A$, e.g.\ we have  
$|(1-\kappa^\tau)(\bv_i - \mu)|\le 2R$ for all $i\in A$, (b) entries outside 
$A$ are small;  
$(1-\kappa^\tau)[\bv_i - \mu] \le -{R}$ for all 
$i\notin A$, and (c) at least one entry in the active set is large; 
$(1-\kappa^\tau)[\bv_i - \mu] \ge 0$ for some  
$i\in A$. Under these conditions, the entries in $A^c$ are negligibly small 
and we can truncate them, resulting in the approximation
\begin{equation*}
e^{(1-\kappa^\tau)\mu}\left[\sum_{q=0}^p \frac{(1-\kappa^\tau)^q}{q!}
\inner{\bx}{(\bv-\mu)^q}_{A} + {e^{-R}\inner{\bx}\ones}_{A^c}\right],
\end{equation*}
which we show approximates $Z_\tau (\bx, \bv)$ to within $e^{O(R + \log 
n) - \Omega(p)}$ 
multiplicative error, where we used $\inner{a}{b}_S \defeq
\sum_{i\in S} a_i b_i$; here, we also 
require that $ \log\frac{\max_i \bx_i}{\min_i \bx_i} = O(R)$, which we guarantee when choosing the initial $\bx$.

The challenge then becomes efficiently mapping any $\tau$ to 
$\{\inner{\bx}{(\bv-\mu)^q}_{A}\}_{q\in[p]}$ for suitable $\mu$ and $A$. 
We address this by jointly bucketing $\tau$ and $\bv$. Specifically, we map 
$\tau$ into a bucket index $k=\floor{\log_2 
\frac{1-\kappa^\tau}{1-\kappa}}$, pick $\mu$ to be the largest integer 
multiple of $R/((1-\kappa)2^k)$ such that $\mu \le \max_i \bv_i$, and set 
$A=\{i\mid |(1-\kappa)2^k (\bv_i - \mu)| \le R\}$. Since $k \le k_{\max} = 
\floor{\log_2 \frac{1}{1-\kappa}}=O(\log n)$, we argue that computing 
$\inner{\bx}{(\bv-\mu)^q}_{A}$ for every possible resulting $\mu$ and $A$ 
takes at most $O(\bn p \log{\frac{1}{1-\kappa}})=\Otil{\bn}$ time, which we can 
charge to initialization. We further show how to support deletions in 
$\Otil{1}$ time by carefully manipulating the computed quantities.%

\paragraph{Supporting sparse updates.}
Building on $\ScM$, we design the data structure\linebreak $\AEM$ that (approximately)
implements 
the entire mirror descent step~\eqref{eq:l1-reg-smd} in time 
$\Otil{1}$.\footnote{The data structures $\AEM$ and $\ScM$ structure 
support two additional operations 
necessary for our algorithm: efficient approximate sampling from 
$x_t$ and maintenance of a running sum of  
$\hat{x}_\tau$. Given the normalization constant approximation, the 
implementation of these 
operations is fairly 
straightforward, so we do not discuss them in the introduction.}
The data structure maintains vectors $\bx \in \Delta^n$ and $\bv\in\R^n$ 
and $K=\ceil{\log_2 (n+1)}$ instances of $\ScM$ denoted 
$\{\ScM_k\}_{k\in[K]}$. 
The $k$th instance tracks a coordinate subset $S_k \subseteq [n]$ 
such that $\{S_k\}_{k\in[K]}$ partitions $[n]$, and has initial data 
$[\bx]_{S_k}$ and $[\bv]_{S_k}$. We let $\tau_k\ge 0$ denote the ``time 
index'' parameter of the $k$th instance. The data structure satisfies two 
invariants; first, the unnormalized iterate $\hat{x}$ satisfies
\begin{equation}\label{eq:bheap-invariant-1}
[\hat{x}]_{S_k} = \left[
\bx \circ \exp\{ (1-\kappa^{\tau_k})\bv \} \right]_{S_k},~\mbox{for all}~ 
k\in[K].
\end{equation}
Second, the partition satisfies
\begin{equation}\label{eq:bheap-invariant-2}
|S_k| \le 2^{k}-1~\mbox{for all}~k\in[K],
\end{equation}
where at initialization we let $S_K=[n]$ and $S_k = \emptyset$ for $k<K$, 
$\bx=x_0$, $\bv=\frac{v}{1-\kappa} - \log x_0$ and $\tau_K = 0$.

The invariant~\eqref{eq:bheap-invariant-1} allows us to efficiently (in time 
$\Otil{K}=\Otil{1}$) query 
coordinates of $x_t= \hat{x}_t/\lone{\hat{x}_t}$, since $\ScM$ allows us to 
approximate $\lone{\hat{x}_t} = \sum_{k\in[K]} 
Z_{\tau_k}([\bx]_{S_k}, [\bv]_{S_k})$ with $Z$ as defined 
in~\eqref{eq:exact-partition-function}. To implement the dense 
step~\eqref{eq:l1smd-dense}, we 
simply increment $\tau_k\gets \tau_k+1$ for every $k$. 
Let $j$ be the nonzero coordinate of $s_t$ in the sparse 
step~\eqref{eq:l1smd-sparse}, and let $k\in[K]$ be such that 
$j\in S_k$. To implement~\eqref{eq:l1smd-sparse},  we delete 
 coordinate $j$ from $\ScM_k$, and create a singleton instance $\ScM_0$ 
maintaining $S_0=\{j\}$ with initial data $[\bx]_{S_0} =e^{s_t} \hat{x}_j$, 
$[\bv]_{S_0}=v_j/(1-\kappa) - \log( e^{s_t} \hat{x}_j )$ and $\tau_0=0$. 
Going from $k=1$ to $k=K$, we merge $\ScM_{k-1}$ into $\ScM_{k}$ 
until the invariant~\eqref{eq:bheap-invariant-2} holds again. For example, 
if before the sparse step we have $|S_1|=1$, $|S_2|=3$ and $|S_3|=2$, we 
will perform 3 consecutive merges, so that afterwards we have 
$|S_1|=|S_2|=0$ and $|S_3|=7$. 

To merge two $\ScM_{k-1}$ into $\ScM_{k}$, we let $S'_k=S_{k-1}\cup S_k$ 
and initialize a new $\ScM$ instance with 
$[\bx]_{S'_k}=[\hat{x}]_{S'_k}$,\footnote{More precisely, for every $j\in S'_k$ 
we set $\bx_j= \hat{x}_j + \varepsilon \max_{i\in S_k'}\hat{x}_i$, where 
$\varepsilon$ is a small padding constant that ensures the bounded multiplicative 
range necessary for correct operation of $\ScM$.}
$[\bv]_{S'_k} = [v]_{S'_k}/(1-\kappa) - \log [\hat{x}]_{S'_k}$ and 
$\tau_k=0$;  
this takes $\Otil{|S_k'|}=\Otil{2^k}$ time due to the 
invariant~\eqref{eq:bheap-invariant-2}. Noting that a merge at level $k$ 
can only happen once in every $\Omega(2^k)$ updates, we conclude that 
the amortized cost of merges at each level is $\Otil{1}$, and (since 
$K=\Otil{1}$), so is the cost of the sparse update.

\paragraph{Back to distribution design (sampling from the sum).}
Our data structure enables us to compute the 
iteration~\eqref{eq:l1-reg-smd} and query coordinates of the iterates $x_t$ 
and $y_t$ in  $\Otil{1}$ amortized time. However, we cannot compute 
$\td\x$ using the distribution~\eqref{eq:l1-sampling-difference} because 
we do not have an efficient way of sampling from $| y_t - y_0 |$; Taylor 
approximation techniques are not effective for approximating the absolute 
value because it is not smooth. To overcome this final barrier, we introduce 
a new design which we call ``sampling from the sum,''
\begin{equation}\label{eq:l1-sampling-sum}
p_{ij}(x,y) = \left(\frac{1}{3} y_i + \frac{2}{3} [y_0]_i \right)\cdot 
\frac{A_{ij}^2}{\ltwo{\ai}^2}.
\end{equation}
Sampling from the modified distribution is simple, as our data structure  
allows us to sample from $y_t$. Moreover, we show that the 
distribution~\eqref{eq:l1-sampling-sum} satisfies a relaxed version 
of~\eqref{eq:l1-centered-basic} where the LHS is replaced by a local norm 
as before, and the RHS is replaced by $L^2 (V_{x_0}(x_t) + V_{y_0}(y_t))$, 
where $V_x(x')$ is the KL divergence between $x$ and $x'$. In 
Table~\ref{table:dist-vr} we list the sampling distributions we design for 
variance reduction in the different domain geometries.

\arxiv{
\begin{table}
	\centering
	\renewcommand{\arraystretch}{2}
	\bgroup
			\everymath{\displaystyle}
	\begin{tabular}{ccc }
		\toprule
		{Setting} & { $p_{ij}$  } & { $q_{ij}$ } \\
		\midrule
		{\bf $\ell_1$-$\ell_1$} & $ 
		y_i \cdot \frac{A_{ij}^2}{\norm{\ai}_2^2}$ & $ 
		x_j\cdot\frac{A_{ij}^2}{\norm{\aj}_2^2}$ 
		\\[12pt]
		{\bf $\ell_2$-$\ell_1$} & $  
		y_i \cdot \frac{|A_{ij}|}{\lones{\ai}}$ & $  
	   \frac{A_{ij}^2}{\norm{A}_{\mathrm{F}}^2}$ 
		\\[12pt]
		{\bf $\ell_2$-$\ell_1$} & $  
		y_i \cdot \frac{|A_{ij}|}{\lones{\ai}}$ & $  
		\propto x_j^2 \cdot \1_{A_{ij}\neq 0}$
		\\[12pt]
		{\bf $\ell_2$-$\ell_1$} & $  
		y_i \cdot \frac{|A_{ij}|}{\lones{\ai}}$ & $  
		\frac{A_{ij} \cdot x_j^2}{\sum_{k \in [n]} \norm{A_{:k}}_1\cdot  x_k^2}$
		\\[12pt]
		{\bf $\ell_2$-$\ell_2$} & $  
		\frac{\lones{\ai}^2}{\sum_{k\in[m]}\lones{\ai[k]}^2}\cdot\frac{|A_{ij}|}{\lones{\ai}}$
		& $  
		\frac{\lones{\aj}^2}{\sum_{k\in[n]}\lones{\aj[k]}^2}\cdot\frac{|A_{ij}|}{\lones{\aj}}$
		\\[12pt]
		{\bf $\ell_2$-$\ell_2$} & $  
		\frac{{ y_i}^2}{\ltwo{y}^2}\cdot\frac{|A_{ij}|}{\lones{\ai}}$ & $  
		\frac{{x_j}^2}{\ltwo{x}^2}\cdot\frac{|A_{ij}|}{\lones{\aj}}$  \\
		\bottomrule
	\end{tabular}
\egroup
	\caption{\small\textbf{The distributions $p,q$  used in our coordinate gradient estimator.} Comments: The estimator is of the form  
		$ {\tg(x,y) = \big(\frac{1}{p_{ij}} {y_{i} A_{i j}}\cdot e_{j}, - 
			\frac{1}{q_{lk}}{A_{lk}x_{k}}\cdot e_{l}\big)}$ where $i,j\sim 
		p$ and 
		$l,k\sim q$. }
	\label{table:dist-sublinear}

\vspace{16pt}
	\centering
	\renewcommand{\arraystretch}{2}
	\bgroup
	\everymath{\displaystyle}
	\begin{tabular}{ccc}
		\toprule
		{Setting} & { $p_{ij}$  } & { $q_{ij}$ } \\
		\midrule
		{\bf $\ell_1$-$\ell_1$} & $  
		\frac{y_i+2[y_0]_i}{3}\cdot\frac{A_{ij}^2}{\ltwo{\ai}^2}$
		& $  \frac{x_j+2[x_0]_j}{3} 
		\cdot\frac{A_{ij}^2}{\ltwo{\aj}^2}$ \\[12pt]
		{\bf $\ell_2$-$\ell_1$} & $  
		\frac{y_i+2[y_0]_i}{3}\cdot\frac{|A_{ij}|}{\lones{\ai}}$ & 
		$  \frac{A_{ij}^2}{\norm{A}_\mathrm{F}^2}$
		\\[12pt]
		{\bf $\ell_2$-$\ell_1$} & $  
		\frac{y_i+2[y_0]_i}{3}\cdot\frac{|A_{ij}|}{\lones{\ai}}$ & 
		$  \propto [x - x_0]_j^2 \cdot \1_{A_{ij} \neq 0}$
		\\[12pt]
		{\bf $\ell_2$-$\ell_1$} & $  
		\frac{y_i+2[y_0]_i}{3}\cdot\frac{|A_{ij}|}{\lones{\ai}}$ & 
		$ \frac{|A_{ij}| \cdot [x - x_0]_j^2}{\sum_{k \in [n]} \norm{A_{:k}}_1 \cdot [x - x_0]_k^2}$
		\\[12pt]
		{\bf $\ell_2$-$\ell_2$} & $ 
		\frac{\lones{\ai}^2}{\sum_{k\in[m]}\lones{\ai[k]}^2}\cdot\frac{|A_{ij}|}{\lones{\ai}}$
		& $
		\frac{\lones{\aj}^2}{\sum_{k\in[n]}\lones{\aj[k]}^2}\cdot\frac{|A_{ij}|}{\lones{\aj}}$
		\\[12pt]
		{\bf $\ell_2$-$\ell_2$} & $ 
		\frac{[y - y_0]_i^2}{\norm{y - y_0}_2^2} \cdot\frac{|A_{ij}|}{\lones{\ai}}$
		& $
		\frac{[x - x_0]_j^2}{\norm{x - x_0}_2^2}\cdot\frac{|A_{ij}|}{\lones{\aj}}$
		\\[12pt]
		\bottomrule
	\end{tabular}
	\caption{\small\textbf{The distributions $p$, $q$ used for  our reduced variance 
		coordinate 
		gradient 
		estimator.} Comments: The estimator is of the form 
		$ \tg(x,y) = \big(A^\top y +\frac{1}{p_{ij}} {(y_{i}-y_{0,i}) A_{i 
					j}} \cdot
			e_{j}$, $- Ax - 
			\frac{1}{q_{lk}}{A_{lk}(x_{k}-x_{0,k})} \cdot e_{l}\big)$ 
		where 
		$i,j\sim p$ 
		and $l,k\sim q$ and $x_0, y_0$ is a reference point.}
	\label{table:dist-vr}
	\egroup
\end{table} }

\subsection{Related work}

\paragraph{Coordinate methods.}
Updating a single coordinate at a time---or more broadly computing only a 
single coordinate of the gradient at every iteration---is a well-studied and 
successful technique in optimization~\cite{Wright15}. Selecting coordinates 
at random is key to obtaining strong performance guarantees: 
\citet{StrohmerV09} show this for linear regression, \citet{ShalevT11} show 
this for $\ell_1$ regularized linear models, and \citet{Nesterov12} shows 
this for general smooth minimization. 
Later works~\citep{LeeS13,ZhuQRY16,NesterovS17} propose accelerated 
coordinate methods. These works share two common themes: selecting the 
gradient coordinate from a non-uniform distribution (see 
also~\cite{RichtarikT16}), and augmenting the 1-sparse stochastic gradient 
with a dense momentum term. These techniques play important roles in our 
development as well.

To reap the full benefits of coordinate methods, iterations must be very cheap, ideally taking near-constant time. However, most coordinate methods require super-constant time, typically in the form of a 
vector-vector computation. Even works that consider coordinate methods 
in a primal-dual context~\cite{ShalevZ13,ZhuLY16,ZhangX17,NamkoongD16,ShalevW16} perform the coordinate updates only on the dual variable and require a 
vector-vector product (or more generally a component gradient 
computation) at every iteration. 

A notable exception is the work of \citet{wang2017randomized,wang2017primal} which develops a primal-dual stochastic coordinate method for solving Markov decision processes, essentially viewing them as $\ell_\infty$-$\ell_1$ bilinear saddle-point problems. Using a tree-based $\ell_1$ sampler data structure similar to the $\ell_1$ sampler we use for simplex domains for the sublinear case, the method allows for $\Otil{1}$ iterations and a potentially sublinear runtime scaling as $\epsilon^{-2}$%
. \citet{TanZML18}  also consider bilinear saddle-point problems and 
variance reduction. Unlike our work, they assume a separable domain, use  
uniform sampling, and do not accelerate their variance reduction scheme 
with extra-gradient steps. The separable domain makes attaining constant 
iteration cost time much simpler, since there is no longer a normalization 
factor to track, but it also rules out applications to the simplex domain.
While \citet{TanZML18} report promising 
empirical results, their theoretical guarantees do not improve upon prior 
work. 

Our work develops coordinate methods with $\Otil{1}$ iteration cost for new types of problems. Furthermore, it maintains the iteration efficiency even in the presence of dense components arising from the update, thus allowing for acceleration via an extra-gradient scheme.

\paragraph{Data structures for optimization.}
Performing iterations in time that is asymptotically smaller than the number 
of variables updated at every iteration forces us to carry out the updates 
implicitly using data structures; several prior works 
employ data structures for exactly the same reason. One of the most
similar examples comes from~\citet{LeeS13}, who design a data structure for 
an accelerated coordinate method in Euclidean geometry. In our 
terminology, their data structure allows performing each iteration in time 
$O(\rcs)$ while implicitly updating variables of size $O(n)$. 
\citet{DuchiSSC08} design a data structure based on 
balanced search trees that supports efficient Euclidean projection to the 
$\ell_1$ ball of vector of the form $u+s$ where $u$ is in the $\ell_1$ ball 
and $s$ is sparse. They apply it in a stochastic gradient method for 
learning $\ell_1$ regularized linear classifier with sparse features. Among the
many applications of this data structure, \citet{NamkoongD16} adapt it to 
efficiently compute Euclidean projections into the intersection of the 
simplex and a $\chi^2$ ball for 1-sparse updates. \citet{ShalevW16} and \citet{wang2017randomized,wang2017primal}, among others, 
 use binary tree data structures to perform multiplicative weights projection to the simplex and sampling from the iterates.

A recent work of~\citet{SidfordT18} develops a data structure which is somewhat similar to our $\AEM$ data structure, for updates arising from a primal-dual method to efficiently solve $\ell_\infty$ regression. Their data structure was also designed to handle updates to a simplex variable which summed a structured dense component and a sparse component. However, the data structure design of that work specifically exploited the structure of the maximum flow problem in a number of ways, such as bounding the sizes of the update components and relating these bounds to how often the entire data structure should be restarted. Our data structure can handle a broader range of structured updates to simplex variables and has a much more flexible interface, which is crucial to the development of our variance-reduced methods as well as our applications.

Another notable use of data structures in optimization appears in second order methods, where a long line of work uses them to efficiently solve 
sequences of linear systems and approximately compute iterates~\cite{Karmarkar84,AndersonRTTW96,LeeS15, CohenLS18, LeeSZ19,Brand19,BrandKSS20}. Finally, several works on low rank 
	optimization make use of sketches to efficiently represent their iterates 
	and solutions~\cite{ClarksonW13,YurtseverUTC17}.

\paragraph{Numerical sparsity.}
Measures of numerical sparsity, such as the $\ell_2$/$\ell_\infty$ or 
$\ell_1$/$\ell_2$ ratios, are continuous and dimensionless 
relaxations of the $\ell_0$ norm. The \emph{stable rank} of a matrix $A$ 
measures the numerical sparsity of its singular values (specifically, their 
squared $\ell_2$/$\ell_\infty$ ratio)~\cite{CohenNW16}. 

For linear 
regression, stochastic 
methods generally outperform exact gradient methods only when $A$ is 
has low stable rank, cf.\ discussion in~\cite[Section 
4.3]{CarmonJST19}, i.e., numerically sparse singular values. In recent work, \citet{GuptaS18} develop algorithms for linear regression and eigenvector problems for matrices with numerically sparse entries (as opposed to singular values). Our paper further broadens the scope of matrix problems 
for which we can benefit from numerical sparsity. Moreover, our results 
have implications for regression as well, improving on~\cite{GuptaS18} in 
certain numerically sparse regimes. 

In recent work by \citet{BabichevOB19}, the authors develop primal-dual 
sublinear methods for $\ell_1$-regularized linear multi-class 
classification (bilinear games in $\ell_1$-$\ell_\infty$ geometry), 
and obtain 
complexity improvements depending on the numerical sparsity of the 
problem. Similarly to our work, careful design of the sampling distribution 
plays a central role in~\cite{BabichevOB19}. They also develop a data 
structure that allows iteration cost independent of the number of classes. 
However, unlike our work, \citet{BabichevOB19} rely on 
sampling entire rows and columns, have iteration costs linear in 
$n+m$, and do not utilize variance reduction. We believe that our 
techniques can yield improvements in their setting.

\arxiv{

\subsection{Paper organization}

In Section~\ref{sec:prelims}, we set up our terminology, notation, the 
interfaces of our data structures, and the different matrix access models we 
consider.  In Section~\ref{sec:framework} we develop our algorithmic framework: we present coordinate stochastic gradient methods in Section~\ref{sec:sublinear} and their reduced variance counterparts in Section~\ref{sec:vr}. In Section~\ref{sec:example} we apply both methods to solving $\ell_1$-$\ell_1$ matrix games; we show how to implement the method using our data structures and analyze the runtime. In Section~\ref{sec:ds}, we discuss in detail the implementation 
and analysis of our data structures. Finally, in Section~\ref{sec:app} we 
specialize our results to obtain algorithms for minimum enclosing ball and maximum inscribed ball problems as 
well as linear regression. Many proof details as well as our algorithms for other domain setups, i.e.\ $\ell_2$-$\ell_1$ and $\ell_2$-$\ell_2$ are deferred to the appendix.
}

\notarxiv
{
	\bibliographystyle{IEEEtran}

	\clearpage
}

\newpage

\section{Preliminaries}
\label{sec:prelims}

In Section~\ref{ssec:local}, we abstract the properties of the different 
domains we handle into a general notion of a ``local norm'' setup
under which we develop our results. In Section~\ref{ssec:problemdef}, we 
give the definition and optimality criterion of the bilinear saddle-point 
problem we study.  In Section~\ref{ssec:accessmodel}, we give 
the matrix access models used  in the algorithms we design. In 
Section~\ref{ssec:interfaces}, we summarize the interfaces and complexity 
of the data structures we design, deferring their detailed implementations to 
Section~\ref{sec:ds}. 

\subsection{Local norm setups}
\label{ssec:local}

The analyses of our algorithms cater to the geometric of each specific domain. To express our results generically, for each  pair of domains $\zset = \xset \times \yset$, we define an 
associated ``local norm setup'', which contains various data tailored for our 
analyses. While initially this notation may appear complicated or cumbersome, later it helps avoid redundancy in the paper. Further, it clarifies 
the structure necessary to generalize our methods to additional domain 
geometries.
\begin{definition}
\label{def:setup}
A \emph{local norm setup} is the quintuplet \setup, where
\begin{enumerate}
	\item $\zset$ is a compact and convex subset of $\zset^* \defeq \R^{n}\times 
	\R^{m}$.
	\item $\norm{\cdot}_{\cdot}$ is a local norm: for every 
	$z\in\zset$, the function $\norm{\cdot}_z:\zset^*\to\R_{\geq 0}$ is a 
	norm on $\zset^*$.
	\item $r:\zset\to\R$ is a convex distance generating function: 
	 its 
	induced Bregman divergence
	\begin{equation*}
	V_{z}(z') \defeq r(z') - r(z) - \inner{\grad r(z)}{z'-z}.
	\end{equation*}
	satisfies
	\begin{equation}\label{eq:strong-convexity}
	\inner{\gamma}{z-z'} - V_{z}(z') \le \half \norm{\gamma}_*^2
	 \defeq 
	\half \max_{s\in\zset}\norm{\gamma}_{s}^2 ~~\text{for 	all }z,z' \in 
	\zset~\text{and}~\gamma\in\zset^*.
	\end{equation}
	\item $\Theta=\max_{z,z'\in\zset}\{r(z)-r(z')\}$ is the range of $r$. For 
	$z^*\in\argmin_{z\in\zset}r(z)$ we have $\Theta$ is an upper bound on 
	the range of $V_{z^*}(z)\le \Theta$ for all $z\in\zset$.
	\item $\clip:\zset^*\to\zset^*$ is a mapping that enforces a 
	local version of~\eqref{eq:strong-convexity}:
	\begin{equation}\label{eq:strong-convexity-local}
	|\inner{\clip(\gamma)}{z-z'}| - V_{z}(z') \le  
	\norm{\gamma}_z^2~~\text{for all 
	}z,z'\in\zset~\text{and}~\gamma\in\zset^*,
	\end{equation}
	and satisfies the distortion guarantee
	 \begin{equation}\label{eq:clipping-distortion}
	 |\inner{\gamma-\clip(\gamma)}{z}| \le \norm{\gamma}_z^2~~\text{for 
	 all 
	 }z\in\zset~\text{and}~\gamma\in\zset^*.
	 \end{equation}
\end{enumerate}
\end{definition}

\begin{table}
	\centering
	\renewcommand{\arraystretch}{1.5} %
	\begin{tabular}{cccc}
		\toprule
		& {\bf $\ell_1$-$\ell_1$} & {\bf $\ell_2$-$\ell_1$} & {\bf 
		$\ell_2$-$\ell_2$}\\
		\midrule
		$\xset$ & $\Delta^n$ &  $\ball^n$& $\ball^n$ \\ 
		$\yset$ & $\Delta^m$ & $\Delta^m$ & $\ball^m$ \\
		$\norm{\delta}_z$ 
		& 
		$\sqrt{\sum_{k\in [n+m]} [z]_k [\delta]_k^2 }$
		& 
		$\sqrt{\norm{\delta\x}_2^2+ \sum_{i \in [m]} [z\y]_i [\delta\y]_i^2}$ 
		& 
		$\norm{\delta}_2$ \\
		$r$ & 
		$\sum_{k\in [n+m]} [z]_k \log [z]_k$ &
		 $\half \norm{z\x}_2^2 + \sum_{i \in [m]} [z\y]_i \log 
		[z\y]_i$ & 
		$\half \norm{z}_2^2$ \\
		$\Theta$ & $\log(mn)$ & $\half + \log(m)$ & $1$\\
		$\clip(\delta)$ & $\sign(\delta)\circ\min\{1,|\delta|\}$ & 
		$\left(\delta\x,\sign(\delta\y)\circ\min\{1,|\delta\y|\}\right)$ & $\delta$ \\
		\bottomrule
	\end{tabular}
	\caption{\textbf{Local norm setups.}  Comments: In each case, $\zset = \xset \times \yset$, 
	$\Delta^n$ is the probability simplex $\{x \mid x \in \R_{\geq 0}^n, 
	\1_n^\top x = 1\}$, $\ball^n$ is the Euclidean ball 
		$\{x \mid x \in \R^n, \norm{x}_2 \leq 1\}$, the operations $\sign$, $\min$, and $|\cdot|$ are performed entrywise on a vector, and $\circ$ stands for the entrywise product between vectors. 
}
	\label{tab:local-norm}
\end{table}

Table~\ref{tab:local-norm} summarizes the three local norm setups we 
consider. Throughout the paper,
\begin{equation*}
\mbox{for a vector $z\in\xset\times \yset$, we denote its $\xset$ and 
$\yset$ blocks by $z\x$ and $z\y$.}
\end{equation*}
In addition, we write coordinate $i$ of any vector $v$ as $[v]_i$.

\begin{proposition}
\label{prop:local-norm}
The quintuplets \setup in Table~\ref{tab:local-norm} satisfy the local norm 
setup requirements in Definition~\ref{def:setup}.
\end{proposition}

\noindent
While Proposition~\ref{prop:local-norm} is not new, for completeness and compatibility with our 
notation we prove it in Appendix~\ref{app:prelims-proofs}.

In each local norm setup, we slightly overload notation and use $\norm{\cdot}$ (without a subscript) to denote the dual norm of $\norm{\cdot}_*$, i.e., $\norm{\eta} \defeq \max_{\delta : \norm{\delta}_* \leq 1} \delta^\top \eta$. In each domain geometry $\norm{\cdot}$ and $\norm{\cdot}_*$ are as follows:
\begin{equation}
\begin{aligned}
\norm{\eta} &= \sqrt{\norm{\eta\x}_1^2 + \norm{\eta\y}_1^2} 
 &\quad \norm{\delta}_*    &= \sqrt{\norm{\delta\x}_\infty^2 + \norm{\delta\y}_\infty^2} 
&& \text{ for } \ellone \\
\norm{\eta} &= \sqrt{\norm{\eta\x}_2^2 + \norm{\eta\y}_1^2} 
 &\quad \norm{\delta}_* &= \sqrt{\norm{\delta\x}_2^2 + \norm{\delta\y}_\infty^2} 
&& \text{ for } \elltwoone \\
\norm{\eta} &= \norm{\eta}_2 
&\quad  \norm{\delta}_* &= \norm{\delta}_2 
&& \text{ for } \elltwo ~.
\end{aligned}
\label{eq:norms_for_setups}
\end{equation}

\subsection{The problem and optimality criterion}
\label{ssec:problemdef}

Throughout, we consider the bilinear saddle point problem
\begin{equation}
\min_{x \in \xset} \max_{y \in \yset} f(x, y), \text{ where } f(x, y) \defeq 
y^\top A x +b^\top x-c^\top y, \text{ for } A \in 
\R^{m \times n}, b \in \R^n \text{ and } c \in \R^m.
\label{eq:sublinear-problem}
\end{equation}
We will always assume that every row and column of $A$ has at least one nonzero entry (else removing said row or column does not affect the problem value), so that the number of nonzeros $\nnz$ is at least $m + n - 1$. To simplify the exposition of the $\ellone$ and $\elltwoone$ setups we 
will assume $b = \mathbf{0}_n$ and $c = \mathbf{0}_m$ as is standard in 
the literature. Adding linear terms to these setups is fairly straightforward and does not affect the complexity (up to logarithmic factors) of our designed algorithms using data structures designed in this paper (specifically $\AEM$ in Section~\ref{ssec:interfaces}); see Section~\ref{sec:app} for an example. The gradient 
mapping associated with \eqref{eq:sublinear-problem} for $z = (z\x, z\y) 
\in \zset=\xset\times \yset$ is
\begin{align}
\label{eq:gdef}
g(z) \defeq (\nabla_x f(z), -\nabla_y f(z)) = (A^\top z\y+b, -Az\x+c).
\end{align}
The mapping $g$ is continuous and monotone, where we call $g$ monotone if and only if \[\inner{g(z') - g(z)}{z' - z} \geq 0,\ \forall z,z'\in\zset .\] This holds due to the convexity-concavity (indeed, bilinearity) of function $f$.
Our goal is to design randomized 
	algorithms for finding an (expected) $\epsilon$-accurate saddle point 
$z\in\zset$ such that, in expectation,
\begin{equation}\label{eq:gap-def}
\E\gap(z) \defeq \E\left[\max_{y' \in \yset} f(z\x,y') - \min_{x' \in \xset} f(x', z\y) \right]
\le \epsilon.
\end{equation}

 In order to do so, we aim to find a sequence $z_1,z_2,\ldots,z_K$ with 
(expected) low  
\emph{average regret}, i.e., such that  $ 
\E\max_{u\in\zset}\left\{\frac{1}{K}\sum_{k=1}^K\inner{g(z_k)}{z_k-u}\right\}
 \le \eps$. Due to bilinearity of $f$ we have
\begin{align*}
\E\gap\left(\frac{1}{K}\sum_{k=1}^K z_k\right) = \E
\max_{u\in\zset}\left\{\frac{1}{K}\sum_{k=1}^K\inner{g(z_k)}{z_k-u}\right\}
 \le \eps.
\end{align*}

Finally, we make the explicit assumption that whenever we are discussing an algorithm in this paper with a simplex domain (e.g.\ in $\ell_1$-$\ell_1$ or $\ell_2$-$\ell_1$ case), the quantity $\Lco/\eps$ is bounded by $(m + n)^3$, as otherwise we are in the high-accuracy regime where the runtimes of interior point methods or cutting-plane methods~\citep{lee2015faster,jiang2020improved} are favorable. Specifically for $\ell_1$-$\ell_1$ matrix games in this regime, interior-point methods~\cite{LeeS15,CohenLS18,BrandKSS20} are always faster, see footnote in Section~\ref{sec:our_results}.
We make this assumption for notational convenience when discussing logarithmic factors depending on multiple quantities, such as $m$, $n$, $\Lco,$ and $\eps^{-1}$.

\subsection{Matrix access models}
\label{ssec:accessmodel}

We design randomized algorithms which require accessing and sampling from the matrix $A\in\R^{n\times m}$ in a variety of ways. Here, we list these operations, where we assume each takes constant time. Specific algorithms only require access to a subset of this list; we make a note of each algorithm's requirements when presenting it.

\begin{enumerate}[label=A\arabic*.]
	\item For $i,j \in [m]\times[n]$, return $A_{ij}$.
	\item For $i \in [m]$ and $p\in \{1, 2\}$, draw $j \in [n]$ with 
	probability $|A_{ij}|^p / \norm{\ai}_p^p$.
	\item For $j \in [n]$ and $p\in \{0, 1, 2\}$, draw $i \in [m]$ with 
	probability $|A_{ij}|^p / \norm{\aj}_p^p$.
	\item For $i \in [m]$  ($j \in [n]$) and $p \in \{1, 2\}$, return 
	$\norm{\ai}_p$  ($\norm{\aj}_p$).
	\item For $p \in \{1, 2\}$, return $\max_{i \in [m]} \norm{\ai}_p$, 
	$\max_{j \in [n]} \norm{\aj}_p$, $\nnz$, $\rcs$, and 
	$\norm{A}_{\textup{F}}$.
\end{enumerate}

Given any representation of the matrix as a list of nonzero entries and their indices, we can always implement the access modes above (in the assumed constant time) with $O(\nnz)$ time preprocessing; see e.g.\ \citet{Vose91} for an implementation of the sampling (in a unit cost RAM model). Our variance-reduced algorithms have an additive $O(\nnz)$ term appearing in their runtimes due to the need to compute at least one matrix-vector product to implement gradient estimators. Thus, their stated runtime bounds hold independently of matrix access assumptions.

\subsection{Data structure interfaces}
\label{ssec:interfaces}

We rely on data structures to maintain and sample from the iterates of our 
algorithms. Below, we give a summary of the operations supported by our 
data structures and their runtime guarantees. We show how to implement these data structures in Section~\ref{sec:ds}.

\subsubsection{$\IM_p$}
\label{ssec:interface-norm}

Given $p \in \left\{1, 2\right\}$, we design a data structure $\IM_p$ which 
maintains an implicit representation of the current iterate $x \in \R^n$ and 
a running sum $s$ of all iterates. At initialization, this data structure takes as input the initial iterate $x_0$ to be maintained and for  $p=2$, the data structure also takes 
as input a fixed vector $v$. It then supports the 
following operations.

\renewcommand{\arraystretch}{1.2}

\begin{table}[h]
	\centering
	\begin{tabular}{c||l|l}
		{\bf Category}                   & {\bf Function} & {\bf Runtime} \\ \hline
		\multicolumn{1}{c||}{\multirow{1}{*}{initialize}}
		& $\Initialize(x_0, v)$: $x \gets x_0$, $s \gets 0$ 
		& $ O(n)$ 
		\\ 	\hline
		\multicolumn{1}{c||}{\multirow{5}{*}{update}}
		& $\Scale(c)$: $x\gets c x$ & $O(1)$  
		\\ \cline{2-3} 
		\multicolumn{1}{c||}{} & $\AddSparse(j,c)$: $[x]_j\gets[x]_j+c$ (if $p = 
		1$, we require $c \geq -[x]_j$)  &  $O(\log n)$$^\dagger$\\ 
		\cline{2-3} 
		\multicolumn{1}{c||}{} & $\AddDense(c)$: $x\gets x+c v$ (supported 
		if 
		$p=2$)
		& $O(1)$\\ 
		\cline{2-3} 
		\multicolumn{1}{c||}{} & $\SumUp()$: $s\gets s+x$ & $O(1)$ \\ \hline
		\multicolumn{1}{c||}{\multirow{3}{*}{query}}    
		& $\Get(j)$: Return $[x]_j$ & $O(1)$\\ \cline{2-3} 
		\multicolumn{1}{c||}{} & $\GetSum(j)$: Return $[s]_j$  & $O(1)$\\ 
		\cline{2-3} 
		\multicolumn{1}{c||}{} & $\GetNorm()$: Return $\norm{x}_p$ & $O(1)$ 
		\\ 
		\hline
		sample$^\dagger$
		& $\Sample()$: Return $j$ with probability $[x]_j^p / \norm{x}_p^p$ 
		& 
		$O(\log n)$  \\ 
	\end{tabular}
\vspace{6pt}

{\small{
$^\dagger$ An alternative implementation does not support $\Sample$, but performs $\AddSparse$ in time $O(1)$.}}
\end{table}

The implementation of $\IM_p$ is given in Section~\ref{sec:iterate_maintain_proof}. In Sections~\ref{ssec:l2l1sub} and~\ref{ssec:vr-l2l1} we use variants of this data structure $\WIM_2$ and $\CIM_2$, and defer the detailed discussions of their implementations to Appendix~\ref{app:ds-proofs}.

\subsubsection{ $\AEM$ }
\label{ssec:interface-simplex}

To maintain multiplicative weights updates with a fixed dense component, we design a data structure $\AEM$ 
initialized with an arbitrary point $x_0\in\Delta^n$, a direction $v\in\R^n$, a decay constant $\kappa\in[0,1]$ and an approximation error parameter $\varepsilon$. In order to specify the implementation of our data structure, we require the following definition.

\begin{definition}[$\beta$-padding]
\label{def:betastable}
For $x, x' \in \Delta^n$, we say $x'$ is a $\beta$-padding of $x$ if $x' = \tx/\norm{\tx}_1$, for a point $\tx \in \R^n_{\ge 0}$ with 
	$\tx \ge x$ entrywise
	and
	$\norm{\tx - x}_1 \le \beta$.
\end{definition}
\noindent
Notions similar to $\beta$-padding appear in previous literature~\cite[e.g.,][]{KoutisMP10}. A key technical property of $\beta$-paddings is that they do not increase entropy significantly (see Lemma~\ref{lem:stable}). 

$\AEM$ has maintains two vectors $x,\hat{x}\in\Delta^n$ that, for an error tolerance parameter $\varepsilon$, satisfy the  invariant
\begin{equation}\label{eq:padding-invariant}
\hat{x}\text{ is a }\veps\text{-padding of }x.
\end{equation}
an error tolerance parameter $\varepsilon$
We now specify the interface, where $\circ$ denotes elementwise product, $[x^\kappa]_j = [x]_j^\kappa$ denotes elementwise power, $\Pi_{\Delta}(z) = z/\lone{z}$ normalizes $z\in \R_{\geq 0}^n$ to lie in the simplex, and $\norm{s}_0$ denotes the number of nonzeroes in vector $s$.
To state our runtimes, we define 
\[\omega \defeq \max\left(\frac{1}{1 - \kappa},\;\frac{n}{\lambda\veps}\right).\]
For most of our applications of $\AEM$, $\omega$ is a polynomial in $m$ and $n$ (our iterate dimensions), so $\log(\omega)=O(\log(mn))$ (with the exception of our maximum inscribed ball application, where our runtimes additionally depend polylogarithmically on the size of the hyperplane shifts $b$; see Remark~\ref{rem:polylogbc}). We defer a more fine-grained runtime discussion to Section~\ref{sec:scm}.

\newcommand{\theAEM}{
\begin{table}[h]
	\centering
	\renewcommand{\arraystretch}{1.25}
	\begin{tabular}{c||l|l}
		{\bf Category}                   & {\bf Function} & {\bf 
			Runtime} \\ \hline
		\multicolumn{1}{c||}{\multirow{1}{*}{initialize}}
		& $\Initialize(x_0, v, \kappa, \varepsilon, \lambda):$ $\kappa \in [0, 1)$, $\varepsilon > 0$, $\min_j [x_0]_j \ge \lambda$ & $O(n\log n \log^2 \omega)$\\ 
		\hline
		\multicolumn{1}{c||}{\multirow{3}{*}{update}}
		& $\MultSparse(g)$: $x \gets \varepsilon \text{-padding of }\Pi_{\Delta}(x \circ \exp(g))$ & 
		$O(\norm{g}_0\log^2 n \log^2 \omega)$ \\ \cline{2-3} 
		\multicolumn{1}{c||}{} & $\DenseStep()$: 
		$x \gets \Pi_{\Delta}(x^{\kappa} 
		\circ \exp(v))$
		& $O(\log n)$ \\ \cline{2-3} 
		\multicolumn{1}{c||}{} & $\SumUp()$: $s\gets s+\hat{x}$ (recall invariant~\eqref{eq:padding-invariant}) & $O(\log n \log \omega)$ \\ 
		\hline
		\multicolumn{1}{c||}{\multirow{2}{*}{query}}    
		& $\Get(j)$: Return 
		$[\hat{x}]_j$ & $O(\log n \log \omega)$  \\ 
		\cline{2-3} 
		\multicolumn{1}{c||}{} & $\GetSum(j)$: Return $[s]_j$& 
		$O(\log^2 \omega)$\\ 
		\hline
		sample                    
		& $\Sample()$: 
		$ \text{Return } j \text{ with probability } 
		[\hat{x}]_j$  & 
		$O(\log n \log \omega)$  \\ 
	\end{tabular}
\end{table}
}
\theAEM

The role of $\AEM$ is in to efficiently implement the regularized 
and reduced-variance stochastic mirror descent steps of the 
form~\eqref{eq:l1-reg-smd}. To do this, we initialize the data structure 
with 
$v=(1-\kappa)\log x_0 - \eta\kappa g_0\x$. Then, the 
iteration~\eqref{eq:l1-reg-smd} consists of calling $\DenseStep()$ 
followed by 
$\MultSparse(-\eta \kappa \tilde{\delta}\x )$.

\section{Framework}\label{sec:framework}

In this section, we develop our algorithmic frameworks. The resulting algorithms have either sublinear or variance-reduced complexities. We develop our sublinear coordinate method framework in Section~\ref{sec:sublinear}, and its variance-reduced counterpart in Section~\ref{sec:vr}.

\subsection{Sublinear coordinate methods}
\label{sec:sublinear}

In Section~\ref{ssec:sublin-analysis} we introduce the concept of a \emph{local gradient estimator}, which allow stronger guarantees for  stochastic mirror descent with clipping (Algorithm~\ref{alg:sublinear}) via local norms analysis. Then, in Section~\ref{ssec:coordest} we state the form of the specific local gradient  estimators we use in our coordinate methods, and motivate the values of $\Lco$ in Table~\ref{table:L}. %
\subsubsection{Convergence analysis}\label{ssec:sublin-analysis}

\begin{definition}
\label{def:sublinear}
For local norm setup \setup, we call a stochastic gradient estimator 
$\tilde{g}:\zset\rightarrow\zset^*$ an 
\emph{$L$-local estimator} if it satisfies the following 
properties for all $z\in\zset$:
\begin{enumerate}
\item Unbiasedness: $\E[\tilde{g}(z)]=g(z)$.
\item Second moment bound: for all $w \in \zset$, 
$\E[\norm{\tilde{g}(z)}_{w}^2] 
\le  L^2$. 
\end{enumerate}
\end{definition}

The following lemma shows that $L$-local estimators are unbiased for $L$-bounded operators.
\begin{lemma}\label{lem:sublinear-implication}
A gradient mapping that admits an $L$-local estimator satisfies 
$\norm{g(z)}_* \le L$ for all $z\in\zset$.
\end{lemma}
\begin{proof}
For every $z\in\zset$, the function  $\norm{\cdot}_z^2$ is convex. Thus by Jensen's inequality,
\[
\norm{g(z)}_z^2 = \norm{\E\tilde{g}(z)}_z^2 \le \E\norm{\tilde{g}(z)}_z^2\le L^2.
\]
Taking supremum over $z\in\zset$ gives $\norm{g(z)}_*^2\le L^2$.
\end{proof}
\noindent
We note that the same result \emph{does not} hold for $\tilde{g}$ because maximum and expectation do not commute. That is, $\E \norm{\tilde{g}}_*^2$ is not bounded by $L^2$. This fact motivates our use of local norms analysis.

Below, we state Algorithm~\ref{alg:sublinear}, stochastic 
mirror descent with clipping, and a guarantee on its rate of convergence using local gradient estimators. 
We defer the proof to Appendix~\ref{sec:framework-proofs} and note here 
that it uses the ``ghost iterates'' technique due to~\citet{NemirovskiJLS09} 
in order to rigorously bound the expected regret with respect to the best 
response to our iterates, rather than a pre-specified point. This technique is purely analytical and does not affect the algorithm. We 
also note that the second inequality in Proposition~\ref{prop:sublinear} holds with any convex-concave function $f$, similarly to~\cite[Corollary 1]{CarmonJST19}; the first uses bilinearity of our problem structure.

\SetKwComment{Comment}{$\triangleright$\ }{}
\SetCommentSty{color{black}}
\begin{algorithm}[h]
	\DontPrintSemicolon
	\KwInput{Matrix $A\in\R^{m\times n}$, %
	$L$-local gradient estimator $\tilde{g}$, clipping function $\clip(\cdot)$}
	\KwOutput{A point with  $O(\frac{\Theta}{\eta T} + \eta L^2)$ expected 
	duality gap} %
	\Parameter{Step-size $\eta$, number of iterations $T$}
	
	$z_0 \gets \argmin_{z\in\zset} r(z)$\;
	
	\vspace{3pt}
		\For{$t=1,\ldots,T$}
		{
			$z_t\gets\arg\min_{z\in\zset}\left\{\langle 
			\clip(\eta\tilde{g}(z_{t-1})),z\rangle+V_{z_{t-1}}(z)\right\}$\;
		}
		\Return $\frac{1}{T + 1}\sum_{t=0}^T z_t$\;

\caption{Stochastic mirror descent}
	\label{alg:sublinear}
\end{algorithm}

\begin{restatable}{proposition}{restatesublinear}
\label{prop:sublinear}
Let \setup  be a local norm 
setup, let $L,\epsilon>0$, and let $\tilde{g}$ be an $L$-local 
estimator. 
Then, for $\eta\le\frac{\eps}{9L^2}$ and $T\ge\tfrac{6\Theta}{\eta\eps}\ge 
\frac{54L^2\Theta}{\eps^2}$, 
Algorithm~\ref{alg:sublinear} 
outputs a point 
$\bar{z}$ such that
\begin{equation*}
\E \gap(\bar{z}) \le \E\left[\sup_{u\in\zset}\frac{1}{T + 1}\sum_{t = 0}^T\inner{g(z_t)}{z_t-u}\right]  \le \epsilon.
\end{equation*}
\end{restatable} %
\subsubsection{Coordinate gradient estimators}
\label{ssec:coordest}

We now state the general form which our local gradient estimators 
$\tilde{g}$  take. At a point $z \in \zset$, for specified sampling 
distributions $p(z), q(z)$, sample $i\x,j\x\sim p(z)$ and $i\y,j\y\sim q(z)$. 
Then, define
\begin{equation}
\label{eq:estimate-l1}
\begin{aligned}
\tilde{g}(z)\defeq \left(\frac{A_{i\x j\x}[z\y]_{i\x} }{p_{i\x 
j\x}(z)}e_{j\x},\frac{-A_{i\y j\y}[z\x]_{j\y} }{q_{i\y j\y}(z)}e_{i\y}\right) + g(0)
~~\mbox{where}~~g(0) = (b, c).
\end{aligned}
\end{equation}
It is clear that regardless of the distributions $p(z), q(z)$, for the gradient 
operator in~\eqref{eq:gdef}, $\tilde{g}(z)$ is an unbiased gradient estimator 
($\E[\tilde{g}(z)]=g(z)$) and $\tilde{g}(z) - g(0)$ is 2-sparse.

\paragraph{Optimal values of $\Lco$.} In the remainder of this section we assume for simplicity the $g(0)=0$ (i.e.\ the objective $f$ in~\eqref{eq:sublinear-problem} has not linear terms). Here we compute the optimal values of $L$ for local gradient estimators (see Definition~\ref{def:sublinear}) of the form \eqref{eq:estimate-l1} for each of the local norm setups we consider. This motivates the values of $\Lco$ we derive in the following sections. 
First, in the $\ellone$ case, the second moment of ${\norm{\tilde{g}\x(z)}_{w\x}^2}$ (the local norm of the $\xset$ block of $\tilde{g}(z)$ at point $w$) is 
\[
\E\left[[w\x]_{j\x} \left(\frac{A_{i\x j\x}[z\y]_{i\x} }{p_{i\x 
j\x}(z)}\right)^2\right] = \sum_{i \in [m], j \in [n]} \frac{A_{ij}^2 [z\y]_i^2 
[w\x]_j}{p_{ij}(z)} \geq \left(\sum_{i \in [m], j \in [n]} |A_{ij}| [z\y]_i 
\sqrt{[w\x]_j}\right)^2.
\]
Since $z\y\in\Delta^m$ and $\sqrt{w\x}\in\ball^n$ with $\norm{\sqrt{w\x}}_2 = 1$, the above lower 
bound is in the worst case $\max_i \ltwo{\ai}^2$. 
Similarly, the best possible bound on the $\yset$ is $\max_j \ltwo{\aj}^2$.
Therefore, in the $\ellone$ setup, no local estimator has parameter $L$ smaller than $\Lco$ in Table~\ref{table:L}.
Next, in the $\ell_2$-$\ell_2$ case, the ($\ell_2$) second moment of the 
$\xset$ block is 
\begin{equation*}
\E\left[\left(\frac{A_{i\x j\x}[z\y]_{i\x} }{p_{i\x j\x}(z)}\right)^2\right] = \sum_{i \in [m], j \in [n]} \frac{A_{ij}^2 [z\y]_i^2 }{p_{ij}(z)} \geq \left(\sum_{i \in [m], j \in [n]} |A_{ij}| [z\y]_i\right)^2 = \left(\sum_{i \in [m]} \norm{\ai}_1 [z\y]_i\right)^2.
\end{equation*}
In the worst case, this is at least $(\sum_{i \in [m]} \norm{\ai}_1)^2$; 
similarly, the best second moment bound for the $\yset$ block is 
$(\sum_{j \in [n]} 
\norm{\aj}_1)^2$, which means that $\Lco$ is similarly unimprovable in the $\elltwo$ setup. %

Finally, in the $\ell_2$-$\ell_1$ case, where $\xset = \ball^n$ and $\yset = \Delta^m$, we again have that the $\ell_2$ second moment of the $\xset$ (ball) block is at least
\begin{equation*}
\E\left[\left(\frac{A_{i\x j\x}[z\y]_{i\x} }{p_{i\x j\x}(z)}\right)^2\right]\ge\left(\sum_{i \in [m], j \in [n]} |A_{ij}| [z\y]_i\right)^2.
\end{equation*}
Here, since $z\y\in\Delta^m$, the worst-case lower bound of the variance 
is $\max_i \lone{\ai}^2$. Further, the local norm (at $w$) second 
moment 
of the $\yset$ (simplex) block is at least
$$\E\left[[w\y]_{i\y}\left(\frac{A_{i\y j\y}[z\x]_{j\y} }{q_{i\y j\y}(z)}\right)^2\right]\ge\left(\sum_{i \in [m], j \in [n]} |A_{ij}| [z\x]_j \sqrt{[w\y]_i}\right)^2.$$
Since $z\x \in \ball^n$ and $\sqrt{w\y}\in\ball^m$, in the worst case this 
second moment can be as high as $\opnorm{|A|}$, where we use $|A|$ to 
denote the elementwise absolute value of $A$. This is better than 
the $\Lco$ in Table~\ref{table:L},
suggesting there is room for 
improvement here. However, the sampling probabilities inducing this 
optimal variance bound are of the form
$$q_{ij}(z;w) \propto |A_{ij}| \sqrt{[w\y]_i} \cdot [z\x]_j,$$
and it unclear how to efficiently sample from this distribution. Improving 
our $\elltwoone$ gradient estimator (or proving that no improvement is 
possible) remains an open problem.
\subsection{Variance-reduced coordinate methods}
\label{sec:vr}

In this section, we develop the algorithmic framework we use in our variance-reduced methods. We first define a type of ``centered-local'' gradient estimator, modifying the local gradient estimators of the previous section. We then give the general form of a variance-reduced method and analyze it in the context of our gradient estimators and the error incurred by our data structure maintenance.

\subsubsection{General convergence result}

\begin{definition}
\label{def:vr}
For local norm setup ($\zset$, $\norm{\cdot}_{\cdot}$, 
$r$, $\Theta$, $\clip$), and given a reference point $w_0=(w_0\x,w_0\y)$, we call a stochastic gradient estimator $\tilde{g}_{w_0}:\zset\rightarrow\zset^*$ an \emph{$L$-centered-local estimator} if it satisfies the following properties:
\begin{enumerate}
\item Unbiasedness: $\E[\tilde{g}_{w_0}(z)]=g(z)$.
\item Relative variance bound: for all $w \in \zset$, 
$\E[\norm{\tilde{g}_{w_0}(z)-g(w_0)}_{w}^2] \le  L^2 V_{w_0}(z)$.
\end{enumerate}
\end{definition}

\begin{remark}\label{rem:vr-jensen}
Similarly to Lemma~\ref{lem:sublinear-implication}, a gradient mapping that admits an $L$-centered-local estimator also satisfies $\norm{g(z)-g(w_0)}^2_*\le L^2 V_{w_0}(z)$, by Jensen's inequality.
\end{remark}

Algorithm~\ref{alg:outerloop}  below is an approximate variant of the variance reduction algorithm in our earlier work~\cite{CarmonJST19} which closely builds upon the ``conceptual prox-method'' of~\citet{Nemirovski04}. The algorithm repeatedly calls a stochastic oracle $\mathcal{O}:\zset\to\zset$ to produce intermediate iterates, and then performs an extragradient (linearized) proximal step using the intermediate iterate.
The main modification compared to~\cite{CarmonJST19} is Line~\ref{line:outerloop_approx}, which accommodates slight perturbations to the extra-gradient step results. These perturbations arise due to input requirements of our data structures: we slightly pad coordinates in simplex blocks to ensure they are bounded away from zero.

\begin{algorithm}
	\label{alg:outerloop}
	\DontPrintSemicolon
	\KwInput{Target approximation quality $\vepsout$, $(\alpha,\vepsi)$-relaxed proximal oracle $\mathcal{O}(z)$ for gradient mapping $g$ and some $\vepsi < \vepsout$, distance-generating $r$}
	\Parameter{Number of iterations $K$.}
	\KwOutput{Point $\bar{z}_K$ with $\E\,\gap(\bar{z}) \le 
	\frac{\alpha\Theta}{K} + 
	\vepsout$}
	$z_0 \gets \argmin_{z\in\zset} r(z)$ \;
	\For{$k = 1, \ldots, K$}
	{
		$z_{k-1/2} \leftarrow \mathcal{O}(z_{k - 1})$ \Comment*[f]{We 
		implement $\mathcal{O}(z_{k-1})$ by calling $\InnerLoop(z_{k-1}, 
		\tilde{g}_{z_{k - 1}}, \alpha)$} \; 
		$z_k^\star \defeq \prox{z_{k-1}}{g(z_{k-1/2})}
		= \argmin_{z \in \zset}\left\{ 
		\inner{g\left(z_{k-1/2}\right)}{z} + 
		\alpha V_{z_{k-1}}(z)\right\}$\;
		$z_k\gets$ any point satisfying $V_{z_k}(u) - V_{z_k^\star}(u) \le \frac{\vepsout - \vepsi}{\alpha}$, for all $u \in \zset$ \label{line:outerloop_approx}}
	\Return $\bar{z}_K = \frac{1}{K}\sum_{k=1}^K z_{k-1/2}$
	\caption{$\OuterLoop(\mathcal{O})$ (conceptual 
	prox-method~\citep{Nemirovski04})} 
\end{algorithm}

The following definition summarizes the key property of the oracle $\mathcal{O}$.

\begin{definition}[{\citep[Definition 1]{CarmonJST19}}]
\label{def:alphaprox}
Let operator $g$ be monotone and $\alpha,\vepsi>0$. 
An \emph{($\alpha,\vepsi$)-relaxed proximal oracle} for $g$ is a 
(possibly randomized) map $\mc{O}:\zset\to\zset$ such that 
$z'=\mc{O}(z)$ satisfies
\begin{equation*}
\E\left[ \max_{u \in \zset}\big\{\inner{g(z')}{z' - u} - \alpha 
V_z(u)\big\}\right] \leq \vepsi.
\end{equation*}
\end{definition}

The following proposition, a variant of {\citep[Proposition 1]{CarmonJST19}}, shows that despite the error permitted tolerated in Line~\ref{line:outerloop_approx}, the algorithm still converges with rate $1/K$. We defer its proof to Appendix~\ref{sec:framework-proofs}.

\begin{restatable}{proposition}{restateouterloopproof}
\label{prop:outerloopproof}
Let $\mathcal{O}$ be an ($\alpha$, $\vepsi$)-relaxed proximal oracle with 
respect to gradient mapping $g$, distance-generating function $r$ with 
range at most $\Theta$ and some $\vepsi\le\vepsout$. Let $z_{1/2}, z_{3/2}, \ldots, z_{K-1/2}$ be iterates of 
Algorithm~\ref{alg:outerloop} and let $\bar{z}_K$ be its output. Then 
\begin{equation*}
\E\,\gap(\bar{z}_K)\le 
\E \max_{u\in \zset} \frac{1}{K}\sum_{k=1}^K 
\inner{g(z_{k-1/2})}{z_{k-1/2}-u}
\le 
\frac{\alpha\Theta}{K} + \vepsout.
\end{equation*}
\end{restatable}

\newcommand{\Amax}{\norm{A}_{\max}}

Algorithm~\ref{alg:innerloop-approx} is a variant of the variance-reduced inner loop 
of~\cite{CarmonJST19}, adapted for local norms and inexact iterates (again, due to approximations made by the data structure). It tolerates error in three places:
\begin{enumerate}[leftmargin=*]
	\item Instead of estimating the gradient at the previous iterate $w_{t-1}$, we estimate it at a point $\hat{w}_{t-1}$ such that
	$w_{t-1}-\hat{w}_{t-1}$ has small norm and similar divergence from the reference point $w_0$  (Line~\ref{line:inner-error-hat}).
	\item Instead of letting the next iterate be the exact mirror descent step $w_t^\star$, we let be a point $w_T$ that is close to $w_t^\star$ in norm and has similar divergences to from $w_0$ and to any any point in $\zset$ (Line~\ref{line:inner-error-star}). 
	\item The output $\tilde{w}$ can be an approximation of the average of the iterates, as long as its difference to the true average has bounded norm (Line~\ref{line:inner-average}).
\end{enumerate}

We quantify the effect of these approximations in Proposition~\ref{prop:innerloopproof}, which gives a runtime guarantee for Algorithm~\ref{alg:innerloop-approx} (where we recall the definition of $\norm{\cdot}$ as the dual norm of $\norm{\cdot}_*$, see \eqref{eq:norms_for_setups}). The proof is deferred to Appendix~\ref{sec:framework-proofs}.

\begin{algorithm}[h]
	\label{alg:innerloop-approx}
	\setstretch{1.1}
	\DontPrintSemicolon
	\KwInput{Initial $w_0\in\zset$, $L$-centered-local gradient estimator $\tilde{g}_{w_0}$,
	oracle quality $\alpha>0$}
	\Parameter{Step size $\eta$, number of iterations $T$, approximation tolerance $\epsaprx$}
	\KwOutput{Point $\tilde{w}$ satisfying Definition~\ref{def:alphaprox}}
	\For{$t = 1, \ldots, T$}
	{
	$\hat{w}_{t - 1}\approx w_{t - 1}$ satisfying (a) 
	$V_{w_0}(\hat{w}_{t-1})-V_{w_0}(w_{t-1})\le\tfrac{\epsaprx}{\alpha}$  and (b) 
	$\norm{\hat{w}_{t - 1}-w_{t - 1}}\le\tfrac{\epsaprx}{LD}$
	\;\label{line:inner-error-hat}
	
	$w_t^\star \leftarrow 
				\argmin_{w\in\zset}\left\{\inner{w}{\clip(\eta \tilde{g}_{w_0}(\hat{w}_{t - 
				1}) - \eta g(w_0)) + \eta g(w_0)}+ \frac{\alpha\eta}{2}V_{w_0}(w) + V_{w_{t - 1}}(w) 
		\right\}$\; 
		
	$w_t\approx w_t^\star$ satisfying 
	\begin{enumerate}[label=(\alph*),noitemsep,partopsep=0pt,topsep=0pt,parsep=0pt]
		\item $\max_u \left[V_{w_t}(u)-V_{w_t^\star}(u)\right]\le\eta\epsaprx$,
		\item $V_{w_0}(w_t)-V_{w_0}(w_t^\star)\le\tfrac{\epsaprx}{\alpha}$, and 
		\item $\norm{w_t - w_t^\star}\le \tfrac{\epsaprx}{2LD}$
	\end{enumerate}\label{line:inner-error-star}%
	}%
	\Return 
	$\tilde{w}\approx\frac{1}{T}\sum_{t=1}^T w_t$ satisfying $\norm{\tilde{w} - \frac{1}{T}\sum_{t=1}^T w_t} \le \tfrac{\epsaprx}{LD}$.\label{line:inner-average}
	\caption{$\InnerLoopApprox(w_0, \tilde{g}_{w_0}, \epsaprx)$}
\end{algorithm}

\begin{restatable}{proposition}{restateinnerloop}
\label{prop:innerloopproof}
	Let ($\zset$, $\norm{\cdot}_{\cdot}$, 
	$r$, $\Theta$, $\clip$) be any local norm setup.
	Let $w_0 \in \zset$, $\alpha \ge \vepsi>0$, and 
	$\tilde{g}_{w_0}$ be an $L$-centered-local estimator for some $L \ge \alpha$. Assume the domain is bounded by $\max_{z\in\zset}\|z\|\le D$, that $g$ is $L$-Lipschitz, i.e.\ $\norm{g(z)-g(z')}_*\le L\norm{z-z'}$, that $g$ is $LD$-bounded, i.e.\ $\max_{z\in\zset}\norm{g(z)}_* \le LD$, and that $\hat{w}_0 = w_0$. Then, for $\eta = \frac{\alpha}{10L^2}$, 
	$T \geq \frac{6}{\eta\alpha} \ge\frac{60L^2}{\alpha^2}$, and $\epsaprx = \frac{\vepsi}{6}$,
	Algorithm~\ref{alg:innerloop-approx} outputs a point $\hat{w} \in 
	\zset$ such that
	\begin{equation}\label{eq:innerloop-guarantee}
	\Ex{}\max\limits_{u\in\zset}
	\left[\inner{g(\tilde{w})}{\tilde{w} - u} - \alpha V_{w_0}(u)\right]
	\le \vepsi,
	\end{equation}
	i.e.\ Algorithm~\ref{alg:innerloop-approx} is an 
	$(\alpha,\vepsi)$-relaxed proximal oracle.
\end{restatable}

\begin{remark}[Assumption of boundedness on $g$]\label{rem:polylogbc} The assumption that $g$ is $LD$-bounded in the dual norm is immediate from other assumptions used in Proposition~\ref{prop:innerloopproof} in the case of the applications in Section~\ref{sec:example}, where we develop methods for solving $\ell_1$-$\ell_1$ matrix games and assume that $g(0) = 0$. In applications in Section~\ref{sec:app}, due to the existence of extra linear terms $b$, $c\neq0$, all complexity bounds will have an additional dependence on~$\log(\norm{[b;c]}_*)$ which we pay in the implementation of data structure $\AEM$ (i.e.\ the parameter $L$ in the bound on $g$ is larger if $\norm{[b;c]}_*$ is large). We hide this extra polylogarithmic factor in the $\widetilde{O}$ notation. 
\end{remark}

We also remark that (up to constants) the bounds on the range of $\vepsi\le\alpha\le L$ in the statement of Proposition~\ref{prop:innerloopproof} correspond to the cases where the inner and outer loop consist of a single iteration.

\subsubsection{Variance-reduced coordinate gradient estimators}
\label{ssec:vr-estimator}
We now state the general form which our centered-local estimators $\tilde{g}_{w_0}$ take, given a reference point $w_0\in\zset$. At a point $z$, for sampling distributions $p(z; w_0), q(z; w_0)$ to be specified, sample $i\x, j\x \sim p(z; w_0)$ and $i\y, j\y \sim q(z; w_0)$. Then, define 
\begin{equation}
\label{eq:estimate-vr}
\begin{aligned}
\tilde{g}_{w_0}(z)= \left(\frac{A_{i\x j\x}[z\y-w_0\y]_{i\x} }{p_{i\x j\x}(z;w_0)}e_{j\x},\frac{-A_{i\y j\y}[z\x-w_0\x]_{j\y} }{q_{i\y j\y}(z;w_0)}e_{i\y}\right)+g(w_0).
\end{aligned}
\end{equation}

It is clear that regardless of the distributions $p(z; w_0), q(z; w_0)$, this is an unbiased gradient estimator ($\E[\tilde{g}_{w_0}(z)] = g(z)$). Furthermore, $\tilde{g}_{w_0}(z) - g(w_0)$ is always 2-sparse.

\section{Matrix games}
\label{sec:example}

In this section we instantiate the algorithmic framework of  Section~\ref{sec:framework} in $\ellone$ setup without linear terms, i.e.\ $b=c=0$ in the objective~\eqref{eq:sublinear-problem}. This is the fundamental ``matrix game'' problem
\[\min_{x\in\Delta^m}\max_{y\in\Delta^n} y^\top Ax.\] 

We give two algorithms for approximately solving matrix games. In Section~\ref{ssec:l1l1sub} we develop a stochastic coordinate method based on Algorithm~\ref{alg:sublinear} with potentially sublinear runtime $\Otil{(\looco/\epsilon)^2}$. In Section~\ref{app:sublinear-proofs} we develop a coordinate variance-reduction based on Algorithm~\ref{alg:outerloop} with runtime $\Otil{\nnz + \sqrt{\nnz}\cdot \looco/\epsilon}$ that improves on the former runtime whenever it is $\Omega(\nnz)$. In both cases we have
\begin{equation}\label{eq:l11def}
\looco \defeq \max\left\{ \max_i \norm{\ai}_2, \max_j \norm{\aj}_2\right\}
\end{equation}
as in Table~\ref{table:L}. 

Instantiations for the $\elltwoone$ and $\elltwo$ setups follow similarly. We carry them out in Appendices~\ref{app:sublinear-proofs} (for stochastic coordinate methods) and~\ref{app:vr-proofs} (for variance reduction methods).

\begin{remark}\label{rem:linear-terms}
	For simplicity in this section (and the remaining implementations in Appendices~\ref{app:sublinear-proofs},~\ref{app:vr-proofs}), we will set $g(0) = 0$ whenever the setup is not $\ell_2$-$\ell_2$, as is standard in the literature. We defer a discussion of how to incorporate arbitrary linear terms in simplex domains to Section~\ref{sec:app}; up to additional logarithmic terms in the runtime, this extension is supported by $\AEM$.
\end{remark}

\paragraph{Assumptions.}
Throughout (for both Sections~\ref{ssec:l1l1sub} and~\ref{ssec:vr-l1}), we assume access to entry queries, $\ell_2$ norms of rows and columns, 
and $\ell_2$ sampling distributions for all rows and columns. We use the 
$\ell_1$-$\ell_1$ local norm setup (Table~\ref{tab:local-norm}). We also define $L_{\max}\defeq\norm{A}_{\max} = \max_{i \in [m], j \in [n]} |A_{ij}|$.

\subsection{$\ellone$ sublinear coordinate method}
\label{ssec:l1l1sub}

\subsubsection{Gradient estimator}
For $z\in\Delta^n\times\Delta^m$ and desired accuracy $\eps>0$, we specify the sampling distributions $p(z), q(z)$:
\begin{equation}\label{eq:l1-prob-def}
p_{ij}(z)\defeq[z\y]_i\frac{A_{ij}^2}{\norm{\ai}_2^2}~~\mbox{and}~~\ 
q_{ij}(z)\defeq[z\x]_j\frac{A_{ij}^2}{\norm{\aj}_2^2}.
\end{equation}

We first state and prove the local properties of this estimator.

\begin{restatable}{lemma}{restateestproplone}
\label{lem:est-prop-l1}
In the $\ellone$ setup, estimator 
(\ref{eq:estimate-l1}) using the sampling distribution in~(\ref{eq:l1-prob-def}) is a\\
 $\sqrt{2}\looco$-local estimator. 
 \end{restatable}

\begin{proof}
	Unbiasedness holds by definition. For arbitrary $w\x$, we have the variance bound:
	\begin{align*}
	\E\left[\norm{\tilde{g}\x(z)}_{w\x}^2\right]
	& \le  
	\sum_{i\in[m],j\in[n]} p_{ij}(z)\cdot\left( [w\x]_j\cdot \left(\frac{A_{ij}[z\y]_i}{p_{ij}(z)}\right)^2 \right) = \sum\limits_{i\in[m],j\in[n]}[w\x]_j \frac{A_{ij}^2[z\y]_i^2}{p_{ij}(z)}\\
	& \le \sum\limits_{i\in[m],j\in[n]}[w\x]_j[z\y]_i\norm{\ai}_2^2
	\le \max_{i\in[m]}\norm{\ai}_2^2 \leq(L_{1,1}^{\textup{co}})^2.
	\end{align*}
	
	Similarly, we have
	\begin{equation*}
	\E\left[\norm{\tilde{g}\y(z)}_{w\y}^2\right] \le (L_{1,1}^{\textup{co}})^2.
	\end{equation*}
	The definition $\norm{\tilde{g}(z)}^2_w = \norm{\tilde{g}\x(z)}^2_{w\x}+\norm{\tilde{g}\y(z)}^2_{w\y}$ yields the claimed variance bound.
\end{proof}

\subsubsection{Implementation details}
\label{ssec:implementsublone}
In this section, we discuss the details of how to leverage the $\IM_1$ data structure to implement the iterations of our algorithm. The algorithm we analyze is Algorithm~\ref{alg:sublinear}, using the local estimator defined in \eqref{eq:estimate-l1}, and the distribution \eqref{eq:l1-prob-def}. We choose
\[\eta = \frac{\epsilon}{18\loocop^2} \text{ and } T= \left\lceil\frac{6\Theta}{\eta\eps}\right\rceil\ge\frac{108\loocop^2 \log(mn)}{\epsilon^2}.\]
Lemma~\ref{lem:est-prop-l1} implies that our estimator satisfies the remaining requirements for Proposition~\ref{prop:sublinear}, giving the duality gap guarantee in $T$ iterations. In order to give a runtime bound, we claim that each
iteration can be implemented in $\log(mn)$ time, with $O(m + n)$ 
additional runtime. 

\paragraph{Data structure initializations and invariants.}
At the start of the algorithm, we spend $O(m + n)$ time 
initializing data structures via $\IMS_1\x.\Initialize(\frac{1}{n}\1_n, 
\mathbf{0}_n)$ and $\IMS_1\y.\Initialize(\frac{1}{m}\1_m, \mathbf{0}_m)$, where $\IMS_1\x, \IMS_1\y$ are appropriate instantiations of $\IM_1$ data structures. Throughout, we preserve the invariant that the points maintained by $\IMS_1\x, \IMS_1\y$ correspond to the $x$ and $y$ blocks of the current iterate $z_t$ at iteration $t$ of the algorithm.

\paragraph{Iterations.}
For simplicity, we only discuss the runtime of updating the $x$ block as 
the $y$ block follows symmetrically. We divide each iteration into the 
following substeps, each of which we show runs in time $O(\log mn)$. We 
refer to the current iterate by $z = (z\x, z\y)$, and the next iterate by $w = 
(w\x, w\y)$. \\

\noindent
\emph{Sampling.}
Recall that
$$p_{ij}(z) \defeq [z\y]_i\frac{A_{ij}^2}{\norm{\ai}_2^2}.$$
We first sample coordinate $i$ via $\IMS_1\y.\Sample()$ in $O(\log m)$.
Next, we sample $j \in [n]$ with probability proportional to $A_{ij}^2$ using the data structure corresponding to $\ai$ in $O(1)$ by assumption of the matrix access model.\\

\noindent
\emph{Computing the gradient estimator.} To compute $c \defeq \clip (A_{ij}[z\y]_i/p_{ij})$, it suffices to compute $A_{ij}$, $[z\y]_i$, and $p_{ij}$. Using an entry oracle for $A$ we obtain $A_{ij}$, and we get  $[z\y]_i$ by calling $\IMS_1\y.\Get(i)$. Computing $p_{ij}$ using the precomputed $\ltwo{\ai}$ and the values of $A_{ij}, [z\y]_i$ therefore takes $O(1)$ time.\\

\noindent
\emph{Performing the update.} For the update corresponding to a proximal step, we have
$$w\x\leftarrow \Pi_{\xset}\left(z\x\circ\exp(-\eta \tilde{g}\x(z))\right)=\frac{z\x\circ\exp(-\eta \tilde{g}\x(z))}{\lones{z\x\circ\exp(-\eta \tilde{g}\x(z))}}. $$
We have computed $\tilde{g}\x(z)$, so to perform this update, we 
call
\begin{align*}
& \xi \gets\IMS_1\x.\Get(j);\\
& \IMS_1\x.\AddSparse(j, (\exp(-\eta c) - 1)\xi);\\
& \IMS_1\x.\Scale(\IM\x.\GetNorm()^{-1});\\
& \IMS_1\x.\SumUp().
\end{align*}
By assumption, each operation takes time $O(\log n)$, giving the desired iteration 
complexity. It is clear that at the end of performing these operations, the invariant that $\IMS_1\x$ maintains the $x$ block of the iterate is preserved.

\paragraph{Averaging.}
After $T$ iterations, we compute the average point $\bar{z}\x$:
$$
[\bar{z}\x]_j\gets\frac{1}{T}\cdot\IMS_1\x.\GetSum(j),\forall j\in[n].
$$ 
By assumption, this takes $O(n)$ time. 

\subsubsection{Algorithm guarantee}

\begin{theorem}
\label{thm:l1l1-sublinear}
In the $\ellone$ setup, the implementation in Section~\ref{ssec:implementsublone} has runtime
\begin{equation*}
O\left(\frac{\loocop^2\log^2(mn)}{\epsilon^2} + m + n \right),
\end{equation*}
and outputs a point $\bar{z} \in \zset$ such that
$\E \gap(\bar{z})
\le \epsilon.
$.
\end{theorem}

\begin{proof}
The runtime follows from the discussion in Section~\ref{ssec:implementsublone}. The correctness follows from Proposition~\ref{prop:sublinear}.
\end{proof}

\begin{remark}
Using our $\IM_1$ data structure, the $\ellone$ algorithm of~\citet{GrigoriadisK95} runs in time  $O(\rcs 
\norm{A}_{\max}^2\log^2(mn)/\eps^2)$, where $\rcs$ is the maximum 
number of nonzeros in any row or column. Our runtime universally 
improves upon it since $(L_{1,1}^{\textup{co}})^2\le\rcs\|A\|_{\max}^2$.	
\end{remark}

\subsection{$\ell_1$-$\ell_1$ variance-reduced coordinate method}
\label{ssec:vr-l1}

\subsubsection{Gradient estimator}
\label{ssec:l1-est}

Given reference point $w_0 \in \Delta^n \times \Delta^m$, for $z\in\Delta^n \times \Delta^m$ and a parameter $\alpha>0$, we specify the sampling distributions $p(z; w_0), q(z; w_0)$:

\begin{equation}\label{eq:vr-l1-prob-def}
\begin{aligned}
p_{ij}(z;w_0) & \defeq \frac{[z\y]_i+2[w_0\y]_i}{3}\cdot\frac{A_{ij}^2}{\ltwo{\ai}^2}
~~\mbox{and}~~
q_{ij}(z;w_0) & \defeq \frac{[z\x]_j+2[w_0\x]_j}{3} \cdot\frac{A_{ij}^2}{\ltwo{\aj}^2}.
\end{aligned}
\end{equation}
We remark that this choice of sampling distribution, which we term ``sampling from the sum'' (of the current iterate and reference point), may be viewed as a computationally-efficient alternative to the distribution specified in \cite{CarmonJST19}, which was based on ``sampling from the difference''. In particular, sampling from the difference is an operation which to the best of our knowledge is difficult to implement in sublinear time, so we believe that demonstrating that this alternative distribution suffices may be of independent interest. In order to show its correctness, we need the following claim, whose proof we defer to Appendix~\ref{ssec:helpervr}.

\begin{restatable}{lemma}{restatelocalnorms}\label{lem:local-norms}
	For $y,y'\in\Delta^m$, divergence $V_{y}(y')$ generated by $r(y) = \sum_{i \in [m]} [y]_i \log [y]_i - [y]_i$
	satisfies 
	\begin{equation*}
	V_{y}(y') \ge \half \norm{y'-y}^2_{\frac{3}{2y+y'}} =
	\half \sum_{i\in[m]} \frac{([y]_i - [y']_i)^2}{\frac{2}{3}[y]_i + 
		\frac{1}{3}[y']_i}.
	\end{equation*}
\end{restatable}

We now show the local properties of this estimator.

\begin{lemma}
	\label{lem:est-prop-vr-l1}
	In the $\ellone$ setup, estimator \eqref{eq:estimate-vr} using the sampling distribution in~\eqref{eq:vr-l1-prob-def} is a \\
	$\sqrt{2}\looco$-centered-local estimator.
\end{lemma}

\begin{proof}
	Unbiasedness holds by definition. For arbitrary $w\x$, we have the variance bound:
	\begin{align*}
	\E\left[\norm{\tilde{g}\x_{w_0}(z)-g\x(w_0)}_{w\x}^2\right]
	& = 
	\sum_{i\in[m],j\in[n]} p_{ij}(z;w_0)\cdot\left( [w\x]_j\cdot \left(\frac{A_{ij}\left([z\y]_i-[w_0\y]_i\right)}{p_{ij}(z;w_0)}\right)^2 \right)\\
	& = \sum\limits_{i\in[m],j\in[n]}[w\x]_j \frac{A_{ij}^2\left([z\y]_i-[w_0\y]_i\right)^2}{p_{ij}(z;w_0)}\\
	& \le \sum\limits_{i\in[m],j\in[n]}[w\x]_j\frac{([z\y]_i-[w_0\y]_i)^2}{\frac{1}{3}[z\y]_i+\frac{2}{3}[w_0\y]_i}\ltwo{\ai}^2\\
	& \le 2\left(\max_i\ltwo{\ai}^2\right)V_{w_0\y}(z\y),
	\end{align*}
	where in the last inequality we used Lemma~\ref{lem:local-norms}. Similarly, we have for arbitrary $w\y$,
	\begin{equation*}
	\E\left[\norm{\tilde{g}\y_{w_0}(z)-g\y(w_0)}_{w\y}^2\right] \le 2\left(\max_j\ltwo{\aj}^2\right)V_{w_0\x}(z\x).\end{equation*}
	Combining these and using
	\begin{align*}\norm{\tilde{g}_{w_0}(z)-g(w_0)}^2_w &\defeq\norm{\tilde{g}\x_{w_0}(z)-g\x(w_0)}^2_{w\x}+\norm{\tilde{g}\y_{w_0}(z)-g\y(w_0)}^2_{w\y}\end{align*}
	yields the desired variance bound. 
\end{proof}

\subsubsection{Implementation details}
\label{ssec:implementvrlone}
In this section, we discuss the details of how to leverage the $\AEM$ data structure to implement the iterations of our algorithm. We first state one technical lemma on the effect of $\beta$-padding (Definition~\ref{def:betastable}) on increasing entropy, used in conjunction with the requirements of Proposition~\ref{prop:innerloopproof} to bound the error tolerance required by our $\AEM$ data structure. The proof is deferred to Appendix~\ref{ssec:helpervr}.
\begin{restatable}{lemma}{restatestableentropy}\label{lem:stable}
	Let $x' \in \Delta^n$ be a $\beta$-padding of $x \in \Delta^n$. Then,
	\[\sum_{j \in [n]} x'_j \log x'_j - \sum_{j \in [n]} x_j \log x_j \le \frac{\beta n}{e}+\beta(1+\beta).\]
\end{restatable}
This leads to the following divergence bounds which will be used in this section.
\begin{lemma}\label{lem:stablediv}
	Let $x' \in \Delta^n$ be a $\beta$-padding of $x \in \Delta^n$. Then 
	\[
	V_{x'}(u)-V_{x}(u)\le \beta,\;\forall u\in\zset
	\]
	and if $\linf{\log(x_0)}\le M$, then
	\[
	V_{x_0}(x')-V_{x_0}(x)\le \beta \left(2M+\frac{n}{e}+1+\beta\right) 
	\]
\end{lemma}
\begin{proof}
	Throughout this proof, let $\tx$ be the point in Definition~\ref{def:betastable} such that $\norm{\tx - x}_1 \le \beta$ and $x' = \tx/\norm{\tx}_1$. 
	
	The first claim follows from expanding
	\[V_{x'}(u) - V_x(u) = \sum_{j \in [n]} u_j\log \frac{x_j}{x'_j} = \sum_{j \in [n]} u_j\log\left( \frac{x_j}{\tx_j} \cdot \norm{\tx}_1\right) \le \log(\norm{\tx}_1) \le \beta.\]
	The first inequality used $u \in \Delta^n$ and $\tx \ge x$ entrywise, and the last inequality used $\log(1 + \beta) \le \beta$.
	
	For the second claim, we have by the triangle inequality
	\[\norm{x - x'}_1 \le \norm{x - \tx}_1 + \norm{\tx - x'}_1 \le \beta + \left(\norm{\tx}_1 - 1\right)\norm{x'}_1 \le 2\beta.\]
	The claim then follows from expanding
	\[V_{x_0}(x') - V_{x_0}(x) = \sum_{j \in [n]} x'_j \log x'_j - \sum_{j \in [n]} x_j \log x_j + \inner{\log x_0}{x - x'},\]
	and applying Lemma~\ref{lem:stable}. 
\end{proof}
The algorithm we analyze is Algorithm~\ref{alg:outerloop} with $K = 3\alpha\Theta/\eps$, $\vepsout=2\eps/3$, $\vepsi = \eps/3$ using Algorithm~\ref{alg:innerloop-approx} as an $(\alpha, \vepsi)$-relaxed proximal oracle. The specific modification we perform to define the iterates $\{z_k\}$ of Algorithm~\ref{alg:outerloop} as modifications of the ideal iterates $\{z_k^\star\}$ uses the following definition.

\begin{definition}
	For a simplex variable $x' \in \Delta^n$, we define $\truncate(x', \delta)$ to be the point $x \in \Delta^n$ with $x_j \propto \max(x'_j, \delta)$ for all $j \in [n]$. For a variable $z$ on two blocks, we overload notation and define $\truncate(z, \delta)$ to be the result of applying $\truncate(\cdot, \delta)$ to each simplex block of $z$.
\end{definition}
For our implementation, in each step of Algorithm~\ref{alg:outerloop}, we will compute the point $z_k^\star$ exactly, and apply the operation $z_k\gets\mathsf{truncate}(z_k^{\star},\delta)$, for $\delta=\tfrac{\vepsout-\vepsi}{\alpha(m+n)}$.
We now quantify the effect of truncation in terms of Bregman divergence to an arbitrary point.

\begin{lemma}[Effect of truncation]
	\label{lem:projecteffect}
	Let $x'\in\Delta^n$, and let $x = \truncate(x', \delta)$. Then, for any $u\in\Delta^n$, and where divergences are with respect to entropy,
	\begin{equation*}
	V_{x}(u) - V_{x'}(u) \le \delta n.
	\end{equation*}
\end{lemma}
\begin{proof} Note that $x$ is a $\delta n$-padding of $x'$, as it is the result of adding at most $\delta$ to each coordinate and renormalizing to lie in the simplex. Consequently, the result follows from Lemma~\ref{lem:stablediv}.
\end{proof}

Lemma~\ref{lem:projecteffect} thus implies our iterates satisfy the requirements of Algorithm~\ref{alg:outerloop}. Our implementation of Algorithm~\ref{alg:innerloop-approx} will use approximation tolerance $\epsaprx=\vepsi/6 = \eps/18$, where we always set \begin{equation}\label{eq:alphaboundsoo}\looco \ge \alpha \ge \vepsi.\end{equation} This matches the requirements of Proposition~\ref{prop:innerloopproof}. 	In the implementation of Algorithm~\ref{alg:innerloop-approx}, we use the centered-local gradient estimator defined in \eqref{eq:estimate-vr}, using the sampling distribution \eqref{eq:vr-l1-prob-def}. For each use of Algorithm~\ref{alg:innerloop-approx}, we choose 
\begin{equation*}\label{eq:etatchoices}
\begin{aligned}\eta &=
 \frac{\alpha}{20\loocop^2 } \text{ and } T =
\left\lceil\frac{6}{\eta\alpha}\right\rceil \ge \frac{120\loocop^2}{\alpha^2}.
\end{aligned}\end{equation*}

Our discussion will follow in four steps: first, we discuss the complexity of all executions in Algorithm~\ref{alg:outerloop} other than calls to the oracles. Next, we discuss the complexity of all initializations of $\AEM$ data structures. Then, we discuss the complexity of all other iterations of Algorithm~\ref{alg:innerloop-approx}. For simplicity, when discussing Algorithm~\ref{alg:innerloop-approx}, we will only discuss implementation of the $x$-block, and the $y$-block will follow symmetrically, while most runtimes are given considering both blocks. Lastly, we discuss complexity of computing the average iterate in the end of the inner loop. Altogether, the guarantees of Proposition~\ref{prop:outerloopproof} and Proposition~\ref{prop:innerloopproof} imply that if the guarantees required by the algorithm hold, the expected gap of the output is bounded by $\eps$.

\paragraph{Outer loop extragradient steps.}
Overall, we execute $K = 3\alpha\Theta/\eps$ iterations of Algorithm~\ref{alg:outerloop}, with $\vepsout = 2\eps/3$, $\vepsi=\eps/3$ to obtain the desired gap, where $\Theta = \log(mn)$ in the $\ellone$ setup. We spend $O(\nnz)$ time executing each extragradient step in Algorithm~\ref{alg:outerloop} exactly to compute iterates $z_k^\star$, where the dominant term in the runtime is computing each $g(z_{k - 1/2})$, for $k \in [K]$. We can maintain the average point $\bar{z}$ throughout the duration of the algorithm, in $O(m + n)$ time per iteration. Finally, we spend an additional $O(m + n)$ time per iteration applying $\truncate$ to each iterate $z_k^{\star}$.

\paragraph{Data structure initializations and invariants.} 

We consider the initialization of data structures for implementing an $(\alpha,\vepsi = \eps/3)$-relaxed proximal oracle with error tolerance $\epsaprx=\eps/18$. First, note that the point $w_0\x$ used in the initialization of every inner loop, by the guarantees of $\truncate$ operation, has no two coordinates with multiplicative ratio larger than $\delta = \eps/(3\alpha(m + n)) \ge (m + n)^{-4}$, by our choice $\alpha \le \looco$ \eqref{eq:alphaboundsoo} and our assumptions on $\looco/\eps$ (cf.\ Section~\ref{ssec:problemdef}). Since clearly a simplex variable in $\Delta^n$ has a coordinate at least $1/n$, the entries of $w_0$ are lower bounded by $\lambda = (m + n)^{-5}$.

Next, we discuss the initial parameters given to $\AEMS\x$, an instance of $\AEM$ which will support necessary operations for maintaining the $x$ variable (we will similarly initialize an instance $\AEMS\y$). Specifically, the invariant that we maintain throughout the inner loop is that in iteration $t$, the``exact vector'' $x$ maintained by $\AEMS\x$ corresponds to the $x$ block of the current iterate $w_t$, and the ``approximate vector'' $\hat{x}$ maintained corresponds to the $x$ block of the approximate iterate $\hat{w}_t$, as defined in Algorithm~\ref{alg:innerloop-approx}.

We will now choose $\tveps$ so that if $\AEMS\x$ is initialized with error tolerance $\tveps$, all requirements of Proposition~\ref{prop:innerloopproof} (e.g.\ the bounds stipulated in Algorithm~\ref{alg:innerloop-approx}) are met. We first handle all divergence requirements. In a given iteration, denote the $x$ blocks of $w_t^\star $, $w_t$ and $\hat{w}_t$ by $x_t^\star $, $x_t$ and $\hat{x}_t$ respectively, and recall $\AEMS\x$ guarantees $x_t$ is a $\tveps$-padding of $x_t^\star$, and $\hat{x}_t$ is a $\tveps$-padding of $x_t^\star $. The former of these guarantees is true by the specification of $\MultSparse$ (which will be used in the implementation of the step, see ``Performing the update'' below), and the latter is true by the invariant on the points supported by $\AEMS\x$. Lines~\ref{line:inner-error-hat} and~\ref{line:inner-error-star} of Algorithm~\ref{alg:innerloop-approx} stipulate the divergence requirements, where $x_0 \defeq w_0\x$,
\begin{align}
	\max\left\{V_{x_0}(x_t) - V_{x_0}\left(x_t^\star \right),V_{x_0}\left(\hat{x}_t\right) - V_{x_0}\left(x_t\right)\right\}&  \le \frac{\epsaprx}{2\alpha} = \frac{\eps}{36\alpha}\label{cond:one-side}\\
	~~\text{and}~~~& \nonumber\\
	\max_u \left[V_{x_t}(u)-V_{x_t^\star }(u)\right] & \le
	\frac{\eta\epsaprx}{2}=\frac{\eta\eps}{36}.\label{cond:other-side}
\end{align}
Clearly, combining this guarantee with a similar guarantee on the $y$ blocks yields the desired bound. Since we derived $\norm{\log w_0}_\infty \le 5\log(mn)$, we claim that choosing
\[\tveps \le \frac{\eps}{36\alpha(m + n)}\]
suffices for the guarantees in~\eqref{cond:one-side}. By the first part of Lemma~\ref{lem:stablediv}, for all sufficiently large $m+n$,
\begin{align*}\max\left\{V_{x_0}(x_t) - V_{x_0}\left(x_t^\star \right),V_{x_0}\left(\hat{x}_t\right) - V_{x_0}\left(x_t\right)\right\}\le \tveps\left(10\log(mn) + \frac{n}{e} + 1 + \tveps\right) \le \frac{\eps}{36\alpha}.
\end{align*}
Similarly for guarantees in~\eqref{cond:other-side}, by the second part of Lemma~\ref{lem:stablediv} we know it suffices to choose 
\[\tveps\le \frac{\eps^2}{720\loocop^2} \le \frac{\eps\alpha}{720\loocop^2 } =\frac{\eta\eps}{36}.\] 
Here, we used the restriction $\vepsi \le \alpha \le \looco$. Next, the norm requirements of Algorithm~\ref{alg:innerloop-approx} (the guarantees in Lines~\ref{line:inner-error-hat},~\ref{line:inner-error-star}, and~\ref{line:inner-average}) imply we require
\[\tveps \le \frac{\eps}{18\sqrt{2}\looco},\] where we used that $g$ is $\norm{A}_{\max}\le \looco$-Lipschitz and the diameter of $\zset$ is bounded by $\sqrt{2}$. Using our assumptions on the size of parameters in Section~\ref{ssec:problemdef}, it suffices to set the error tolerance 
\[\tveps = (m + n)^{-8}.\] 
To give the remainder of specified parameters, $\AEMS\x$ is initialized via $\Initialize(w_0\x,v,\kappa,\tilde{\veps})$ for
\begin{align*}\kappa \defeq \frac{1}{1 + \eta\alpha/2},\; v \defeq (1 - \kappa) \log w_0\x - \eta\kappa g\x(w_0). \end{align*}
To motivate this form of updates, note that each iteration of Algorithm~\ref{alg:innerloop-approx} requires us to compute
\[\argmin\left\{\inner{c_t e_j + g\x(w_0)}{x} + \frac{\alpha}{2}V_{w_0\x}(x) + \frac{1}{\eta} V_{w\x_{t}}(x) \right\}.\]
We can see that the solution to this update is given by
\begin{equation}\label{eq:updateform}\left[w_{t + 1}^\star \right]\x \gets \Pi_{\Delta}\left([w_t\x]^\kappa \circ 
\exp\left((1 - \kappa)\log w_0\x - \eta\kappa g\x(w_0)\right) \circ 
\exp(-\eta\kappa c_t e_j)\right).\end{equation}
This form of update is precisely supported by our choice of $\kappa$ and $v$, as well as the $\DenseStep$ and $\MultSparse$ operations. By the choice of parameters, we note that $1 - \kappa \ge (m + n)^{-8}$.

Finally, in order to support our sampling distributions and gradient computations, we compute and store the vectors $w_0$ and $g(w_0)$ in full using $O(\nnz(A))$ time at the beginning of the inner loop. In $O(m + n)$ time, we also build two data structures which allow us to sample from entries of the given fixed vectors $w_0\x$, and $w_0\y$, in constant time respectively. 

Following Section~\ref{ssec:interface-simplex}, we defined the parameter
\[\omega \defeq \max\left(\frac{1}{1 - \kappa},\;\frac{n}{\lambda\tilde{\eps}}\right) \le (m+n)^{13},\text{ so that }\log(\omega)=O(\log(mn)).\]
Altogether, these initializations take time $O(\nnz + (m + n)\log^3(mn))$, following Section~\ref{ssec:interface-simplex}.

\paragraph{Inner loop iterations.} We discuss how to make appropriate modifications to the $x$-block. For simplicity we denote our current iterate as $z$, and the next iterate as $w$. Also, we denote $\hat{z}$ as the concatenation of implicit iterates that the two $\AEM$ copies maintain (see Section~\ref{ssec:interface-simplex} for more details), which is $\tilde{\veps}$ close in $\ell_1$ distance to $z$, the prior iterate, so that we can query or sample entries from $\hat{z}$ using $\AEMS\x$ and $\AEMS\y$. Each inner loop iteration  consists of using a gradient estimator at $\hat{z}$ satisfying $\norm{\hat{z}-z}_1\le\tilde{\veps}$, sampling indices for the computation of $\tilde{g}_{w_0}(\hat{z})$, computing the sparse part of $\tilde{g}_{w_0}(\hat{z})$, and performing the approximate update to the iterate. We show that we can run each substep using data structure $\AEMS\x$ in time $O(\log^4(mn))$, within the error tolerance of Proposition~\ref{prop:innerloopproof} due to the definition of $\tilde{\veps}$. Combining with our discussion of the complexity of initialization, this implies that the total complexity of the inner loop, other than outputting the average iterate, is
\[O(T\log^4(mn) + \nnz + (m + n)\log^3(mn)) = O\left(\frac{\loocop^2\cdot\log^4(mn)}{\alpha^2}+ \nnz + (m+n)\log^3(mn)\right).\]

\noindent
\emph{Sampling.} Recall that the distribution we sample from is given by 
\begin{equation*}
\begin{aligned}
	p_{ij}(\hat{z};w_0)  \defeq \frac{[\hat{z}\y]_i+2[w_0\y]_i}{3}\cdot\frac{A_{ij}^2}{\ltwo{\ai}^2}.
\end{aligned}
\end{equation*}
First, with probability $2/3$, we sample a coordinate $i$ from the precomputed data structure for sampling from $w_0\y$ in constant time; otherwise, we sample $i$ via $\AEMS\y.\Sample()$. Then, we sample an entry of $\ai$ proportional to its square via the precomputed data structure (cf.\ Section~\ref{ssec:accessmodel}) in constant time. This takes in total $O(\log mn)$ time. \\

\noindent
\emph{Computing the gradient estimator.} Proposition~\ref{prop:innerloopproof} requires us to compute the sparse component of the gradient estimator \eqref{eq:estimate-vr} at point $\hat{z}$. To do this for the $x$ block, we first query $[w_0\x]_j$ and $[\hat{z}\y]_i\leftarrow\AEMS\y.\Get(i)$, and then access the precomputed norm $\norm{\ai}_2$ and entry $A_{ij}$. We then compute 
\[c=\clip\left(A_{ij}\left[\hat{z}\y - w_0\y\right]_i \cdot \frac{3}{[\hat{z}\y]_i+2[w_0\y]_i}\cdot\frac{\ltwo{\ai}^2}{A_{ij}^2}\right).\]
By the guarantees of $\AEM$, this takes total time bounded by $O(\log(mn))$.\\
 
\noindent
\emph{Performing the update.} To perform the update, by observing the form of steps in Algorithm~\ref{alg:innerloop-approx} with our choice of entropy regularizer, the update form given by the regularized mirror-descent step is (as derived in the discussion of the initialization of $\AEMS\x$, see \eqref{eq:updateform})
\begin{equation*}\label{eq:xdense-update}
{[w^\star}\x \gets \Pi_{\Delta}((w\x)^\kappa \circ 
\exp((1-\kappa)\log w_0\x - \eta\kappa g\x(w_0)-\eta\kappa c e_j)).
\end{equation*}
To implement this, recalling our choice of the vector $v$ in the initialization of $\AEMS\x$, it suffices to call
\begin{align*}
& \AEMS\x.\DenseStep(); \\
& \AEMS\x.\MultSparse(- \eta \kappa c e_j);\\
& \AEMS\x.\SumUp().
\end{align*}
By assumption, each operation takes time bounded by $O(\log^4(mn))$, where we note the vector used in the $\MultSparse$ operation is $1$-sparse. The implementation of this update is correct up to a $\tveps$-padding, whose error we handled previously. By the discussion in the data structure initialization section, this preserves the invariant that the $x$ block of the current iterate is maintained by $\AEMS\x$.

\paragraph{Average iterate computation.}
At the end of each run of Algorithm~\ref{alg:innerloop-approx}, we compute and return the average iterate via calls $\AEMS\x.\GetSum(j)$ for each $j \in [n]$, and scaling by $1/T$, and similarly query $\AEMS\y$. The overall complexity of this step is $O((m + n)\log^2(mn))$. The correctness guarantee, i.e.\ that the output approximates the average in $\ell_1$ norm up to $\epsaprx/LD$, is given by the choice of $\tilde{\veps}$ and the guarantees of $\AEM$, where in this case the domain size $D$ is bounded by $\sqrt{2}$. This is never the dominant factor in the runtime, as it is dominated by the cost of initializations.

\subsubsection{Algorithm guarantee}
\label{sssec:algol1l1}

\begin{theorem}
	\label{thm:l1l1}
	In the $\ellone$ setup,  let $\tnnz \defeq \nnz + (m+n)\log^3(mn)$. The implementation in Section~\ref{ssec:implementvrlone} with the optimal choice of $\alpha = \max\left(\epsilon/3, \looco\log^{2}\left(mn\right)/\sqrt{
	\tnnz}\right)$ has runtime
	\begin{align*}O\left(\left(\tnnz + \frac{\loocop^2\log^4(mn)}{\alpha^2}\right)\frac{\alpha\log(mn)}{\eps}\right) = O\left(\tnnz + \frac{\sqrt{\tnnz} \looco \log^{3}(mn)}{\epsilon} \right)\end{align*}
	and outputs a point $\bar{z} \in \zset$ such that
	\[\E\gap(\bar{z}) \le \epsilon.\]
\end{theorem}

\begin{proof}
The correctness of the algorithm is given by the discussion in Section~\ref{ssec:implementvrlone} and the guarantees of Proposition~\ref{prop:outerloopproof} with $K=3\alpha\Theta/\eps$, $\vepsout=2\eps/3$, $\vepsi=\eps/3$, Proposition~\ref{prop:innerloopproof} with $\epsaprx=\eps/15$, and the data structure $\AEM$ with our choice of \[\tilde{\veps}\defeq (m + n)^{-8},\] to meet the approximation conditions in Line~\ref{line:inner-error-hat},~\ref{line:inner-error-star} and~\ref{line:inner-average} of Algorithm~\ref{alg:innerloop-approx}. The runtime bound is given by the discussion in Section~\ref{ssec:implementvrlone}, and the optimal choice of $\alpha$ is clear.
\end{proof}

\section{Data structure implementation}
\label{sec:ds}

In this section, we give implementations of our data structures, 
fulfilling the interface and runtime guarantees of  
Section~\ref{ssec:interfaces}. In Section~\ref{sec:iterate_maintain_proof} we provide the implementation of $\IM_p$ for $p\in\{1,2\}$ used for sublinear coordinate methods. In Section~\ref{ssec:ds-globalm}, we provide an implementation of $\AEM$ used in variance-reduced coordinate methods for simplex domains, provided we have an implementation of a simpler data structure, $\ScM$, which we then provide in Section~\ref{sec:scm}.%

\subsection{$\IM_p$}
\label{sec:iterate_maintain_proof}

The $\IM_p$, $p \in \{1, 2\}$ data structure is
described in Section~\ref{ssec:interface-norm} and used for tracking the 
iterates in our fully stochastic methods and the Euclidean part of our the 
iterates in our variance-reduced methods. The data structure maintains an 
internal representation of $x$, the current iterate, and $s$, a running sum 
of all iterates.
The main idea behind the efficient implementation of the data structure is 
to maintain $x$ and $s$ as a linear combination of sparsely-updated 
vectors. In particular, the data structure has the following state: 
scalars $\xi_u$, $\xi_v$, $\sigma_u$, $\sigma_v$, $\iota$, $\nu$; vectors 
$u, u', v$, and the scalar  $\norm{v}_2^2$; the vector $v$ is only relevant 
for variance reduction and is therefore set 0 for the non-Euclidean case 
$p=1$. 

We maintain the following invariants on the data structure state at the end 
of every operation:
\begin{itemize}
	\item $x = \xi_u u + \xi_v v$, the internal representation of $x$
	\item $s = u' + \sigma_u u + \sigma_v v$, the internal representation of 
	running sum $s$
	\item $\iota = \inner{x}{v}$, the inner product of the iterate with fixed vector $v$
	\item $\nu = \norm{x}_p$, the appropriate norm of the iterate
\end{itemize}

In addition, to support sampling, our data structure also maintains a binary tree $\dist_x$ of depth $O(\log n)$. Each leaf node is associated with a coordinate $j \in [n]$, and each internal node is associated with a subset of coordinates corresponding to leaves in its subtree. For the node corresponding to $S \subseteq [n]$ (where $S$ may be a singleton), we maintain 
the sums $\sum_{j \in S} [u]_j^p$, $\sum_{j \in S} [u]_j [v]_j$, and $\sum_{j 
\in S} [v]_j^p$.

We now give the implementation of each operation supported by $\IM_d$, 
followed by proofs of correctness and of the runtime bounds when 
applicable.

\subsubsection{Initialization} 

\begin{itemize}\item$\Initialize(x_0, v)$. Runs in time $O(n)$.
	
	If $p=1$  set $v \gets \mathbf{0}_n$; otherwise we compute and store 
	$\norm{v}_2^2$. Initialize the remaining data structure state as 
	follows: $(\xi_u, \xi_v, u) \leftarrow (1, 0, x_0)$, $(\sigma_u, \sigma_v, 
	u') \leftarrow (0, 0, \mathbf{0}_n)$, $(\iota, \nu) \leftarrow 
	(\inner{x_0}{v}, \norm{x_0}_p)$. Initialize $\dist_x$, storing the relevant 
	sums in each internal node. 
\end{itemize}

It is clear that $x = \xi_u u + \xi_v v$, $s = u' + \sigma_u u + \sigma_v 
v$, and that the invariants of $\iota, \nu$ hold. Each step takes $O(n)$ 
time; for the first 4 steps this is immediate, and the final recursing 
upwards from the leaves spends constant time for each internal node, 
where there are $O(n)$ nodes.

\subsubsection{Updates}

\begin{itemize}
	\item $\Scale(c)$: $x \leftarrow cx$. Runs in time $O(1)$.
	
	Multiply each of $\xi_u, \xi_v, \nu, \iota$ by $c$.
	\item $\AddSparse(j, c)$: $[x]_j \leftarrow [x]_j + c$, with the guarantee $c \geq -[x]_j$ if $p = 1$. Runs in time $O(\log n)$.
	\begin{enumerate}
		\item $u \leftarrow u + \frac{c}{\xi_u} e_j$.
		\item $u' \leftarrow u' - \frac{c\sigma_u}{\xi_u} e_j$.
		\item If $p = 1$, $\nu \leftarrow \nu + c$. 
			 If $p = 2$, $\nu \leftarrow \sqrt{\nu^2 + 2c[\xi_u u + \xi_v v]_j + c^2}$.
		\item $\iota \leftarrow \iota + c [v]_j$.
		\item For internal nodes of $\dist_x$ on the path from leaf $j$ to the root, update $\sum_{j \in S} [u]_j^p$, $\sum_{j \in S} [u]_j [v]_j$ appropriately.
	\end{enumerate}
	\item $\AddDense(c)$: $x \leftarrow x + cv$. Runs in time $O(1)$. 
	(Supported only for $p=2$).
	
	Set $\xi_v \leftarrow \xi_v + c$, $\nu \leftarrow \sqrt{\nu^2 + 2c\iota 
	+ c^2\norm{v}_2^2}$, and $\iota \leftarrow \iota + c\norm{v}_2^2$.
	\item $\SumUp()$: $s \leftarrow s + x$. Runs in time $O(1)$.
	
	Set $\sigma_u \leftarrow \sigma_u + \xi_u$ and $\sigma_v \leftarrow 
	\sigma_v + \xi_v$.
\end{itemize}

Each of the runtime bounds clearly hold; we now demonstrate that the necessary invariants are preserved. Correctness of $\Scale$ and $\SumUp$ are clear. Regarding correctness of $\AddSparse$, note that (ignoring the $v$ terms when $p = 1$)
\begin{align*}
\xi_u\left(u + \frac{c}{\xi_u}e_j \right) + \xi_v v &= \xi_u u + \xi_v v+ ce_j,\\
\left(u' - \frac{c\sigma_u}{\xi_u} e_j\right) + \sigma_u\left(u + \frac{c}{\xi_u}e_j \right) + \sigma_v v &= u' + \sigma_u u + \sigma_v v.
\end{align*}
When $p = 1$, the update to $\nu$ is clearly correct. When $p = 2$, because only $[x]_j$ changes,
\begin{align*}
[\xi_u u + \xi_v v + ce_j]_j^2 &= [\xi_u u + \xi_v v]_j^2 + 2c[\xi_u u + \xi_v v]_j + c^2,\\
\left([\xi_u u + \xi_v v + ce_j]_j\right) \cdot [v]_j &= \left([\xi_u u + \xi_v v]_j\right) \cdot [v]_j + c[v]_j.
\end{align*}
Thus, the updates to the norm and inner product are correct. Regarding 
correctness of $\AddDense$ when $p = 2$, we have
\begin{align*}
\xi_u u + (\xi_v + c) v&= \xi_u u + \xi_v v + cv, \\
\norm{x + cv}_2^2 &= \nu^2 + 2c\iota + c^2 \norm{v}_2^2, \\
\inner{x + cv}{v} &= \iota + c\norm{v}_2^2.
\end{align*}
Here, we use the invariants that $\nu = \norm{x}_2$ and $\iota = \inner{x}{v}$.%

\subsubsection{Queries}

\begin{itemize}
	\item $\Get(j)$: Return $[x]_j$. Runs in time $O(1)$.
	
	Return $\xi_u [u]_j + \xi_v [v]_j$.
	\item $\GetSum(j)$: Return $[s]_j$. Runs in time $O(1)$.
	
	Return $[u']_j + \sigma_u [u]_j + \sigma_v [v]_j$.
	\item $\Norm()$: Return $\norm{x}_p$. Runs in time $O(1)$.
	
	Return $\nu$.
\end{itemize}

By our invariants, each of these operations is correct.

\subsubsection{Sampling}

The method $\Sample$ returns a coordinate $j$ with probability 
proportional to $[x]_j^p$ in time $O(\log n)$. To implement it, we 
recursively perform the following procedure, where the recursion depth is at 
most $O(\log n)$, starting at the root node and setting $S = [n]$:
\begin{enumerate}
	\item Let $S_1, S_2$ be the subsets of coordinates corresponding to the children of the current node.
	\item Using scalars $\xi_u, \xi_v$, and the maintained $\sum_{j \in S_i} [u]_j^p$, $\sum_{j \in S_i} [u]_j [v]_j$, $\sum_{j \in S_i} [v]_j^p$ when appropriate, compute $\sum_{j \in S_i} [x]_j^p = \sum_{j \in S_i} [\xi_u u + \xi_v v]_j^p$ for $i \in \{1, 2\}$.
	\item Sample a child $i \in \{1, 2\}$ of the current node proportional to $\sum_{j \in S_i} [x]_j^p$ by flipping an appropriately biased coin. Set $S \leftarrow S_i$.
\end{enumerate}

It is clear that this procedure samples according to the correct probabilities. Furthermore, step 2 can be implemented in $O(1)$ time using precomputed values, so the overall complexity is $O(\log n)$.

\subsection{$\GlobM$}\label{ssec:ds-globalm}

In this section, we give the implementation of $\AEM$ which supports 
dense update to simplex mirror descent iterates. For convenience, we 
restate its interface, where we recall the notation $\norm{g}_0$ for the 
number of nonzero entries in $g$, Definition~\ref{def:betastable} of an  
$\varepsilon$-padding, the invariant 
\begin{equation}\label{eq:aemhatx}\hat{x} \text{ is a } \varepsilon 
\text{-padding of } x,
\end{equation}
and the notation
\[\omega \defeq \max\left(\frac{1}{1 - \kappa},\;\frac{n}{\lambda\eps}\right).\]

\theAEM

\newcommand{\sigmin}{\sigma_{\min}}

We build $\AEM$ out of a
simpler data structure called $\ScM$, which maintains the simplex 
projection of fixed vectors raised elementwise to arbitrary powers; this 
suffices to support consecutive $\DenseStep$ calls without $\MultSparse$ 
calls between them. To add support for $\MultSparse$, we combine $O(\log 
n)$ instances of $\ScM$ in a formation resembling a binomial heap: for 
every entry updated by $\MultSparse$ we delete it from the $\ScM$ instance 
currently holding it, put it in a new singleton $\ScM$ instance (after 
appropriate scaling due to $\MultSparse$), and merge this singleton into 
existing instances. We now give a brief description of the $\ScM$ interface, 
and based on it, describe the implementation of $\AEM$. 
 We will provide the implementation of $\ScM$ in Section~\ref{sec:scm}. 

$\ScM$ 
is 
initialized with vectors $\bx\in\R^{n'}_{\ge 0}$ and $\bg\in\R^{n'}$ (with 
$n'\le n$) and 
supports efficient approximate queries on vectors of the form
\begin{equation*}
x[\sigma] \defeq \bx \circ \exp\left(\sigma\bg\right),
\end{equation*}
for any scalar $\sigma\in [\sigmin, 1]$. More specifically, the data 
structure allows efficient computation of  $\lone{x[\sigma]}$ (to within 
small multiplicative error $\epsscm$), as well as entry queries, sampling 
and running sum accumulation from a vector $\hat{x}[\sigma]$ satisfying
\begin{equation}
\label{eq:scmhatx}
 \hat{x}[\sigma] \text{ is a } \epsscm 
\text{-padding of } \Pi_{\Delta}\left(\bx \circ 
\exp\left(\sigma\bg\right)\right) = \frac{x[\sigma]}{\lone{x[\sigma]}}.
\end{equation}

We make the following assumptions on the input to the data structure:
\begin{equation*}
\lamscm \le [\bx]_i \le 1~\mbox{for all}~i\in[n']~~\mbox{and}~~
\sigma \in (\sigmin, 1).
\end{equation*}
The upper bounds on $\bx$ and $\sigma$ are arbitrary, and we may 
choose  $\lamscm$  and $\sigmin$ to be very small since the data 
structure runtime depends on them only logarithmically. To summarize this 
dependence, we define
\begin{equation*}
\omscm \defeq \max\left\{
\frac{1}{\sigmin} , \frac{n}{\lamscm\epsscm}
\right\}. 
\end{equation*}
With these assumptions and notation, we define the formal interface of 
$\ScM$.

\newcommand{\theSCM}{
\begin{table}[h]
	\centering
	\renewcommand{\arraystretch}{1.25}
	\begin{tabular}{c||l|l}
		{\bf Category}                   & {\bf Function} & {\bf 
			Runtime} \\ \hline
		\multicolumn{1}{c||}{\multirow{1}{*}{initialize}}
		& $\Initialize(\bx, \bg, \sigmin, \epsscm, \lamscm)$
		& $O(n'\log n \log^2\omscm)$\\ 
		\hline
		\multicolumn{1}{c||}{\multirow{2}{*}{update}}
		& $\Del(j)$: Remove coordinate $j$ from $\bx$, $\bg$ & $O(1)$ \\ \cline{2-3} 
		\multicolumn{1}{c||}{} & $\SumUp(\gamma, \sigma)$: $s \gets s+ 
		\gamma \hat{x}[\sigma]$, with $\hat{x}[\sigma]$ defined 
		in~\eqref{eq:scmhatx} & $O(\log\omscm)$ \\ 
		\hline
		\multicolumn{1}{c||}{\multirow{3}{*}{query}}
		& $\Get(j)$: Return 
		$[\hat{x}[\sigma]]_j$ & $O(\log\omscm)$ \\ \cline{2-3} 
		\multicolumn{1}{c||}{} & $\GetSum(j)$: Return $[s]_j$. & 
		$O(\log^2\omscm)$\\ \cline{2-3} 
		\multicolumn{1}{c||}{} & $\GetNorm(\sigma)$: Return $1 \pm \veps$ 
		approx. of $\norm{\bx \circ \exp(\sigma\bg)}_1$ & $O(\log\omscm)$ 
		\\ 
		\hline
		sample                    
		& $\Sample(\sigma)$: 
		$ \text{Return } j \text{ with probability } 
		[\hat{x}[\sigma]]_j$  & 
		$O(\log n \log \omscm)$  \\ 
	\end{tabular}
\end{table}
}
\theSCM

\subsubsection{$\AEM$ state}

Throughout this section, we denote $K \defeq \lceil \log n\rceil$. 
$\GlobM$ maintains a partition of $[n]$ into $K$ sets $S_1 \ldots, 
S_K$ (some of them possibly empty) that satisfy the invariant
\begin{equation}\label{eq:aem-rank-k-invariant}
|S_k| \le 2^k~\mbox{for all}~k\in[K].
\end{equation}
We refer to the index $k$ as ``rank'' and associate with each rank 
$k\in[K]$ the 
following data
\begin{enumerate}
	\item Scalar $\gamma_k\ge0$ and nonnegative integer $\tau_k$.
	\item Vectors $\bx_k, \bg_k \in \R^{|S_k|}$ such that $\lamscm \le 
	[\bx_k]_i \le 1$ for all $i \in [|S_k|]$, where $\lamscm = \min(\veps/n, 
	\lambda)$.
	\item A $\ScM$ instance, denoted $\ScM_k$, initialized with $\bx_k, 
	\bg_k$ and $\lamscm$ defined above, $\sigmin = 1-\kappa$ and 
	$\epsscm = \veps / 10$, so 
	that $\log \omscm = O(\log \omega)$. 
\end{enumerate}
$\AEM$ also maintains a vector $u \in \R^n$ for auxiliary running 
sum storage.

Define the vector $\delta\in \R^n$ by
\begin{equation}\label{eq:aem-delta-def}
\left[\delta\right]_{S_k} = 
\log\left(\gamma_k \left[\bx_k \circ \exp\left((1 - \kappa^{\tau_k}) 
\bg_k\right)\right]\right),~k\in\{1,\ldots,K\},
\end{equation}
where $\left[\delta\right]_{S_k}$ denotes the coordinates of $\delta$ in 
$S_k$. 
Recall that $x$ denotes the point in $\Delta^n$ maintained throughout the 
operations of $\AEM$; we maintain the key invariant that the point $x$ is 
proportional to $\exp(\delta)$, i.e.,
\begin{equation}\label{eq:scminvar}
x = \frac{\exp(\delta)}{\norm{\exp(\delta)}_1}.
\end{equation}
Specifically, we show in Section~\ref{sssec:aem_updates} that our 
implementation of $\DenseStep$ modifies $\bx$, $\bg$, $\{\tau_k\}_{k = 
0}^K$, $\{\gamma_k\}_{k = 0}^K$ so that the resulting effect on 
$\delta$, per definition \eqref{eq:scminvar}, is
\begin{equation}\label{eq:deltadensestep}\delta \gets \kappa \delta + v.\end{equation}
Similarly, our implementation of $\MultSparse$ modifies the state so that the resulting effect on $\delta$ is 
\begin{equation}\label{eq:deltamultsparse}\frac{\exp(\delta)}{\norm{\exp(\delta)}_1} \gets  \veps\text{-padding of } \frac{\exp(\delta + v)}{\norm{\exp(\delta + v)}_1}.\end{equation}
We remark that the role of $\gamma_k$ is to scale $\bx_k$ so that it lies 
coordinatewise in the range $[\lamscm, 1]$, conforming to the $\ScM$ 
input requirement. This is also the reason we require the 
$\veps$-padding operation in the definition of $\MultSparse$. 

\subsubsection{$\veps$-padding point $\hat{x}$}\label{sssec:aem-padding}
We now concretely define the point $\hat{x}$, 
which is the $\veps$-padding of $x$ that $\AEM$ maintains. Let
\begin{equation}\label{eq:aem-Gamma}
\Gamma \defeq \sum_{k = 1}^{K} \gamma_k \ScM_k.\GetNorm(1 - 
\kappa^{\tau_k}),
\end{equation}
be the approximation of $\exp(\delta)$ derived from the $\ScM$ instances. 
For any $j\in[n]$, let $k_j$ be such that $j\in S_{k_j}$, and let $i_{j}$ be the 
index of $j$ in $S_{k_j}$. The $j$th coordinate of $\hat{x}$ is
\begin{equation}\label{eq:hatxdef}
[\hat{x}]_j \defeq \frac{\gamma_{k_j} \ScM_{k_j}.\GetNorm(1 - 
	\kappa^{\tau_{k_j}})}{\Gamma} \cdot \ScM_{k_j}.\Get(i_j, 1 - 
\kappa^{\tau_{k_j}}).
\end{equation}
Since for each $k$, $\sum_{j \in S_k} \ScM_k.\Get(j, 1 - \kappa^{\tau_k}) = 
\ScM_k.\GetNorm(1 - \kappa^{\tau_k})$ we have that $\hat{x} \in 
\Delta^n$. We now prove that $\hat{x}$ is a $\veps$-padding of 
$x$. To do so, we prove the following lemma.

\begin{lemma}
	\label{lem:stabcombine}
	Let $\epsscm\le \frac{1}{10}$ and $\{S_k\}_{k = 1}^K$ be a partition of 
	$[n]$. Suppose for each $k \in [K]$, $\hat{x}_k \in \Delta^{|S_k|}$ is an $\epsscm$-padding of 
	$x_k \in \Delta^{|S_k|}$. Further, suppose we have positive scalars 
	$\{\nu_k\}_{k 
		= 1}^K$, $\{\hat{\nu}_k\}_{k = 1}^K$ satisfying
	\[(1 - \epsscm) \nu_k \le \hat{\nu}_k \le (1 + \epsscm)\nu_k, \text{ for 
	all 
	} 1 \le k \le K.\]
	Then, for $N = \sum_{k = 1}^K \nu_k$ and $\hat{N} = \sum_{k = 1}^K 
	\hat{\nu}_k$, we have that $\hat{x} \defeq \sum_{k = 1}^K 
	\frac{\hat{\nu}_k}{\hat{N}} \hat{x}_k$ is a $10\epsscm$-padding 
	of 
	$x \defeq \sum_{k = 1}^K \frac{\nu_k}{N} x_k$.
\end{lemma}
\begin{proof}
	For every $k\in[K]$, let $\tilde{x}_k$ to be such that $\tilde{x}_k \ge 
	x_k$ 
	elementwise, $\hat{x}_k = \tilde{x}_k/\norm{\tilde{x}_k}_1$, and 
	$\norm{\tilde{x}_k - x_k}_1 \le \epsscm$. Consider the point
	\begin{equation*}
	\tilde{x}\defeq \sum_{k = 1}^K \frac{\tilde{\nu}_k}{\hat{N}} \tilde{x}_k,
	~~\mbox{where}~~\tilde{\nu}_k \defeq \hat{\nu}_k\frac{\max_{k\in[K]} 
		\lone{\tilde{x}_k}}{\lone{\tilde{x}_k}} 
	\cdot \frac{1+\epsscm}{1-\epsscm},
	\end{equation*}
	so that $\hat{x} = \tilde{x}/\lone{\tilde{x}}$. Since $\tilde{x}_k \ge x_k$ 
	elementwise, 
	$\hat{\nu}_k \ge (1-\epsscm){\nu}_k$ and $\hat{N} \le (1+\epsscm)N$, 
	we 
	have that $\tilde{x} \ge {x}$ elementwise. Furthermore, we have 
	$\hat{\nu}_k \le (1+\epsscm){\nu}_k$ and $\hat{N} \ge (1-\epsscm)N$, 
	and the properties 
	$\tilde{x}_k \ge x_k$ 
	and $\lone{\tilde{x}_k - x_k}$ imply $1\le 
	\lone{\tilde{x}_k}\le1+\epsscm$ 
	as well as $\frac{\max_{k\in[K]} \lone{\tilde{x}_k}}{\lone{\tilde{x}_k}} \le 
	1+\epsscm$. Therefore
	\begin{equation*}
	\lone{\tilde{x}-x} \le \sum_{k=1}^K \frac{\nu_k}{N}\lone{ 
		\frac{(1+\epsscm)^3}{(1-\epsscm)^2}  \tilde{x}_k - x_k}
	\le\left(\frac{(1+\epsscm)^3}{(1-\epsscm)^2}- 1\right)(1+\epsscm) + 
	\epsscm \le 10\epsscm,
	\end{equation*}
	where the final bound is verified numerically for $\epsscm \le 1/10$.
\end{proof}

The $\ScM$ interface guarantees that calls to $\ScM_k.\GetNorm$ return 
$\norm{\bx_k \circ \exp((1 - \kappa^{\tau_k})\bg_k)}_1$ to within a $1 
\pm \epsscm$ multiplicative factor, and moreover that $\Get$ 
returns entries from an $\epsscm$-padding of $\Pi_\Delta(\bx_k 
\circ \exp((1 - \kappa^{\tau_k})\bg_k))$. Thus, applying 
Lemma~\ref{lem:stabcombine} with our definition of $\hat{x}$ in 
\eqref{eq:hatxdef} yields that $\hat{x}$ is a $10\epsscm = \veps$-padding of $x$.

\subsubsection{$\AEM$ initialization and updates}
\label{sssec:aem_updates}

We give the implementation and prove runtimes of $\Initialize$, $\MultSparse$, $\DenseStep$, and $\UpdateSum$.

\paragraph{$\Initialize$.} Upon initialization of $\GlobM$, we set 
$\gamma_K = \max_{j \in [n]} [x_0]_j$ and $\tau_k=0$ for all $k$. We let 
$S_K = [n]$ (so that $S_k = 
\emptyset$ for all $k<K$) and instantiate a 
single instance of $\ScM$ of rank $K$ with parameters 
\[\bx_K = \frac{x_0}{\gamma_K},\; \bg_K = \frac{v}{1 - \kappa} - \log 
x_0,\; \epsscm = \frac{\veps}{10},\; \lamscm = \min\left(\frac{\veps}{n},\; 
\lambda\right).\]
It is clear that the invariant \eqref{eq:scminvar} holds at initialization, and 
that the coordinates of $\bx_K$ lie in the appropriate range, since we 
assume that $x_0\in[\lambda,1]^n$. We will use 
the same choices of 
$\epsscm$, $\lamscm$ for every $\ScM$ instance. 
The overall complexity of this operation is $O(n \log n\log^2\omega)$.

\paragraph{$\MultSparse$.} We state the implementation of $\MultSparse$, 
prove that the resulting update is \eqref{eq:deltamultsparse}, and finally 
give its runtime analysis. We perform $\MultSparse(g)$ in sequence for 
each nonzero coordinate of $g$. Let $j$ denote such nonzero coordinate 
and let $k_j$ be such that $j\in S_{k_j}$; the operation consists of the 
following 
steps.
\begin{enumerate}
	\item Remove $j$ from $S_{k_j}$ and delete the corresponding 
	coordinate from $\ScM_{k_j}$ (via a call to $\Del$).
	\item Let $S_0={j}$ and initialize $\ScM_0$ with initial data $\bx$ and 
	$\bg$ described below.  
	\item For $k$ going from $1$ to $K$, set $S_k\gets S_k \cup S_{k-1}$ 
	and 
	$S_{k-1} =\emptyset$, 
	merging $\ScM_{k}$ and $\ScM_{k-1}$ as described below. If the new set 
	$S_k$ satisfies $|S_k| \le 2^k$, break the loop; else, proceed to 
	the next $k$. 
\end{enumerate}

We now state the initial data given to each $\ScM$ upon initialization in the 
steps above. Whenever a $\ScM_k$ is created supporting $S_k \subseteq 
[n]$, we first compute $\delta_i$ for each $i \in S_k$ according to 
\eqref{eq:scminvar}.
When creating the singleton instance $\ScM_0$ we perform the update 
\begin{equation}\label{eq:aem-multsparse-exact}
\delta_j \gets \delta_j + g_j; 
\end{equation}
this implements multiplication of the $j$th 
coordinate by $\exp(g_j)$. To instantiate $\ScM_k$, we set $\tau_k=0$,  $\gamma_k 
\gets \max_{i \in S_k} \exp([\delta]_{i})$ and modify $\delta$ according to
\begin{equation}\label{eq:aem-multsparse-pad}
[\delta]_{S_k} \gets \max\left\{[\delta]_{S_k} ,\log\left(\lamscm \cdot 
\gamma_k\right)\right\}.
\end{equation}
In other words, we raise very small entries of $[\delta]_{S_k}$ to ensure that 
the ratio between any two entries of $\left[\exp(\delta)\right]_{S_k}$ is in 
the range $[\lamscm^{-1}, \lamscm]$.  We then give $\ScM_k$ the initial 
data
\begin{equation}\label{eq:aem-scm-init}
\bx_k = \frac{1}{\gamma_k}\left[\exp\left(\delta\right)\right]_{S_k},\; 
\bg_k = \left[\frac{v}{1 - \kappa} - \delta\right]_{S_k}.
\end{equation}
It is clear that entries of $\bx_k$ are in the range $[\lamscm, 1]$, and 
invariant \eqref{eq:scminvar} holds at initialization, as $\tau_k = 0$. 
Therefore, the operation~\eqref{eq:aem-multsparse-exact} implements 
$x\gets \Pi_\Delta(x \circ \exp(g))$ exactly; it remains to show that the 
operation~\eqref{eq:aem-multsparse-pad} amounts to an $\veps$-padding of $\exp(\delta)/\lone{\exp(\delta)}$. We do so by invoking the 
following  
lemma, substituting the values of $\delta$ 
before~\eqref{eq:aem-multsparse-exact} for $\delta^{-}$ and $\delta^{+}$ 
respectively, and $\lamscm \le \veps/n$ for $\rho$.

\begin{lemma}
\label{lem:raisedeltastable}
Let $\delta^{-},\delta^{+}\in \R^n$ satisfy
\begin{equation*}
[\delta^{+}]_i = \max\left\{ [\delta^{-}]_i, \max_{j}[\delta^{-}]_j + \log \rho 
\right\}
~\mbox{for all}~j\in[n]~\mbox{and}~\rho \le 1.
\end{equation*}
Then, $\exp(\delta^{+})/\norm{\exp(\delta^{+})}_1$ is a $\rho n$-padding of $\exp(\delta^{-})/\norm{\exp(\delta^{-})}_1$. 
\end{lemma}
\begin{proof}
Let $x=\exp(\delta^{-})/\norm{\exp(\delta^{-})}_1$,  
$x'=\exp(\delta^{+})/\norm{\exp(\delta^{+})}_1$, and $\tilde{x} = 
\exp(\delta^{+})/\norm{\exp(\delta^{-})}_1$. Clearly $x' = \tilde{x} 
/\lone{\tilde{x}}$ and $\tilde{x} \ge x$ element-wise. Moreover, letting 
$M=\max_{j}\exp([\delta^{-}]_j)$, we have
\begin{equation*}
\lone{\tilde{x}-x} = \frac{\lone{\exp(\delta^{+})-\exp(\delta^{-})}}{
	\lone{\exp(\delta^{-})}} \le 
\frac{\rho M |\{i\mid [\delta^{+}]_i \ne [\delta^{-}]_i\}| }{
	\lone{\exp(\delta^{-})}} \le \frac{\rho M \cdot 
	n}{\lone{\exp(\delta^{-})}} \le \rho n,
\end{equation*}
establishing the $\rho n$-padding property.
\end{proof}

Finally, we discuss runtime. Recall that the cost of initializing a $\ScM$ 
with $|S|$ elements is $O(|S|\log n\log^2\omega)$. So, step 1 of our 
implementation of $\MultSparse$, i.e., calling $\Del$ once and initializing 
a rank-1 $\ScM$ per nonzero element, costs $O(\norm{g}_0\log 
n\log^2\omega)$. We now discuss costs of merging in step 2. We show 
these merges cost an amoritized $O(\log^2 n\log^2\omega)$ per nonzero 
coordinate of $g$, leading to the claimed bound. Specifically, we show the 
cost of $T$ deletions and initializations due to nonzero entries of 
$\MultSparse$ arguments is $O(T\log^2 n\log^2\omega)$. Consider the 
number of times a rank-$k$ set can be created through merges: we claim it 
is upper bounded by $O\left(T/2^k\right)$. It follows that the overall 
complexity of step 2 is
\begin{equation*}
O\left(\sum_{k = 0}^{K} 2^k \frac{T}{2^k}\log n\log^2\omega\right) = O(T\log^2 n \log^2\omega).
\end{equation*}
The claimed bound on the number of rank-$k$ merges holds because at 
least $2^{k - 1}$ deletions (and hence that many $\MultSparse$ calls) must 
occur between consecutive rank-$k$ merges. To see this, for each $k$, 
maintain a 
potential $\Phi_k$ for the sum of cardinalities of all rank $\ell$ sets for 
$\ell < k$. Each deletion increases $\Phi_k$ by at most 1. For an insertion 
merge to create a rank-$k$ set, $\Phi_k$ must have been at least $2^{k - 
1} + 2$; after the merge, it is 0, as in its creation, all rank-$\ell$ sets for 
$\ell < k$ must have been merged. So, there must have been at least $2^{k 
- 1}$ deletions in between merges.

\paragraph{$\DenseStep$.} To implement $\DenseStep$ we simply 
increment $\tau_k \gets \tau_k + 1$ for all $k$; clearly, this takes time 
$O(\log n)$. We now show that \eqref{eq:scminvar} is maintained, i.e., that 
the resulting update to the variable $\delta$ under $\DenseStep$ is 
\eqref{eq:deltadensestep}. Recall that $\ScM_k$ is initialized 
according~\eqref{eq:aem-scm-init}. 
Clearly, \eqref{eq:scminvar} holds at initialization for the set $S_k$, as $1 - 
\kappa^0 = 0$. We now show that it continues to hold after any number of 
$\DenseStep$ calls. Let 
$\delta_0$ be the value 
of $[\delta]_{S_k}$ when $\ScM_k$ is initialized, and let $\delta_\tau$ be 
the value of $[\delta]_{S_k}$ after $\tau$ calls to $\DenseStep$, each 
performing the update $x\gets \Pi_\Delta(x^{\kappa}\circ \exp(v))$. This is 
consistent with the update $\delta_{\tau + 1} = \kappa 
\delta_\tau + [v]_{S_k}$, which requires
\begin{equation*}
\delta_{\tau} = \kappa^{\tau} \delta_0 + \sum_{\tau'=0}^{\tau-1} 
\kappa^{\tau'} [v]_{S_k} = \kappa^{\tau} \delta_0 + 
\frac{1-\kappa^{\tau}}{1-\kappa}[v]_{S_k}
= \log(\gamma_k \bx_k) + (1-\kappa^\tau) \bg_k,
\end{equation*}
where in the final transition we substituted $\delta_0=\log(\gamma_k 
\bx_k)$ and $[v]_{S_k}=(1-\kappa)[\bg_k + \delta_0]$ according 
to~~\eqref{eq:aem-scm-init}. We see that the required of form of 
$\delta_\tau$ is identical to its definition~\eqref{eq:aem-delta-def} and 
consequently that~\eqref{eq:scminvar} holds.

\paragraph{$\UpdateSum$.} We maintain the 
running sum $s$ via the invariant
\begin{equation}\label{eq:globmsdef} [s]_{S_k} = [u]_{S_k} +  
\ScM_k.\GetSum(),\; \forall k \in [K],\end{equation}
which we preserve in two separate procedures. First, whenever $\SumUp()$ 
is called, we compute the quantity $\Gamma$ defined 
in~\eqref{eq:aem-Gamma}, and for each $k\in[K]$ call
\[\ScM_k.\SumUp\left(\frac{\gamma_k \ScM_k.\GetNorm(1 - \kappa^{\tau_k})}{\Gamma}\right).\]
It is straightforward to see that this indeed preserves the invariant \eqref{eq:globmsdef} for our definition of $\hat{x}$ in \eqref{eq:hatxdef}, and takes time $O(\log n\log\omega)$. Next, whenever a coordinate is deleted from a $\ScM_k$ instance, or an entire $\ScM_k$ instance is deleted due to a merge operation, we update
\[u_j \gets u_j + \ScM_k.\GetSum(j)\]
for every deleted coordinate $j$, or $j$ involved in the merge, respectively. 
We charge the cost of this operations to that of new $\ScM$ instance, 
which we accounted for in the analysis of $\MultSparse$.

\subsubsection{Queries}

\paragraph{$\Get(j)$.} Recalling our definition of $\hat{x}$ \eqref{eq:hatxdef}, we compute $\Gamma$ in time $O(\log n\log\omega)$ by obtaining $\GetNorm(1 - \kappa^{\tau_k})$ for each $k$, and then call $\Get(j, 1 - \kappa^{\tau_k})$ in time $O(\log \omega)$ for the relevant $k$. 

\paragraph{$\GetSum(j)$.} Recalling our definition of $s$ 
\eqref{eq:globmsdef}, we implement $\GetSum(j)$ in $O(\log \omega)$ 
time via a single call to $\GetSum$ on the relevant $\ScM$ instance, and 
querying a coordinate of $u$.

\paragraph{$\Sample$.} Recalling \eqref{eq:hatxdef}, we first compute 
$\Gamma$, as well as all $\gamma_k \ScM_k.\GetNorm(1 - \kappa^{\tau_k})$, in $O(\log
n\log\omega)$ time. We then sample an instance $\ScM_k$, for $0 \le k \le K$, proportional to the value $\gamma_k \ScM_k.\GetNorm(1 - \kappa^{\tau_k})$, in $O(\log n)$ time. Finally, for the sampled instance, we call $\ScM_k.\Sample(1 - \kappa^{\tau_k})$ to output a coordinate in $O(\log n \log\omega)$ time. By 
the definition of $\hat{x}$ used by $\GlobM$, as well as the definitions 
used by each $\ScM_k$ instance, it is clear this preserves the correct 
sampling probabilities. %

\subsection{$\ScM$}
\label{sec:scm}

Finally, we provide a self-contained treatment of $\ScM$, the main building 
block in the implementation of $\AEM$ described above.

\subsubsection{Interface}
\label{ssec:scm_if}

For ease of reference we restate the interface of $\ScM$, where for the sake 
of brevity we drop the subscript $\text{scm}$ from $\veps$ and 
$\lambda$, and use $n$ rather than $n'$ to denote the input dimension.
Recall that for the vectors $\bx$ and $\bg$ given at initialization, the data 
structure keeps track of vectors of the form
\begin{equation}\label{eq:hxsigma}
\hat{x}[\sigma] \defeq \text{a } \veps \text{-padding of } 
\Pi_{\Delta}\left(\bx \circ \exp\left(\sigma\bg\right)\right),
\end{equation}
where $\sigma$ is any scalar in the range $\{0\}\cup[\sigmin, 1]$. 

The implementation of the data structure relies on three internal 
parameters: polynomial approximation order $p\in\N$, truncation 
threshold $R\ge0$, and $\sigma$ discretization level $K\in\N$. To satisfy 
the accuracy requirements we set these as
\begin{equation*}
R = \Theta(1) \log\frac{1}{\veps\lambda},~~p=\Theta(1) 
\log\frac{1}{\veps\lambda},~~\mbox{and}~~K=\ceil{\log\frac{1}{\sigmin}};
\end{equation*}
we give the runtime analysis in terms of these parameters.
\begin{table}[h]
	\centering
	\renewcommand{\arraystretch}{1.25}
	\begin{tabular}{c||l|l}
		{\bf Category}                   & {\bf Function} & {\bf 
			Runtime} \\ \hline
		\multicolumn{1}{c||}{\multirow{1}{*}{initialize}}
		& $\Initialize(\bx, \bg, \sigmin, \veps, \lambda)$: require $\bx \in 
		[\lambda, 1]^{n}$ & $O(n pK \log n )$\\ 
		\hline
		\multicolumn{1}{c||}{\multirow{2}{*}{update}}
		& $\Del(j)$: Remove coordinate $j$ from $\bx$, $\bg$ & 
		$O(1)$ \\ \cline{2-3} 
		\multicolumn{1}{c||}{} & $\SumUp(\gamma, \sigma)$: $s \gets s+ 
		\gamma \hat{x}[\sigma]$ & $O(p)$ \\ 
		\hline
		\multicolumn{1}{c||}{\multirow{3}{*}{query}}
		& $\Get(j, \sigma)$: Return 
		$[\hat{x}[\sigma]]_j$ & $O(p)$ \\ \cline{2-3} 
		\multicolumn{1}{c||}{} & $\GetSum(j)$: Return $[s]_j$. & 
		$O(pK)$ \\ 
		\cline{2-3} 
		\multicolumn{1}{c||}{} & $\GetNorm(\sigma)$: Return $1 \pm \veps$ 
		approx. of $\norm{\bx \circ \exp(\sigma\bg)}_1$ & $O(p)$ \\ 
		\hline
		sample                    
		& $\Sample(\sigma)$: 
		$ \text{Return } j \text{ with probability } 
		[\hat{x}[\sigma]]_j$  & 
		$O(p \log n)$  \\ 
	\end{tabular}
\end{table}

\subsubsection{Overview}

We now outline our design of $\ScM$, where the main challenge is supporting 
efficient $\GetNorm$ operations under no assumptions on the numerical 
range of the input $\bg$. 

\paragraph{Exponential approximation via Taylor expansion.} 
Our main strategy is to replace the exponential in the 
definition of $\hat{x}[\sigma]$ with its Taylor expansion of order 
$p=O(\log\frac{n}{\veps\lambda})$, giving the following approximation to 
the $\GetNorm(\sigma)$
\begin{equation*}
\norm{\bx \circ \exp\left(\sigma\bg\right)}_1 \approx 
\inner{\bx}{\sum_{q=0}^p \frac{1}{q!} (\sigma \bg)^q}
= \sum_{q=0}^p \frac{\sigma^q}{q!} \inner{\bx}{\bg^q},
\end{equation*}
where $q$th powers are applied to $\bg$ elementwise. By pre-computing 
all the 
inner products $\{\inner{\bx}{\bg^q}\}_{q=0}^p$ at initialization, we may 
evaluate this Taylor approximation of $\GetNorm$ in time $O(p)$. The 
validity of the approximation relies on the following well-known fact.

\begin{fact}[{Theorem 4.1 in \cite{SachdevaV14}}]
	\label{fact:approxexp}
	Let $\veps', R\ge0$. A Taylor series $f_p(t) = \sum_{q=0}^p 
	\frac{t^q}{q!}$ of degree $p=O(R+\log\frac{1}{\veps'})$ satisfies 
	\[|\exp(t) - f_p(t)| \le \exp(t)\veps'
	~~\mbox{for all}~t\in[-R,0].\]
\end{fact}

\newcommand{\hsigma}{\hat{\sigma}}

\paragraph{Truncating small coordinates and $\sigma$ discretization.} 
For Fact~\ref{fact:approxexp} to directly imply the desired approximation 
guarantee for $\GetNorm$, the entries of $\sigma\bg$ must all lie in 
$[-R,0]$ for some $R=\Otil{1}$. However, this will not hold in general, as 
our data structure must support any value of $\bg$. For a fixed value of 
$\sigma$, we can work instead with a shifted and truncated version of 
$\bg$, i.e., 
\begin{equation*}
\td[\sigma,\mu] \defeq \max\{ \bg - \mu, -R/\sigma\}, 
\end{equation*}
where 
the offset $\mu$ is roughly the maximum element of $\bg$. 
Fact~\ref{fact:approxexp} allows us to approximate the exponential of 
$\sigma\td[\sigma,\mu]$, and for $R=\Theta(\ltnfrac)$ we argue that the 
truncation of the smallest entries of $\delta$ results in small multiplicative 
error. Unfortunately, the dependence of $\td[\sigma,\mu]$ on $\sigma$ 
would defeat the purpose of efficient computation, because it is 
impossible  
to precompute $\{\inners{\bx}{\td[\sigma,\mu]^q}\}_{q=0}^p$ for every 
$\sigma\in [\sigmin,1]$. To address this, we argue that truncation of the 
form $\td[\hsigma,\mu]$ is accurate enough for any $\sigma\in[\hsigma/2, 
\hsigma]$. Therefore, it suffices to to discretize $[\sigmin,1]$ into 
$K=\lceil{\log\frac{1}{\sigmin}}\rceil$ levels 
\begin{equation*}
\hsigma_k \defeq 2^{k-1}\sigmin
\end{equation*}
and precompute $\inners{\bx}{\td[\hsigma_k,\mu]^q}$ for every $k\in 
[K]$ in $q\le p$. This allows us to compute $\GetNorm(\sigma)$ in 
$O(p)=\Otil{1}$ time, with $O(n p K)=\Otil{n}$ preprocessing time.

\newcommand{\actind}{i^*_k}
\newcommand{\dmax}{\bar{\delta}_{\max}}

\paragraph{Supporting deletion via lazy offset selection.} 
Had the dataset not supported deletions, we could have simply 
set $\mu$ to be the largest entry of $\bg$ (independent of  $k$). However, 
with  
deletions the largest entry of $\bg$ could change, potentially invalidating 
the truncation. To address this, we maintain a different threshold $\mu_{k}$ for every $k\in[K]$, and argue that the approximation remains valid if the invariant
\begin{equation}\label{eq:scm-mu-invariant}
\dmax \le \mu_k \le \dmax + \frac{R}{2\hsigma_k}
~\mbox{for every $k\in [K]$}
\end{equation}
holds, 
where $\dmax \defeq \max_j \bg_j$.
Writing
\begin{equation}\label{eq:tdk-def}
\td[k] \defeq \td[\hsigma_k, \mu_k] = \max\left\{ \bg - \mu_k, -\frac{R}{\hsigma_k}\right\}
~\mbox{for every}~k\in[K],
\end{equation}
the data structure only needs to maintain $\mu_k$ and $\inners{\bx}{\td[k]^q}$ for every $k\in 
[K]$ in $q\le p$.

When deleting coordinate $j$, for every $k$ we test whether the invariant~\eqref{eq:scm-mu-invariant} remains valid.\footnote{
	We can query the maximum entry of $\bg$ under deletions in $O(1)$ 
	time via a standard data structure, e.g.\ a doubly-linked list of the sorted 
	entries of $\bg$.}
If it does, we keep $\mu_k$ the same and implement deletion (for this value of $k$) in time $O(p) = \Otil{1}$ by subtracting $[\bx]_j [\td[k]]_j^q$ from $\inners{\bx}{\td[k]^q}$ for every $q\le p$. If the invariant is no longer valid, we reset $\mu_k$ to the new value of $\dmax$ and recompute $\inners{\bx}{\td[k]^q}$ for every $q\le p$. Note that the re-computation time is proportional to the number of un-truncated coordinates in the newly defined $\td[k]$. The key observation here is that 
every re-computation decreases $\mu_k$ by at least $R/(2\hsigma_k)$ and so no element of $\bg$ can remain un-truncated for more than two re-computation. Therefore, the cost of recomputing inner products due to deletions, for the entire lifetime of the data structure, is at most $O(n p K)= \Otil{n}$, which we charge to the cost of initialization.

\newcommand{\ks}{{k^{\star}}}

\paragraph{Explicit expression for $\hat{x}[\sigma]$.} 
Following the preceding discussion, for any $\sigma \ge \sigma_{\min}$ we set
\begin{equation*}
\ks = \ceil{\log_2 \frac{\sigma}{\sigma_{\min}}}
,~~\mbox{so that}~~
\sigma \in \left[\frac{\hsigma_{\ks}}{2}, \hsigma_{\ks}\right],
\end{equation*}
and define
\begin{flalign}
Z[\sigma] &\defeq e^{\sigma \mu_\ks} \sum_{q=0}^p \frac{\sigma^q}{q!} \inner{\bx}{\td[\ks]^q} \approx
\norm{\bx \circ \exp\left(\sigma\bg\right)}_1 %
\label{eq:zsigma-def}
\\
\hat{x}[\sigma] & \defeq \frac{e^{\sigma \mu_\ks}}{Z[\sigma]} 
\sum_{q=0}^p \frac{\sigma^q}{q!} \;\bx \circ \td[\ks]^q \approx
\Pi_{\Delta}\left(\bx \circ \exp\left(\sigma\bg\right)\right),
\label{eq:xhat-def}
\end{flalign}
with $\td$ as defined in~\eqref{eq:tdk-def}.

\subsubsection{Correctness}
We now prove that the approximation guarantees of $\ScM$ hold.
\begin{proposition}\label{prop:scm-correct}
	There exist $R = O(1) \cdot\log\frac{n}{\veps\lambda}$ and $p=O(1) \cdot
	\log\frac{n}{\veps\lambda}$ such that for all $\sigma\in [\sigma_{\min},1]$, if the invariant~\eqref{eq:scm-mu-invariant} holds  we have that 
	$Z[\sigma]$ is an $\veps$ multiplicative approximation of $\norm{\bx \circ \exp\left(\sigma\bg\right)}_1$ and 
	$\hat{x}[\sigma]$ is a $\veps$-padding of $\Pi_{\Delta}\left(\bx \circ \exp\left(\sigma\bg\right)\right)$, 
	with $Z[\sigma]$ and $\hat{x}[\sigma]$ defined in Eq.s~\eqref{eq:zsigma-def} and~\eqref{eq:xhat-def} respectively.
\end{proposition}

\begin{proof}
	To simplify notation, we write $\mu = \mu_\ks$ and $\hsigma = \hsigma_{\ks}$. 
	We begin by noting that the inequalities~\eqref{eq:scm-mu-invariant} and $\sigma \le \hsigma$ imply that $\sigma \td_i[\ks] \in [-R, 0]$ for every $i\in [n]$ and we may therefore apply Fact~\ref{fact:approxexp} to obtain
	\begin{equation}\label{eq:scm-elementwise-lb}
	\sum_{q=0}^p \frac{\sigma^q}{q!} \;\bx_j \td_i[\ks]^q 
	\ge  (1-\veps') \bx_i \exp(\sigma \td_i[k]) 
	\ge (1-\veps')  e^{-\sigma \mu} \bx_i \exp(\sigma \bg_i[k]) 
	\end{equation}
	for every $i\in[n]$. Therefore, we have
	\begin{equation}\label{eq:scm-zsigma-lower}
	Z[\sigma] \ge (1-\veps') \norm{\bx \circ \exp\left(\sigma\bg\right)}_1.
	\end{equation}
	Similarly, we have
	\begin{equation*}
	\sum_{q=0}^p \frac{\sigma^q}{q!} \;\bx_j \td_i[\ks]^q 
	\le   (1+\veps') \bx_i \exp(\sigma \td_i[k]) 
	\le  (1+\veps')  e^{-\sigma \mu} \bx_i (\exp(\sigma \bg_i[k]) + \exp(-\sigma R /\hsigma) )
	\end{equation*}
	Note that the condition~\eqref{eq:scm-mu-invariant} also implies that $\td_j[\ks] \ge -R/(2\hsigma)$ for some $j\in[n]$ (namely the maximal element of $\bg$). Using also $\bx_j \ge \lambda$, we have
	\begin{equation*}\label{eq:scm-taylor-upper}
	e^{\sigma \mu} \exp(-\sigma R /\hsigma) \le  \exp(-\sigma R / (2\hsigma)) \frac{\bx_j}{\lambda} \exp(\sigma \bg_j).
	\end{equation*}
	Taking $R \ge 2\log \frac{2n}{\lambda \veps'}$ and recalling that $\sigma \ge \hsigma/2$, we have  $\exp(-\sigma R / (2\hsigma))  \le \lambda \veps'/(2n)$ and consequently 
	\begin{equation*}
	e^{\sigma \mu} \exp(-\sigma R /\hsigma)
	 \le \veps' \bx_j \exp(\sigma \bg_j) \le \frac{\veps' }{n} \norm{\bx \circ \exp\left(\sigma\bg\right)}_1.
	\end{equation*}
	Substituting back and using $\bx_i \le 1$ and $\veps' < 1$ gives
	\begin{equation}\label{eq:scm-elementwise-ub}
	e^{\sigma \mu} \sum_{q=0}^p \frac{\sigma^q}{q!} \;\bx_j \td_i[\ks]^q 
	\le   (1+\veps') \bx_i \exp(\sigma \bg_i) + \frac{\veps'}{n}  \norm{\bx \circ \exp\left(\sigma\bg\right)}_1.
	\end{equation}
	Summing over $i\in[n]$, we obtain
	\begin{equation}\label{eq:scm-zsigma-upper}
	Z[\sigma] \le  (1+2\veps')\norm{\bx \circ \exp\left(\sigma\bg\right)}_1 
	\end{equation}
	Therefore, 
	$Z[\sigma]$ is a $2\veps'$-multiplicative approximation of $\norm{\bx \circ \exp\left(\sigma\bg\right)}_1$.

	It remains to show that $\hat{x}[\sigma]$ is a $\veps$-padding of $x[\sigma] \defeq  \Pi_{\Delta}\left(\bx \circ \exp\left(\sigma\bg\right)\right)$. First, if we define $\tx = \frac{1+2\veps'}{1-\veps'}\hat{x}[\sigma]$ then the bounds~\eqref{eq:scm-elementwise-lb} and~\eqref{eq:scm-zsigma-upper} imply that $\tx \ge x[\sigma]$ elementwise.
	Also, the bounds~\eqref{eq:scm-zsigma-lower} and~\eqref{eq:scm-elementwise-ub} imply that
	\begin{equation*}
	\hat{x}_i[\sigma] - x_i[\sigma] \le \frac{(1+\veps')x_i[\sigma] + \veps'/n}{1-\veps'}
	\end{equation*} 
	for every $i\in [n]$.
	Therefore, for $\veps' < 1/10$, 
	\begin{equation*}
	\lone{\tx - x[\sigma]} \le \left(\frac{1+2\veps'}{1-\veps'}\right)^2 - 1 \le 10\veps',
	\end{equation*}
	so that $\hat{x}[\sigma]$ is a $10\veps'$ padding of $x[\sigma]$. 
	Taking $\veps' = \veps/10$ concludes the proof.
\end{proof}

\subsubsection{Implementation: data structure state and initialization}

Besides storing $\bx$ and $\bg$, the data structure maintains the following fields.
\begin{enumerate}
	\item An offset $\mu_k \in \R$ for every $k\le  K=\ceil{\log\frac{1}{\sigmin}}$, initialized as $\mu_k = \max_j [\bg]_j$ for all $k$.
	\item A balanced binary tree with $n$ leaves. For node $v$ in the tree, $k\in[K]$ and $q\in \{0,\ldots,p\}$, we store
	\begin{equation*}
	A_v[k,q] \defeq \inner{\bx}{\td[k]^q}_{S_v},
	\end{equation*}
	where $\td[k] = \max\{ \bg - \mu_k, -R/\hsigma_k\}$ as before, the set $S_v$ contains the leaves in the subtree rooted in $v$, and $\inner{a}{b}_S \defeq \sum_{i\in S} a_i b_i$. When referring to the root of the tree we omit the subscript, i.e., we write 
	\begin{equation*}
	A[k,q] \defeq \inner{\bx}{\td[k]^q}.
	\end{equation*}
	\item A vector $u\in\R^n$ and coefficients $c_{k,q}\in \R$ for every $k\in[K]$ and $q\in \{0,\ldots,p\}$, for maintaining the running sum. We initialize them all to be $0$.  The running sum obeys the following invariant:
	\begin{equation}\label{eq:scm-running-sum-invariant}
	s = u + \sum_{k=1}^K \sum_{q=0}^p \frac{c_{k,q}}{q!} \bx \circ \td[k]^q.
	\end{equation}
	\item A doubly linked list of the sorted entries of $\bg$, with a pointer to the maximal element of $\bg$ as well as pointers to the largest element smaller than $\mu_k - R/\hsigma_k$ for every $k\in[K]$. 
\end{enumerate}

Initializing the data structure for maintaining the maximum element takes time $O(n\log n)$ due to the need to sort $\bg$. With it, initializing $\mu_k$ is trivial and so is the initialization of $u$ and $c_{q,k}$. Initializing the data stored in the binary tree takes time $O(npK)$, since for every value $k$ and $q$ and internal node $v$ with children $v',v''$ we can recursively compute $A_{v}[k,q]$ as $A_{v'}[k,q]+A_{v''}[k,q]$. We will also charge some additional deletion costs to the initialization runtime, resulting in the overall complexity $O(npK\log n)$.

\subsubsection{Implementation: queries and sampling}

\paragraph{$\GetNorm(\sigma)$.} We compute $\ks = \ceil{\log_2 \frac{\sigma}{\sigma_{\min}}}$ and 
return $Z[\sigma] = e^{\sigma \ks} \sum_{q=0}^p \frac{\sigma^q}{q!} A[\ks,q]$. Clearly, this takes $O(p)$ time and Proposition~\ref{prop:scm-correct} provides the claimed approximation guarantee.

\paragraph{$\Get(j, \sigma)$.} We compute $Z[\sigma]$ and $\ks$ as described above and return $\frac{e^{\sigma \mu_\ks}}{Z[\sigma]} \sum_{q=0}^p \frac{\sigma^q}{q!} \bx_j \td_j[\ks]^q$ in accordance with the form~\eqref{eq:xhat-def} of $\hat{x}[\sigma]$. Again, this takes $O(p)$ time and Proposition~\ref{prop:scm-correct} provides the claimed approximation guarantee.

\paragraph{$\GetSum(j)$.} Recalling the invariant~\eqref{eq:scm-running-sum-invariant}, we return $u_j + \sum_{k=1}^K \sum_{q=0}^p \frac{c_{k,q}}{q!} \bx_j \td_j[k]^q$ in time $O(pK)$. 

\paragraph{$\Sample(\sigma)$.} We perform a random walk from the root of our binary tree data structure to a leaf. A each internal node $v$ with children $v'$ and $v''$, we select node $v'$ with probability
\begin{equation*}
\frac{\inner{\ones}{\hat{x}[\sigma]}_{S_{v'}}}{\inner{\ones}{ \hat{x}[\sigma]}_{S_v}} = \frac{
	\sum_{q=0}^p \frac{\sigma^q}{q!} A_{v'}[\ks,q]
}{
	\sum_{q=0}^p \frac{\sigma^q}{q!} A_{v}[\ks,q]
},
\end{equation*}
and otherwise select $v''$, where $\ks=\ceil{\log_2 \frac{\sigma}{\sigma_{\min}}}$. We return the index associated with the leaf in which we end the walk; the probability of returning index $j$ is exactly $[\hat{x}[\sigma]_j$. Each step in the walk takes time $O(p)$ and there are $O(\log n)$ steps, so the total time is $O(p\log n)$. 

\subsubsection{Implementation: updates}

\paragraph{$\UpdateSum(\sigma)$.} Recalling the invariant~\eqref{eq:scm-running-sum-invariant} and the form~\eqref{eq:xhat-def} of $\hat{x}[\sigma]$, we
compute $\ks$ and $Z[\sigma]$ as in the $\GetNorm$ implementation, and update
 update 
\begin{equation*}
c_{\ks,q} \gets c_{\ks,q} + \frac{e^{\sigma \mu_{\ks}} \sigma^q}{Z[\sigma]}
\end{equation*}
for every $q\in \{0,\ldots, p\}$. This takes time $O(p)$. 

\paragraph{$\Del(j)$.} We set $[\bg_j] \gets -\infty$, remove the element corresponding to index $j$ from the doubly linked list, and perform the following operations for each $k\in[K]$ separately. 
First, we check if the new maximum element of $\bg$ is a least $\mu_k - R/(2\hsigma_k)$. If it is, we leave $\mu_k$ unchanged and we simply update
\begin{equation*}
A_v[k,q] \gets A_v[k,q] - \bx_j \td_j[k]^q
\end{equation*}
for every $q\le p$ and node $v$ on the path from the root to the leaf corresponding to index $j$. Since the length of the path is $O(\log n)$, this update takes time $O(p\log n)$.

Otherwise, the new maximum element is less than $\mu_k - R/(2\hsigma_k)$, and we must change $\mu_k$ in order to maintain the invariant~\eqref{eq:scm-mu-invariant}. Let $\mu_k^\mathrm{new}$ be the new maximum element of $\bg$, and let
\begin{equation*}
U_k = \left\{i ~\middle|~ [\bg]_i \ge  \mu_k^\mathrm{new} + \frac{R}{\hsigma_k} \right\}
\end{equation*}
be the new set of un-truncated indices. (We find the elements in this set when we update the pointer to the first element smaller than $\mu_k - R/\hsigma_k$).  We recompute $A_v[k,q] = \inner{\bx}{\td[k]^q}_{S_v}$ for every $q\le p$ and every node $v$ with a child in $U_k$. Performing the computation recursively from leaf to root, this take at most $O(|U_k|p\log n)$ time. To maintain the invariant~\eqref{eq:scm-running-sum-invariant} as the definition of $\td[k]$ changes, we update
\begin{equation*}
u_j \gets u_j + \sum_{q=0}^p \frac{c_{k,q}}{q!} \bx_j \left(
[\bg_j - \mu_k^{\mathrm{new}}]^q - [\max\left\{\bg_j - \mu_k, -R/\hsigma_k\right\}]^q \right)
~~\mbox{for every }j\in U_k;
\end{equation*}
this update takes $O(|U_k| p)$ time. Finally, we update $\mu_k \gets \mu_k^{\mathrm{new}}$. 

Summing over $k\in[K]$, deletion operations of the first kind (with $\mu_k$ unchanged) take at most $O(Kp\log n)$ time per call to $\Del$. Operations of the second kind (with $\mu_k$ decreased) take time $O( N p \log n)$ throughout the data structure lifetime, where $N=\sum_{t\ge 1}\sum_{k=1}^K |U_k^{(t)}|$ and for each $k\in[K]$ we write $U_{k}^{(1)},U_{k}^{(2)},\ldots$ to denote the different sets $U_k$ generated by all calls to $\Del$. For each $k$, if $\mu_k$ is decreased at all then it must decrease by at least $R/(2\hsigma_k)$. Therefore, by definition of $U_k$, an index $j$ can belong to $U_k^{(t)}$ for at most 2 values of $t$. Consequently, we have $N=O(nK)$. Therefore, deletion operations of the second kind contribute at most $O(nKp\log n)$ to the total runtime, which we charge to initialization.

\section{Applications}
\label{sec:app}

In this section, we leverage the techniques of this paper to obtain improved runtimes for solving certain structured optimization problems.

 In Sections~\ref{app:maxIB} and~\ref{app:minEB}, we use a variant of our variance-reduced coordinate method in the $\ell_2$-$\ell_1$ setup to obtain algorithms for solving the maximum inscribed ball (Max-IB) and minimum enclosing ball (Min-EB) problems. Our algorithms improve upon the runtimes of those in \citet {ZhuLY16} by a factor depending on the sparsity of the matrix. This improvement stems from a preprocessing step in \cite{ZhuLY16}  where the input is randomly rotated to improve a norm dependence of the algorithm. Our methods avoid this preprocessing and obtain runtimes dependent on the both the sparsity and numerical sparsity of the data, providing universal improvements in the sparse regime, in the non-degenerate case where the span of the points is full-rank.

In Section~\ref{app:reg}, we use the results of our variance-reduced algorithm in the $\elltwo$ setup (cf. Section~\ref{ssec:vr-l2}) to obtain improved regression algorithms for a variety of data matrices, including when the matrix is numerically sparse or entrywise nonnegative.

Our methods in this section rely on an extension of the outer loop of this paper (Algorithm~\ref{alg:outerloop}) for \emph{strongly monotone minimax problems}, developed in our previous work \cite{CarmonJST19}. Specifically, for a separable regularizer $r(x, y) = r\x(x) + r\y(y)$ on a joint space for any of our setups, consider the following composite bilinear minimax problem:
\begin{equation}
	\min_{x\in\xset}\max_{y\in\yset}f(x,y)\defeq y^\top Ax+\mu\x \phi
	(x)-\mu\y\psi(y),\text{ where }\phi=V_{x'}\x,\; \psi=V_{y'}\y .\label{eq:def-strongly-monotone}
\end{equation}

We call such problem a $(\mu\x,\mu\y)$-strongly monotone problem; this is a special case of a generalization of the notion of strong convexity, in the case of convex minimization. For general strongly-monotone problems, \citet{CarmonJST19} provided a variant of Algorithm~\ref{alg:outerloop} with the following guarantee.

\begin{proposition}[Proposition 5,~\citet{CarmonJST19}]
\label{prop:outerloopproof-sm}
For problem~\eqref{eq:def-strongly-monotone}, denote $\mu\defeq\sqrt{\mu\x\mu\y}$ and $\rho\defeq \sqrt{\mu\x/\mu\y}$. Let $\mathcal{O}$ be an ($\alpha$,$\veps$)-relaxed proximal oracle for operator $g(x, y) \defeq (\nabla_x f(x, y), -\nabla_y f(x, y))$, let $\Theta$ be the range of $r$, and let $\norm{(\grad_x f(z), -\grad_y f(z'))}_*\le G$, for all $z,z'\in\zset$. Let $z_K$ be the output of $K$ iterations of $\mathtt{OuterLoopStronglyMonotone}$, Algorithm 7 of \citet{CarmonJST19}. Then
\begin{equation*}
\E \gap(z_K) \le 
\sqrt{2}G\sqrt{\left(\left(\frac{\alpha}{\mu+\alpha}\right)^K \left(\rho+\frac{1}{\rho}\right)\Theta 
+\frac{\veps}{\mu}\right)}.
\end{equation*}
Each iteration $k \in [K]$ consists of one call to $\mathcal{O}$, producing a point $z_{k - 1/2}$, and one step of the form
\begin{equation}\label{eq:extragrad_strongly_step} z_k \gets \left\{\inner{g(z_{k - 1/2})}{z} + \alpha \hat{V}_{z_{k - 1}}(z) + \mu \hat{V}_{z_{k - 1/2}}(z)\right\}, \end{equation}
where $\hat{V}\defeq \rho V\x + \rho^{-1}V\y$.  In particular, by setting
\[\veps = \frac{\mu\eps^2}{4G^2},\]
using $K=\Otil{\alpha/\mu}$ iterations, we have the guarantee $\E \gap(z_K)\le\eps$.
\end{proposition}

The ($\alpha$,$\veps$)-relaxed proximal oracle works similarly as in Algorithm~\ref{alg:innerloop-approx} except for the additional composite terms. For completeness we include the algorithm with its theoretical guarantees and implementation in Section~\ref{ssec:comp} (see Algorithm~\ref{alg:innerloop-approx-comp}, Proposition~\ref{prop:innerloopproof-comp} and Section~\ref{ssec:imp-comp}).

In all of our applications discussed in this section, the cost of each step \eqref{eq:extragrad_strongly_step} is $O(\nnz)$, stemming from the computation of $g(x, y)$. The resulting algorithms therefore have runtime
\[\tilde{O}\left(\left(\nnz + (\text{cost of implementing } \mathcal{O})\right)\cdot \frac{\alpha}{\mu}\right).\]

\subsection{Maximum inscribed ball}
\label{app:maxIB}

In the maximum inscribed ball (Max-IB) problem, we are given a polyhedron $P\subset\R^n$ defined by $m$ halfspaces $\{H_i\}_{i \in [m]}$, each characterized by a linear constraint $H_i = \{ x\in \R^n : \langle a_i,x\rangle+b_i\ge 0 \}$, i.e.\ $P = \cap_{i \in [n]} H_i$.
 The goal is to (approximately) find a point $x^*\in P$ that maximizes the smallest distance to any of the bounding hyperplanes $H_i$, i.e.\ 
 \[
x_* \in \argmax_{x \in P} \min_{i \in [n]} \frac{\langle a_i,x\rangle+b_i}{\norm{a_i}_2} ~.
 \]
 More formally, if the optimal radius of the maximum inscribed ball is $r^*$, the goal is to find an $\eps$-accurate solution, i.e.\ a point in $P$ which has minimum distance to all bounding hyperplanes at least $(1-\eps)r^*$.  

Given halfspace information $A$, $b$ where $\ai = a_i$ for all $i \in [m]$, the polytope is defined by $P= \{x \mid Ax+b\ge 0\}$. We use the following notation in this section: $B\defeq \norm{b}_\infty$,  $r^*$ is the value of the maximum inscribed ball problem, $R$ is the radius of the minimum enclosing ball, which is defined as the Euclidean ball containing $P$ with smallest radius possible, $x^*$ is the center of the maximum inscribed ball, and $\rho$ is an upper bound on the aspect ratio $R/r^*$. As in \citet{ZhuLY16}, we will make the following assumptions:
\begin{enumerate}
\item 	The polytope is bounded, and thus $m\ge n$. This is without loss of generality since when the polytope is unbounded, the aspect ratio $\rho=\infty$, and our runtime result holds trivially.
\item $\|\ai\|_2^2=1$ for all $i \in [m]$, so $\norm{A}_{2\rightarrow\infty}=1$, by properly scaling $A$ (one can consider the trivial case when for some $i$, $a_i=0$ separately).
\item The origin is inside polytope $P$, i.e.\ $O\in P$, by properly shifting $P$.
\end{enumerate}

We also define the following constant (see Appendix~\ref{ssec:vr-l2l1}) in this section with respect to the rescaled matrix $A$,
\[
\ltoco \defeq \min\left\{\ltooco,\ltotco,\ltohco\right\}\le \sqrt{\rcs} \cdot \Lrc^{2, 1} \le \sqrt{\rcs},
\]
given the definitions of $\ltooco,\ltotco,\ltohco$ as in~\eqref{eq:def-ltooco1}, \eqref{eq:def-ltooco2}, and~\eqref{eq:def-ltooco3}, and the second assumption above (namely, that $\Lrc^{2, 1} = \max_{i \in [m]} \norm{\ai}_2 = 1$).

\citet{ZhuLY16} show that solving Max-IB is equivalent to the following minimax problem:
\begin{equation}
r^*\defeq\max_{x\in\R^n}\min_{y\in\Delta_m}f(x,y)\defeq y^\top Ax+y^\top b,\label{eq:MaxIB}
\end{equation}
and moreover, to solve the problem to $\eps$-multiplicative accuracy, it suffices to find $x^*_\eps$ that solves the minimax problem to $\eps$-multiplicative accuracy in terms of the one-sided gap of the $x$ block, i.e.\ 
$$\min_{y\in\Delta_m}f(x^*_\eps,y)\ge(1-\eps)f(x^*,y^*),$$
 where $(x^*,y^*)$ is the optimal saddle point of problem \eqref{eq:MaxIB}. We first state several bounds on the parameters of the problem from \citet {ZhuLY16}. 
 
\begin{fact}[Geometric properties of Max-IB]
\label{lem:maxIB-fact}
We have $\ltwo{x^*}\le 2R$, and 
\[r^*=\max_{x\in\R^n}\min_{y\in\Delta_m}f(x,y)\defeq y^\top Ax+y^\top b\le B\le 2R.\] 
\end{fact}

These facts imply that we can instead consider the constrained minimax problem (where we overload our definition of $f$ for the rest of the section):
\begin{align}
r^*\defeq\max_{x\in\ball^n}\min_{y\in\Delta_m}f(x,y) 
& = y^\top \tilde{A}x+y^\top b,~~\text{ where }\tilde{A}=2R\cdot A. \label{eq:MaxIB-adapted}
\end{align}

We first use a ``warm start'' procedure to find a constant multiplicative estimate of $r^*$, which uses the strongly monotone algorithm $\mathtt{OuterLoopStronglyMonotone}$ of \citet{CarmonJST19} together with Algorithm~\ref{alg:innerloop-approx-comp} of Section~\ref{ssec:comp} as a relaxed proximal oracle on the  $(\mu, \mu)$-strongly monotone problem \[\max_{x\in\ball^n}\min_{y\in\Delta_m} f_\mu(x,y)\defeq y^\top \tilde{A}x+y^\top b +\mu\sum_{i\in[m]}[y]_i\log[y]_i-\frac{\mu}{2}\ltwo{x}^2,\] and a line search over parameter $\mu$. The following lemma is an immediate consequence of Proposition~\ref{prop:outerloopproof-sm} and Corollary~\ref{prop:innerloopproof-comp}, whose proof we defer to Appendix~\ref{app:maxIBproofs}.

\begin{lemma}\label{lem:maxIB-preprocess}
We can spend $\Otil{\nnz + \rho\sqrt{\nnz}\cdot \ltoco}$ time preprocessing to obtain a $8$-multiplicative approximation $\hat{r}$ of $r^*$, i.e.\ 
\[\frac{\hat{r}}{8} \le r^* \le \hat{r}.\]
\end{lemma}

Finally, we use our variance-reduced coordinate algorithm, namely Algorithm~\ref{alg:innerloop-approx-comp} as a relaxed proximal oracle in $\mathtt{OuterLoopStronglyMonotone}$ together with Proposition~\ref{prop:outerloopproof-sm} once more to solve \eqref{eq:MaxIB-adapted} to the desired accuracy. The implementation in Section~\ref{ssec:imp-comp} and complexity results in Section~\ref{ssec:vr-l2l1} yield the runtime. This implementation crucially uses our development of the $\AEM$ data structure in order to obtain a runtime depending directly on $\rcs$ rather than dimensions of the matrix, as well as independence on $B$. For completeness, a proof can be found in Appendix~\ref{app:maxIBproofs}.

\begin{theorem}
\label{thm:MaxIB}
The algorithm of Section~\ref{ssec:vr-l2l1} can be used to find an $\eps$-accurate solution $x_\eps^*$ to Max-IB satisfying $\min_{y\in\Delta_m}f(x^*_\eps,y)\ge(1-\eps)r^*$ with high probability in time~\footnote{Here $\widetilde{O}$ is hiding an additional factor of $\text{polylog}(\|b\|_{\infty})$ due to the additional cost in the runtime of $\AEM$, caused by the linear term $b$ (see Remark~\ref{rem:polylogbc}).} $$\widetilde{O}\left(\nnz + \frac{\rho\sqrt{\nnz} \cdot \ltoco
		}{\epsilon}\right)= \widetilde{O}\left(\nnz + \frac{\rho\sqrt{\nnz\cdot\rcs}
	}{\epsilon}\right).$$
\end{theorem}

\begin{remark}
	Because we assumed $m\ge n$, 
in the case $A$ is dense, up to logarithmic terms our runtime improves upon the runtime of $\Otil{\rho m\sqrt{n}/\eps}$ in \citet {ZhuLY16} by a factor of at least
\[\sqrt{\frac{mn}{\nnz}\cdot\frac{m}{\rcs}}\] generically. This is an improvement when $A$ is sparse or column-sparse, i.e.\ $\nnz\ll 
mn$, or $\rcs\ll m$. Such a saving is larger when $A$ has numerical sparsity so that e.g.\ 
$\ltocop^2 \le \max_{i\in[m]}\lones{\ai}^2 + \left(\max_{i \in [m]} \norm{\ai}_1\right)\left(\max_{j \in [n]} \norm{\aj}_1\right)< \rcs\cdot\max_{i\in[m]}\norm{A_i}_2^2$.
\end{remark}

\subsection{Minimum enclosing ball}
\label{app:minEB}

In the minimum enclosing ball (Min-EB) problem, we are given a set of data points $\{a_1,\ldots,a_m\}$ with $a_1=0,\max_{i \in [m]}\|a_i\|= 1$.\footnote{This can be assumed without loss of generality by shifting and rescaling as in \citet{ZhuLY16} and considering the trivial case when all $a_i, i\in[m]$ are equal.} The goal is to find the minimum radius $R^*$ such that there exists a point $x$ with distance at most $R^*$ to all points. Following the presentation of~\citet {ZhuLY16}, we consider Min-EB in an equivalent form. Define the vector $b$ to have $b_i=\half\norm{a_i}_2^2$ entrywise. Then, Min-EB is equivalent to the minimax problem
\begin{align}
	R^*\defeq\min_{x\in\R^n}\max_{y\in\Delta^{m}}\frac{1}{2}\sum\limits_{i}y_i\|x-a_i\|_2^2 = \min_{x\in\R^n}\max_{y\in\Delta^m} f(x,y), \text{ where } f(x, y) \defeq y^\top Ax+y^\top b+\frac{1}{2}\norm{x}_2^2.
	\label{eq:minEB}
\end{align}
By assumption, $\nA=1$. We let $(x^*,y^*)$ be the optimal solution to the saddle point problem.  We first state several bounds on the quantities of the problem. These bounds were derived in \citet {ZhuLY16} and obtained by examining the geometric properties of the problem. 

\begin{fact}
\label{lem:minEB-fact}
The following bounds hold: $\norm{x^*}_2\le 1$, and $R^*\ge1/8$.
\end{fact}

To achieve a multiplicative approximation, since $R^*\ge1/8$ by Fact~\ref{lem:minEB-fact}, it suffices to obtain a pair $(x^*_\eps,y^*_\eps)$ achieving $\max_y f(x^*_\eps,y)-\min_x f(x,y^*_\eps)\le \eps/8$. In light of minimax optimality, Lemma~\ref{lem:minEB-reg} (proved in Section~\ref{sec:proofs_from_62}) shows that it suffices to consider, for $\eps'=\Theta(\eps/\log m)$, solving the following $(1,\eps')$-strongly monotone problem to sufficient accuracy:
\begin{align}
\min_{x\in\R^n}\max_{y\in\Delta^m} f_{\eps'}(x,y)\defeq y^\top Ax+y^\top b-\eps'\sum_{i \in [m]} [y]_i\log[y]_i+\frac{1}{2}\norm{x}_2^2.
\label{eq:minEB-adapted}
\end{align}

\begin{lemma}\label{lem:minEB-reg}
Setting $\eps'=\eps/(32\log m)$, an $\eps/16$-accurate solution or (\ref{eq:minEB-adapted}) is an $\eps/8$-accurate solution to the original problem (\ref{eq:minEB}).
\end{lemma}

As an immediate result of the above lemma, the runtime in Section~\ref{ssec:vr-l2l1} and the correctness proofs of Proposition~\ref{prop:outerloopproof-sm} and Corollary~\ref{prop:innerloopproof-comp}, we obtain the following guarantee.

\begin{theorem}
	\label{thm:MinEB}	
	The strongly monotone algorithm $\mathtt{OuterLoopStronglyMonotone}$ of \citet{CarmonJST19}, using Algorithm~\ref{alg:innerloop-approx-comp} of Section~\ref{ssec:comp} and the estimator of Section~\ref{ssec:vr-l2l1} as a relaxed proximal oracle, finds an $\eps$-accurate solution $x_\eps^*$ to Min-EB satisfying $R^*\le \max_y f(x^*_\eps,y)\le (1+\eps)R^*$ with high probability in time%
	 \[\Otil{\nnz+\frac{\sqrt{\nnz}\cdot \ltoco}{\sqrt{\eps}}} = \Otil{\nnz+\frac{\sqrt{\nnz \cdot \rcs}}{\sqrt{\eps}}}.\]
	
\end{theorem}

\begin{remark}
	When $m \geq n$,\footnote{When $m<n$, the runtime of the algorithm in \citet {ZhuLY16} still holds and is sometimes faster than ours. } up to logarithmic terms our runtime improves the $\Otil{m\sqrt{n}/\sqrt{\eps}}$ runtime of \citet {ZhuLY16} by a factor of 
\[\sqrt{\frac{mn}{\nnz}\cdot\frac{m}{\rcs}}\] generically. This is an improvement when $A$ is sparse or column-sparse, i.e.\ $\nnz\ll 
mn$, or $\rcs\ll m$. As in  Section~\ref{app:maxIB}, the improvement is larger when $A$ is numerically sparse, i.e.\ when $\ltocop^2 \le \max_{i\in[m]}\lones{\ai}^2 + \left(\max_{i \in [m]} \norm{\ai}_1\right)\left(\max_{j \in [n]} \norm{\aj}_1\right)< \rcs\cdot\max_{i\in[m]}\norm{A_i}_2^2$.
\end{remark}

\subsection{Regression}
\label{app:reg}

We consider the standard $\ell_2$ linear regression problem in a data matrix $A\in\R^{m\times n}$ and vector $b \in \R^m$, i.e.\ $\min_{x\in\R^n}\ltwo{Ax-b}$. In particular, we consider the equivalent primal-dual form,
\begin{equation}\label{def:minimax-reg}
	\min_{x\in\R^n}\max_{y\in\ball^m} f(x,y)\defeq y^\top(Ax-b).
\end{equation}
Throughout, we assume the smallest eigenvalue of $A^\top A$ is $\mu>0$ and denote an optimal solution to~\eqref{def:minimax-reg} by $z^* = (x^*,y^*)$ (where $x^*$ is the unique solution to the regression problem). Our strategy is to consider a sequence of modified problems, parameterized by $\beta>0$, $x'\in\R^n$:
\begin{equation}\label{def:minimax-reg-mod}
	\min_{x\in\R^n}\max_{y\in\ball^m}f^\beta_{x'}(x,y)\defeq y^\top(Ax - b)+\frac{\beta}{2}\ltwo{x-x'}^2-\frac{\beta}{2}\ltwo{y}^2.
\end{equation}
We denote the optimal solution to \eqref{def:minimax-reg-mod} by $z^*_{(\beta, x')} = (x_{(\beta,x')}^*,y_{(\beta,x')}^*)$; when clear from context, for simplicity we drop $\beta$ and write $z^*_{x'} = (x_{x'}^*,y_{x'}^*)$ (as $\beta = \sqrt{\mu}$ throughout our algorithm). Lemma~\ref{lem:reg-opt-rel} (proved in Section~\ref{ssec:app-reg}) states a known relation between the optimal solutions for~\eqref{def:minimax-reg} and~\eqref{def:minimax-reg-mod}.
 
\begin{lemma}
\label{lem:reg-opt-rel}
Letting $(x^*,y^*)$ be the optimal solution for~\eqref{def:minimax-reg} and $(x^*_{x'},y^*_{y'})$ be the optimal solution for~\eqref{def:minimax-reg-mod}, the following relation holds:
\[
\ltwo{x^*_{x'}-x^*}\le\frac{1}{1+\frac{\mu}{\beta^2}}\ltwo{x'-x^*}.
\]	
\end{lemma}

We give a full implementation of the regression algorithm in Algorithm~\ref{alg:reg} (see Section~\ref{ssec:app-reg}), and state its correctness and runtime in Theorem~\ref{thm:reg}. The algorithm repeatedly solves problems of the form \eqref{def:minimax-reg-mod} in phases, each time using Lemma~\ref{lem:reg-opt-rel} to ensure progress towards $x^*$. Observing that each subproblem is $(\beta, \beta)$-strongly monotone, each phase is conducted via $\mathtt{OuterLoopStronglyMonotone}$, an algorithm of \citet{CarmonJST19}, using the $\elltwo$ algorithms of Section~\ref{ssec:vr-l2} as a proximal oracle. Due to the existence of composite terms, our inner loop steps are slightly different than in Section~\ref{ssec:vr-l2}; we give a more formal algorithm for the relaxed proximal oracle and its implementation in Algorithm~\ref{alg:innerloop-approx-comp} and Appendix~\ref{ssec:comp}. We remark that by a logarithmic number of restarts per phase, a standard argument boosts Theorem~\ref{thm:reg} to a high-probability claim.

\begin{restatable}{theorem}{restateregthm}
	\label{thm:reg}	
	
	Given data matrix $A\in\R^{m\times n}$, vector $b\in\R^m$, and desired accuracy $\eps\in(0,1)$, assuming $A^\top A\succeq\mu I$ for $\mu>0$, Algorithm~\ref{alg:reg} outputs an expected $\eps$-accurate solution $\tilde{x}$, i.e.\ 
	\[
	\E\left[ \ltwo{\tilde{x}-x^*} \right]\le \eps,
	\]
	and runs in time
	\[
	\Otilb{\nnz+\sqrt{\nnz}\cdot \frac{\max\left\{\sqrt{\sum_i\lones{\ai}^2}, \sqrt{\sum_j\lones{\aj}^2}\right\}}{\sqrt{\mu}}}.
	\]
\end{restatable}

We give two settings where the runtime of Algorithm~\ref{alg:reg} improves upon the state of the art.

\paragraph{Entrywise nonnegative $A$.} In the particular setting when all entries of $A$ are nonnegative,\footnote{More generally, this holds for arbitrary $A\in\R^{m \times n}$ satisfying $\||A|\|_{\mathrm{op}}\le\|A\|_{\mathrm{op}}$.} by Proposition~\ref{prop:reg-improve} our complexity as stated in Theorem~\ref{thm:reg} improves by a factor of $\sqrt{\nnz/ (m+n)}$ the runtime of accelerated gradient descent~\citep{Nesterov83}, which is the previous state-of-the-art in certain regimes with runtime $\O{\nnz\cdot\|A\|_{\mathrm{op}}/\sqrt{\mu}}$. This speedup is most beneficial when $A$ is dense. 

\paragraph{Numerically sparse $A$.} For numerically sparse $A$ with $\lones{\ai}/\ltwo{\ai}= O(1)$, $\lones{\aj}/\ltwo{\aj}=O(1)$ for all $i\in[m]$, $j\in[n]$, we can choose $\alpha=\sqrt{\mu}$ in Algorithm~\ref{alg:reg} and obtain the runtime
\[
\O{\nnz+\frac{\max\bigl\{\sum_i\lones{\ai}^2, \sum_j\lones{\aj}^2\bigr\}}{\mu}}
=
\O{\nnz+\frac{\lfro{A}^2}{\mu}}
\] 
using the argument in Theorem~\ref{thm:reg}. Under a (similar, but weaker) numerically sparse condition $\lones{\ai}/\ltwo{\ai}=O(1)$, the prior state-of-the-art stochastic algorithm~\cite{Johnson13} obtains a runtime of $O(\nnz+\rcs\cdot\lfro{A}^2/\mu)$, and the recent state-of-the-art result in the numerically sparse regime~\cite{GuptaS18} improves this to $O(\nnz+(\lfro{A}^2/\mu)^{1.5})$ when $\rcs = \Omega(\lfro{A}/\sqrt{\mu})$. Improving universally over both, our method gives $O(\nnz+\lfro{A}^2/\mu)$ in this setting.

\section*{Acknowledgements}
This research was supported in part by Stanford Graduate Fellowships, NSF CAREER Award CCF-1844855, NSF Graduate Fellowship DGE-1656518 and a PayPal research gift. We thank the anonymous reviewers who helped improve the completeness and readability of this paper by providing many helpful comments.

\newpage
\arxiv
{
\bibliographystyle{abbrvnat} 
\setlength{\bibsep}{6pt}

}

\newpage

\notarxiv{\appendices}
\arxiv{
\addtocontents{toc}{\protect\setcounter{tocdepth}{1}}

\appendix
\part*{Appendix}
}
\section{Deferred proofs from Section~\ref{sec:prelims}}
\label{app:prelims-proofs}

\begin{proof}[Proof of Proposition~\ref{prop:local-norm}]
It is clear that our choices of $\xset, \yset$ are compact and convex, and 
that the local norms we defined are indeed norms (in all cases, they are 
quadratic norms). Validity of our choices of $\Theta$ follow from the 
well-known facts that for $x \in \Delta^n$, the entropy function $\sum_{j 
\in [n]} x_j \log x_j$ is convex with range $\log n$, and that for $x \in \ball^n$, 
$\half\norm{x}_2^2$ is convex with range $\half$. 

In the Euclidean case we have that $V_x(x')=\half\ltwo{x-x'}^2$ and  
\eqref{eq:strong-convexity} follows from the Cauchy-Schwarz and 
Young inequalities:
\begin{equation*}
 \inner{\gamma}{x' - x} \le \half\ltwo{\gamma}^2 + \half\ltwo{x-x'}^2.
\end{equation*}
Similarly, for the simplex we have that entropy is 1-strongly-convex with respect to
$\lone{\cdot}$ and therefore $V_y(y') \ge \half\lone{y-y'}^2$, and we obtain~\eqref{eq:strong-convexity} from the H{\"o}lder and Young inequalities, 
\begin{equation*}
\inner{\gamma}{y' - y} \le \half\linf{\gamma}^2 + \half\lone{y-y'}^2,
\end{equation*}
where we note that $\linf{\gamma}=\norm{\gamma}_*$ in this case.

Finally, $\clip(\cdot)$ is not the identity only when the corresponding domain is the simplex, in which case entropy 
satisfies 
the local norms bound~\cite[cf.][Lemma 13]{CarmonJST19},
\begin{equation*}
\inner{\gamma}{y-y'} - V_{y}(y') \le  
\sum_{i\in [m]} \gamma_i^2 y_i~~\text{for all 
}y,y'\in\Delta^m~\text{and $\gamma\in\R^m$ such that }\linf{\gamma}\le 
1.
\end{equation*}
Noting that $\linf{\clip(\gamma)}\le 1$ and that 
$\norm{\clip(\gamma)}_y\le \norm{\gamma}_y$ for all $y\in\Delta^m$, 
we have the desired local bound~\eqref{eq:strong-convexity-local}. Finally, 
for every coordinate $i\in[m]$ we have
\begin{equation*}
| \gamma_i - [\clip(\gamma)]_i | = | |\gamma_i| - 1 | \; \indic{|\gamma_i|>1} \le  |\gamma_i| 
\;\indic{|\gamma_i|>1} 
\le |\gamma_i|^2.
\end{equation*}
Consequently, $|\inner{\gamma-\clip(\gamma)}{z}|\le \sum_{i\in [m]} 
\gamma_i^2 z_i$, giving the distortion bound~\eqref{eq:clipping-distortion}.
\end{proof} %

\section{Deferred proofs from Section~\ref{sec:framework}}
\label{sec:framework-proofs}

\subsection{Proof of Proposition~\ref{prop:sublinear}}

In this section, we provide a convergence result for mirror descent under 
local norms. We require the following well-known regret bound for mirror descent.
\begin{lemma}[{\cite[][Lemma 12]{CarmonJST19}}]\label{lem:mirror-descent}
	Let $Q:\zset\to\R$ be convex, let $T\in\N$,
	$z_0\in\zset$ and 
	$\gamma_0,\gamma_1,\ldots,\gamma_T\in\zset^*$. The sequence 
	$z_1,\ldots,z_T$ defined 
	by
	\begin{equation*}
	z_{t} = \argmin_{z\in\zset}\left\{ 
	\inner{\gamma_{t-1}}{z}+Q(z)+V_{z_{t-1}}(z)
	\right\}
	\end{equation*}
	satisfies for all $u\in\zset$ (denoting $z_{T+1}\defeq u$),
	\begin{flalign}
	\sum_{t = 0}^T \inner{\gamma_{t} }{z_t - u} + \sum_{t=1}^T\inner{\grad Q(z_t)}{z_t -u} & \le 
	V_{z_0}(u) +
	\sum_{t=0}^{T }
	\{\inner{\gamma_t}{z_t-z_{t+1}} - V_{z_t}(z_{t+1})\}.
	\label{eq:mirror-descent}
	\end{flalign}
\end{lemma} 

The proposition follows from this regret bound, the properties of the local 
norm setup, and the ``ghost iterate'' argument due 
to~\cite{NemirovskiJLS09}.

\newcommand{\halfs}{\tfrac{1}{2}}

\restatesublinear*
\begin{proof}
	Defining
	\begin{equation*}
	\tilde{\Delta}_t \defeq g(z_t) - \frac{1}{\eta}\clip(\eta \tilde{g}(z_t))
	\end{equation*}
	and the ghost iterates
	\begin{equation*}
	s_t = \argmin_{s\in\zset}\left\{\inner{\half \eta \tilde{\Delta}_{t-1}}{s} 
	+ V_{s_{t - 1}}(s) \right\}
	~~\text{with}~~s_0 = w_0,
	\end{equation*}
	we  rearrange the regret as
	\begin{align}\label{eq:reg-decomp}
	\eta\sum_{t=0}^T\inner{g(z_t)}{z_t - u} & \le 
	\sum_{t=0}^T
	\inner{\clip(\eta\tilde{g}(z_t))}{z_t-u} +
	\sum_{t =0}^T\inner{\eta \tilde{\Delta}_t}{s_t - u} + 
	\sum_{t=0}^T\inner{\eta \tilde{\Delta}_t}{z_t - s_t},
	\end{align}
	and bound each term in turn.

	We first apply Lemma~\ref{lem:mirror-descent} with $Q=0$ and 
	$\gamma_t = \clip(\eta \tilde{g}(z_t))$, 
	using~\eqref{eq:strong-convexity-local} 
	to conclude that
	\begin{equation}\label{eq:mirror-descent-one-step}
	\sum\limits_{t=0}^T
	\inner{\clip(\eta\tilde{g}(z_t))}{z_t-u}
	\le V_{z_0}(u) + 
	\eta^2\sum\limits_{t=0}^T\norm{\tilde{g}(z_t)}_{z_t}^2,~\mbox{for
		all } u\in\zset.
	\end{equation}	
	Next, we apply 
	Lemma~\ref{lem:mirror-descent} again, this time with 
	$\gamma_t=\halfs\eta 
	\tilde{\Delta}_t$, to obtain  the regret bound
	\begin{flalign}\label{eq:ghost-regret-1}
	\sum_{t=0}^T  \inners{\eta \tilde{\Delta}_t}{s_t-u}
	&\le 
	2V_{z_0}(u)
	+\sum_{t=0}^T 	\left\{\inner{\eta \tilde{\Delta}_t}{s_t-s_{t+1}} - 
	2V_{s_t}(s_{t+1})\right\} \nonumber \\
	& \le  2V_{z_0}(u) +\eta^2\sum_{t=0}^T \left\{ 
	\norm{\tilde{g}(z_t)}_{s_t}^2 
	+ \half\norm{g(z_t)}_*^2\right\},
	\end{flalign} 
	for all $u\in\zset$, where we used 
	\begin{flalign*}
	\inner{\eta \tilde{\Delta}_t}{s_t-s_{t+1}} - 
	2V_{s_t}(s_{t+1})
	\le &\inner{\eta g(z_t)}{s_t-s_{t+1}} - 
	V_{s_t}(s_{t+1}) \\&+ |\inner{\clip(\eta 
		\tilde{g}(z_t))}{s_t-s_{t+1}}| - 
	V_{s_t}(s_{t+1}),
	\end{flalign*}
	and then appealed to the bounds~\eqref{eq:strong-convexity} 
	and~\eqref{eq:strong-convexity-local} in the definition of the local norm setup. Now, substituting~\eqref{eq:mirror-descent-one-step} and~\eqref{eq:ghost-regret-1} into~\eqref{eq:reg-decomp}, maximizing over $u$, and taking an expectation, we obtain 
	\begin{equation}\label{eq:tookexpect}
	\begin{aligned}
	\E \sup_{u\in\zset} \sum_{t =0}^T\inner{\eta g(z_t)}{z_t - u}
	&\le 3\Theta + \eta^2 \E \sum_{t=0}^T
	\left\{   \norm{\tilde{g}(z_t)}_{z_t}^2  + \norm{\tilde{g}(z_t)}_{s_t}^2 
	+ \halfs\norm{g(z_t)}_*^2 \right\} \\
	&+ \E \sum_{t = 0}^T \inner{\eta\tilde{\Delta}_t}{z_t - 
		s_t}.
	\end{aligned}
	\end{equation}
	To bound the last term we use the fact that 
	$g(z_t)= \E\left[\tilde{g}(z_t) \mid z_t, s_t\right]$ (which follows from 
	the first part of Definition~\ref{def:sublinear}). We then write
	\begin{equation}\label{eq:ghost-regret-2}
	\left|\E \inner{\eta \tilde{\Delta}_t}{z_t - s_t}\right|
	\le \E \left|\inner{\eta\tilde{g}(z_t)-\clip(\eta\tilde{g}(z_t))}{z_t - 
		s_t}\right|
	\le \eta^2\norm{\tilde{g}(z_t)}_{z_t}^2 + 
	\eta^2\norm{\tilde{g}(z_t)}_{s_t}^2,
	\end{equation}
	where the first inequality is by Jensen's inequality, and the last is due to the 
	property~\eqref{eq:clipping-distortion} of the local norm setup. Substituting \eqref{eq:ghost-regret-2} into \eqref{eq:tookexpect}, we obtain %
	\begin{equation*}
	\E \sup_{u\in\zset} \sum_{t =0}^T\inner{\eta g(z_t)}{z_t - u}
	\le 3\Theta + \eta^2 \E \sum_{t=0}^T
	\left\{  2 \norm{\tilde{g}(z_t)}_{z_t}^2  + 2\norm{\tilde{g}(z_t)}_{s_t}^2 
	+ \half\norm{g(z_t)}_*^2 \right\}.
	\end{equation*}
	Finally, using the second moment bound of local gradient estimator 
	(Definition~\ref{def:sublinear}) and its consequence 
	Lemma~\ref{lem:sublinear-implication}, we may bound each of the 
	expected squared norm terms by $L^2$. Dividing through by $\eta (T+1)$ 
	gives
	\begin{equation*}
	\E\,\sup_{u\in\zset}\left[\frac{1}{T+1}\sum_{t=0}^T
	\inner{g(z_t)}{z_t-u} \right]
	\le \frac{3\Theta}{\eta (T+1)}+\frac{9\eta L^2}{2}.
	\end{equation*}
	Our choices $\eta=\frac{\epsilon}{9L^2}$ and  
	$T\ge\tfrac{6\Theta}{\eta\eps}$ imply that the right hand side is at most 
	$\epsilon$, as required.
\end{proof}

\subsection{Proof of Proposition~\ref{prop:outerloopproof}}
	
\restateouterloopproof*
\begin{proof}
	For some iteration $k$, we have by the optimality conditions on $z_k^\star$ that
	\begin{equation*}
	\inner{g(z_{k-1/2})}{z_k^\star - u} \leq \alpha\left(V_{z_{k - 1}}(u) - V_{z_k^\star}(u) 
	- 
	V_{z_{k - 1}}\left(z_k^\star\right)\right)~~\forall u\in\zset.
	\end{equation*}
	Summing over $k$, writing \arxiv{$\inner{g(z_{k-1/2})}{z_{k}^\star - u} = 
		\inner{g(z_{k-1/2})}{z_{k-1/2} - u} 
		- \inner{g(z_{k-1/2})}{z_{k-1/2} - z_{k}^\star}$,}
	\notarxiv{
		\begin{equation*}
		\inner{g(z_{k-1/2})}{z_{k} - u} = 
		\inner{g(z_{k-1/2})}{z_{k-1/2} - u} 
		- \inner{g(z_{k-1/2})}{z_{k-1/2} - z_{k}}
		\end{equation*}
	} and rearranging yields
	\begin{equation}
	\label{eq:outerloop}
	\begin{aligned}
	\sum_{k=1}^K\inner{g(z_{k-1/2})}{z_{k-1/2} - u} \leq & \alpha V_{z_0}(u)  + \sum_{k=1}^K\alpha\left(V_{z_k}(u)-V_{z_k^\star}(u)\right) \\
	& + \sum_{k=1}^K \left( \inner{g(z_{k-1/2})}{z_{k-1/2} - z_k^\star}- \alpha 
	V_{z_{k - 1}}\left(z_k^\star\right)\right),
	\end{aligned}
	\end{equation}
	for all $u\in\zset$. Since $z_0$ minimizes $r$, the first term is bounded by $V_{z_0}(u) \le r(u)-r(z_0) \le \Theta$. The second term is bounded by the definition of $z_k$ in Algorithm~\ref{alg:outerloop}:
	\[\sum_{k = 1}^K \alpha\left(V_{z_k}(u) - V_{z_k^{\star}}(u)\right) \le K(\vepsout-\vepsi).\] 
	Thus, maximizing \eqref{eq:outerloop}
	over $u$ and then taking an expectation yields
	\begin{equation*}
	\begin{aligned}
	\E\max_{u\in\zset}\sum_{k=1}^K\inner{g(z_{k-1/2})}{z_{k-1/2} - u} &\leq 
	\alpha\Theta  + K(\vepsout-\vepsi) \\
	&+ 
	\sum_{k=1}^K \E \left[ \inner{g(z_{k-1/2})}{z_{k-1/2} - z_{k}^\star}- \alpha 
	V_{z_{k - 1}}\left(z_k^\star\right)\right].
	\end{aligned}
	\end{equation*}
	Finally, by Definition~\ref{def:alphaprox}, $\E \left[ 
	\inner{g(z_{k-1/2})}{z_{k-1/2} - z_{k}^\star}- \alpha 
	V_{z_{k - 1}}(z_k^\star)\right] \le \vepsi$ for every $k$, and the result follows 
	by dividing by $K$.
\end{proof}

\subsection{Proof of Proposition~\ref{prop:innerloopproof}}

We provide a convergence result for the variance-reduced stochastic mirror descent scheme in Algorithm~\ref{alg:innerloop-approx}. We first state the following helper bound which is an application of Lemma~\ref{lem:varshrink}. It is immediate from the variance bound of local-centered estimators (Property 2 of Definition~\ref{def:vr}) and the fact that all local norms (whether the domains are balls or simplices) are quadratic. 
\begin{lemma}\label{lem:centered-var}
	For any $w \in \zset$, $(L,\epsilon)$-centered-local estimator $\tilde{g}_{w_0}$ satisfies
	$$\E \norm{\tilde{g}_{w_0}(z) - g(z)}_{w}^2\le L^2 V_{w_0}(z).$$
\end{lemma}

\begin{lemma}
\label{lem:varshrink}
Let $\norm{\cdot}_D$ be a quadratic norm in a diagonal matrix, e.g. for some $D = \diag(d)$ and $d \geq 0$ entrywise, let $\norm{x}_D^2 = \sum d_i x_i^2$. Then, if $X$ is a random vector, we have
$$\E\norm{X - \E[X]}_D^2 \leq \E\norm{X}_D^2.$$
\end{lemma}
\begin{proof}
This follows from the definition of variance:
$$ \E\norm{X - \E[X]}_D^2 = \E\norm{X}_D^2 - \norm{\E[X]}_D^2 \leq  \E\norm{X}_D^2.$$
\end{proof}

\restateinnerloop*

\begin{proof}
	
	For any $u\in\zset$, and defining $\tilde{\Delta}_t \defeq  \tilde{g}(\hat{w}_t)- g(w_0)$ and $\Delta_t \defeq g(\hat{w}_t)- g(w_0)$, we have
	\begin{equation}\label{eq:boundthreeterms}
	\begin{aligned}
	\sum_{t  \in [T]}\inner{\eta g(w_t)}{w_t-u} &=  \sum_{t \in [T]} \inner{\clip(\eta \tilde{\Delta}_t)+\eta g(w_0)}{w_t-u} + \sum_{t \in [T]} \inner{\eta \Delta_t-\clip(\eta\tilde{\Delta}_t)}{w_t-u} \\
	&+ \sum_{t \in [T]} \inner{\eta g(w_t) - \eta g(\hat{w}_t)}{w_t-u}.
	\end{aligned}
	\end{equation}

	We proceed to bound the three terms on the right hand side of \eqref{eq:boundthreeterms} in turn. For the first term, recall the guarantees for the ``ideal'' iterates of Algorithm~\ref{alg:innerloop-approx}, 
	\begin{align*}
	w_t^\star = 
	\argmin_{w\in\zset}\left\{\inner{\clip(\eta \tilde{\Delta}_t)+\eta g(w_0)}{w} + \frac{\alpha\eta}{2}V_{w_0}(w) + V_{w_{t - 1}}(w) 
	\right\}.
	\end{align*}
	By using the optimality conditions of these iterates, defining $Q(z) \defeq \inner{\eta g(w_0)}{z} + 
	\tfrac{\alpha\eta}{2}V_{w_0}(z)$, $\gamma_t \defeq \clip(\eta \tilde{\Delta}_t)$, and defining for notational convenience $w_{T + 1}^\star \defeq u$,
	\begin{equation}\label{eq:idealmirror}
	\begin{aligned}
	\sum_{t \in [T]}\inner{\gamma_{t-1} + \grad Q(w_t^\star )}{w_t^\star -u} & \le \sum_{t\in[T]}\inner{-\nabla V_{w_{t-1}}(w_t^\star)}{w_t^\star-u}\\
	& = \sum_{t\in[T]}\left(V_{w_{t-1}}(u)-V_{w_t^{\star}}(u)-V_{w_{t-1}}(w_t^\star)\right)\\
	& = V_{w_0}(u) + \sum_{t \in [T]}\left(V_{w_t}(u)-V_{w_t^\star }(u)\right) - \sum_{t = 0}^{T} V_{w_{t}}(w_{t+1}^\star ).
	\end{aligned}
	\end{equation}
	We thus have the chain of inequalities, recalling $\gamma_0 = 0$,
	\begin{equation}\label{eq:prop-vr-first-1}
	\begin{aligned}
	& \sum_{t \in [T]}\inner{\clip(\eta \tilde{\Delta}_t)+\eta g(w_0)}{w_t-u}+\frac{\alpha\eta}{2}\sum_{t\in [T]}\inner{\grad V_{w_0}(w_t^\star)}{w_t^\star -u}\\
	= & \sum_{t = 0}^{T}\inner{\gamma_t} {w_t-u}+\sum_{t\in [T]}\inner{\grad Q(w_t^\star)}{w_t^\star -u}+ \sum_{t\in[T]}\inner{\eta g(w_0)}{w_t-w_t^\star}\\
	\stackrel{(i)}{\le} & 
	V_{w_0}(u) + \sum_{t\in[T]}\left(V_{w_t}(u)-V_{w_t^\star }(u)\right) + \sum_{t=0}^{T} \left(\langle\gamma_t,w_t-w_{t+1}^\star \rangle - V_{w_{t}}(w_{t+1}^\star )\right) + \sum_{t\in[T]}\inner{\eta g(w_0)}{w_t-w_t^\star}\\
	\stackrel{(ii)}{\le} & V_{w_0}(u) + 2\eta\epsaprx T + \sum_{t=0}^{T} \left( \langle\clip(\eta\tilde{\Delta}_t),w_t-w_{t+1}^\star \rangle-V_{w_t}(w_{t+1}^\star)\right)
	\stackrel{(iii)}{\le}  V_{w_0}(u)+2\eta\epsaprx T+\sum_{t=0}^{T}\eta^2\norm{\tilde{\Delta}_t}_{w_t}^2.
	\end{aligned}
	\end{equation}
	Here, $(i)$ was by rearranging \eqref{eq:idealmirror} via the equality
	\[\sum_{t \in [T]} \inner{\gamma_{t - 1}}{w_t^\star - u} = \sum_{t = 0}^T \inner{\gamma_t}{w_t - u} - \sum_{t = 0}^T \inner{\gamma_t}{w_t - w_{t + 1}^\star},\] 
	$(ii)$ was by the conditions $\max_u \left[V_{w_t}(u)-V_{w_t^\star}(u)\right]\le\eta\epsaprx$ and $\norm{w_t-w_t^\star}\le \tfrac{\epsaprx}{LD}$ satisfied by the iterates, and $(iii)$ was by the property of clipping~\eqref{eq:strong-convexity}, as defined in the problem setup. Now by rearranging and using the three-point property of Bregman divergence (i.e.\ $\langle -\nabla V_{w'}(w),w-u\rangle = V_{w'}(u)-V_{w}(u)-V_{w'}(w)$), it holds that
	\begin{equation}\label{eq:prop-vr-first-2}
	\begin{aligned}
	\sum_{t \in [T]}\inner{\clip(\eta \tilde{\Delta}_t)+\eta g(w_0)}{w_t-u}\le & 
	V_{w_0}(u)+2\eta\epsaprx T+\eta^2\sum_{t=0}^{T }\norm{\tilde{\Delta}_t}_{w_t}^2  +\frac{\alpha\eta}{2}\sum_{t\in[T]}\left( V_{w_0}(u)-V_{w_0}(w_t^\star)\right)\\
	\le & V_{w_0}(u)+3\eta\epsaprx T+\eta^2\sum_{t=0}^T\norm{\tilde{\Delta}_t}_{w_t}^2  +\frac{\alpha\eta}{2}\sum_{t\in[T]}\left( V_{w_0}(u)-V_{w_0}(\hat{w}_t)\right),
	\end{aligned}
	\end{equation}
	where the second inequality follows from the condition $V_{w_0}(\hat{w}_t)-V_{w_0}(w_t^\star)\le\frac{2\epsaprx}{\alpha}$ satisfied by iterates of Algorithm~\ref{alg:innerloop-approx}. To bound the second term of \eqref{eq:boundthreeterms}, we define the ghost iterate sequnce $\{s_t\}$ by
	\begin{equation*}
	s_t = \argmin_{s\in\zset}\left\{\frac{1}{2}\inners{ \eta \Delta_{t-1}-\clip(\eta\tilde{\Delta}_{t-1})}{s} + 
	V_{s_{t - 1}}(s) 
	\right\}
	~~\text{with}~~s_0 = w_0.
	\end{equation*}
	Applying Lemma~\ref{lem:mirror-descent} with $Q=0$ and 
	$\gamma_t=\tfrac{1}{2} (\eta \Delta_t-\clip
	(\eta\tilde{\Delta}_t))$, and observing that again $\gamma_0 = 0$, 
	\begin{align*}
	& \sum_{t\in[T]}  \inners{\eta \Delta_t-\clip
		(\eta\tilde{\Delta}_t)}{s_t-u}\\
	\le & 
	2V_{w_0}(u) + \sum_{t=0}^T\left\{\langle \eta \Delta_t-\clip
	(\eta\tilde{\Delta}_t) , s_t-s_{t+1}\rangle-2V_{s_t}(s_{t+1})\right\}\\
	\le &  2V_{w_0}(u) + \eta^2\sum_{t=0}^T \left(\norm{\Delta_t}_{*}^2 + \norm{\tilde{\Delta}_t}_{s_t}^2\right).
	\end{align*}
Here, we used properties \eqref{eq:strong-convexity} and \eqref{eq:strong-convexity-local}. Consequently,
	\begin{equation}\label{eq:prop-vr-second-1}
	\begin{aligned}
	& \sum_{t\in[T]} \inner{\eta \Delta_t-\clip
		(\eta\tilde{\Delta}_t)}{w_t-u}\\
	= &  \sum_{t\in[T]}  \inner{\eta \Delta_t-\clip
		(\eta\tilde{\Delta}_t)}{w_t-s_t} + \sum_{t\in[T]}\inner{\eta \Delta_t-\clip
		(\eta\tilde{\Delta}_t)}{s_t-u}\\
	\le &\  2 V_{w_0}(u) + \eta^2\sum_{t=0}^T \left(\norm{\Delta_t}_{*}^2 + \norm{\tilde{\Delta}_t}_{s_t}^2\right) + \sum_{t\in[T]}  \inner{\eta \Delta_t-\clip
		(\eta\tilde{\Delta}_t)}{w_t-s_t}.
	\end{aligned}
	\end{equation}
To bound the third term of \eqref{eq:boundthreeterms}, we use the condition $\norm{w_t - \hat{w}_t} \le \tfrac{\epsaprx}{LD}$ which implies
	\begin{equation}\label{eq:prop-vr-third-1}
	\begin{aligned}
	\sum_{t\in[T]}\langle \eta g(w_t)-\eta g(\hat{w}_t),w_t-u\rangle & \le \sum_{t\in[T]}\norm{\eta g(w_t)-\eta g(\hat{w}_t)}_*\norm{w_t-u}\le 2\eta\epsaprx T.
	\end{aligned}
	\end{equation}
Combining our three bounds~\eqref{eq:prop-vr-first-2},~\eqref{eq:prop-vr-second-1}, and~\eqref{eq:prop-vr-third-1} in the context of \eqref{eq:boundthreeterms}, using $\hat{w}_0=w_0$ and $\tilde{g}(w_0)=g(w_0)$, and finally dividing through by $\eta T$, we obtain
	\begin{equation}\label{eq:prop-vr-allthree-1}
	\begin{aligned}
	& \frac{1}{T}\sum_{t\in[T]}\inner{ g(w_t)}{w_t-u} - \left(\frac{3}{\eta T}+\frac{\alpha}{2}\right)V_{w_0}(u) \\
	\le &\  5\epsaprx+\frac{1}{T}\sum_{t\in[T]}\left(\eta\norm{\tilde{\Delta}_t}_{w_t}^2 + \eta\norm{\tilde{\Delta}_t}_{s_t}^2 + \eta\norm{\Delta_t}_{*}^2 + \inner{\Delta_t-\frac{1}{\eta}\clip(\eta\tilde{\Delta}_t)}{w_t-s_t}-\frac{\alpha}{2}V_{w_0}(\hat{w}_t)\right).
	\end{aligned}
	\end{equation}	
Since $T \ge \tfrac{6}{\alpha\eta}$, taking a supremum over $u \in \zset$ in \eqref{eq:prop-vr-allthree-1} and then an expectation yields
\begin{equation}\label{eq:secondlinelezero}
\begin{aligned}
&\E \sup_{u \in \zset} \left[\frac{1}{T}\sum_{t\in[T]}\inner{ g(w_t)}{w_t-u} - \alpha V_{w_0}(u)\right] \le 5\epsaprx \\
+ &\frac{1}{T} \E\left[\sum_{t \in [T]}  \eta\norm{\tilde{\Delta}_t}_{w_t}^2 + \eta\norm{\tilde{\Delta}_t}_{s_t}^2 + \eta\norm{\Delta_t}_{*}^2 + \inner{\Delta_t-\frac{1}{\eta}\clip(\eta\tilde{\Delta}_t)}{w_t-s_t}-\frac{\alpha}{2}V_{w_0}(\hat{w}_t)\right].
\end{aligned}
\end{equation} 
We will show the second line of \eqref{eq:secondlinelezero} is nonpositive. To do so, observe for each $t \in [T]$, by the property~\eqref{eq:clipping-distortion} of $\clip(\cdot)$, since conditional on $w_t$, $s_t$, $\tilde{\Delta}_t$ is unbiased for deterministic $\Delta_t$,
\begin{equation}\label{eq:prop-vr-allthree-var-1}
\begin{aligned}
\left|\E\inner{\Delta_t-\frac{1}{\eta}\clip(\eta\tilde{\Delta}_t)}{w_t-s_t} \right| = \left|\E\inner{\tilde{\Delta}_t-\frac{1}{\eta}\clip(\eta\tilde{\Delta}_t)}{w_t-s_t} \right| \le \eta \norm{\tilde{\Delta}_t}_{w_t}^2+\eta \norm{\tilde{\Delta}_t}_{s_t}^2.
\end{aligned}
\end{equation}
Finally, by using property 2 of the centered-local estimator $\tilde{\Delta}_t$, as well as Remark~\ref{rem:vr-jensen}, we have for each $t \in [T]$,
	\begin{equation}\label{eq:prop-vr-allthree-var-2}
	\begin{aligned}
	\E\left[\eta\norm{\tilde{\Delta}_t}_{w_t}^2\right]\le \eta L^2 V_{w_0}(\hat{w}_t),\;\E\left[\eta\norm{\tilde{\Delta}_t}_{s_t}^2\right]\le \eta L^2 V_{w_0}(\hat{w}_t),\text{ and } \eta\norm{\Delta_t}_{*}^2 \le \eta L^2 V_{w_0}(\hat{w}_t).
	\end{aligned}
	\end{equation}
Using bounds \eqref{eq:prop-vr-allthree-var-1} and \eqref{eq:prop-vr-allthree-var-2} in \eqref{eq:secondlinelezero}, as well as $\eta\le \tfrac{\alpha}{10L^2}$,
	\begin{equation}\label{eq:prop-vr-allthree-2}
\E \sup_{u \in \zset} \left[\frac{1}{T}\sum_{t\in[T]}\inner{ g(w_t)}{w_t-u} - \alpha V_{w_0}(u)\right] \le 5\epsaprx.
	\end{equation}
	For the final claim, denote the true average iterate by $\bar{w}\defeq\frac{1}{T}\sum_{t\in[T]}w_t$. We have $\forall u\in\zset$,
	\begin{align*}
	\inner{g(\tilde{w})}{\tilde{w} - u} & \stackrel{(i)}{=} -\inner{g(\tilde{w})}{u} = \inner{g(\bar{w})-g(\tilde{w})}{u}+\inner{g(\bar{w})}{\bar{w}-u}\\
	& \stackrel{(ii)}{\le} \epsaprx+\inner{g(\bar{w})}{\bar{w}-u}\\
	&= \frac{1}{T}\sum_{t\in[T]}\inner{g(w_t)}{w_t-u}+\epsaprx.
	\end{align*}
	Here, $(i)$ used the fact that linearity of $g$ gives $\inner{g(z)}{z}=0$, $\forall z\in\zset$, and $(ii)$ used H\"older's inequality $\inner{g(\bar{w})-g(\tilde{w})}{u}\le \norm{{g(\bar{w})-g(\tilde{w})}}_*\norm{u}\le 2LD\norm{\tilde{w}-\bar{w}}\le \epsaprx$ following from the approximation guarantee $\norm{\tilde{w}-\bar{w}}\le\tfrac{\epsaprx}{2LD}$. Combining with \eqref{eq:prop-vr-allthree-2} yields the conclusion, as $6\epsaprx = \vepsi$.
	
\end{proof}
\section{Deferred proofs for sublinear methods}
\label{app:sublinear-proofs}

\subsection{$\elltwo$ sublinear coordinate method}
\label{ssec:l2l2sub}
\paragraph{Assumptions.} The algorithm in this section will assume access to entry queries, $\ell_1$ norms of rows and columns, and $\ell_1$ sampling distributions for rows and columns. Further, it assumes the ability to sample a row or column proportional to its squared $\ell_1$ norm; given access to all $\ell_1$ norms, the algorithm may spend $O(m + n)$ constructing these sampling oracles in $O(m + n)$ time, which does not affect its asymptotic runtime. We use the $\ell_2$-$\ell_2$ local norm setup (Table~\ref{tab:local-norm}). We define

\begin{equation}\label{eq:l22def}\lttco \defeq  \sqrt{\sum\limits_{i\in[m]}\lones{\ai}^2+\sum\limits_{j\in[n]}\lones{\aj}^2}.\end{equation}

\subsubsection{Gradient estimator}
For $z\in\ball^n\times\ball^m$, we specify two distinct choices of 
sampling distributions $p(z), q(z)$ which obtain the optimal Lipschitz 
constant. The first one is an oblivious distribution: %
\begin{equation}\label{eq:l2-improved-prob1-def}
p_{ij}(z)\defeq \frac{\lones{\ai}^2}{\sum_{k\in[m]}\lones{\ai[k]}^2}\cdot\frac{|A_{ij}|}{\lones{\ai}}
~~\mbox{and}~~\ 
q_{ij}(z)\defeq \frac{\lones{\aj}^2}{\sum_{k\in[n]}\lones{\aj[k]}^2}\cdot\frac{|A_{ij}|}{\lones{\aj}}.
\end{equation}

The second one is a dynamic distribution:
\begin{equation}\label{eq:l2-improved-prob2-def}
p_{ij}(z)\defeq \frac{{[z\y]_i}^2}{\ltwo{z\y}^2}\cdot\frac{|A_{ij}|}{\lones{\ai}}
~~\mbox{and}~~\ 
q_{ij}(z)\defeq \frac{{[z\x]_j}^2}{\ltwo{z\x}^2}\cdot\frac{|A_{ij}|}{\lones{\aj}}.
\end{equation}
We now state the local properties of each estimator.

\begin{restatable}{lemma}{restateestpropltwo}
\label{lem:est-prop-l2-improved}
In the $\elltwo$ setup, estimator
\eqref{eq:estimate-l1} using the sampling distribution in~\eqref{eq:l2-improved-prob1-def} or~\eqref{eq:l2-improved-prob2-def} is an $\lttco$-local estimator.
\end{restatable}

\begin{proof}
	For convenience, we restate the distributions here: they are respectively 
	\begin{equation*}
	p_{ij}(z)\defeq \frac{\lones{\ai}^2}{\sum_{k\in[m]}\lones{\ai[k]}^2}\cdot\frac{|A_{ij}|}{\lones{\ai}}
	~~\mbox{and}~~\ 
	q_{ij}(z)\defeq \frac{\lones{\aj}^2}{\sum_{k\in[n]}\lones{\aj[k]}^2}\cdot\frac{|A_{ij}|}{\lones{\aj}}
	\end{equation*}
	and
	\begin{equation*}
	p_{ij}(z)\defeq \frac{{[z\y]_i}^2}{\ltwo{z\y}^2}\cdot\frac{|A_{ij}|}{\lones{\ai}}
	~~\mbox{and}~~\ 
	q_{ij}(z)\defeq \frac{{[z\x]_j}^2}{\ltwo{z\x}^2}\cdot\frac{|A_{ij}|}{\lones{\aj}}.
	\end{equation*}
	
	Unbiasedness holds by definition. We first show the variance bound on the $x$ block for distribution~\eqref{eq:l2-improved-prob1-def}: 
	\begin{align*}
	\E\left[\ltwo{\tilde{g}\x(z)}^2\right]
	& =
	\sum_{i\in[m],j\in[n]} p_{ij}(z)\cdot \left(\frac{A_{ij}[z\y]_i}{p_{ij}(z)}\right)^2 = \sum\limits_{i\in[m],j\in[n]}\frac{A_{ij}^2[z\y]_i^2}{p_{ij}(z)}\\
	& = \sum\limits_{i\in[m],j\in[n]}\frac{|A_{ij}|}{\lones{\ai}}[z\y]_i^2\cdot\left(\sum\limits_{i\in[m]}\lones{\ai}^2\right) =
	\sum\limits_{i\in[m]}\lones{\ai}^2.
	\end{align*}
	
	Similarly, we have
	\begin{equation*}
	\E\left[\ltwo{\tilde{g}\y(z)}^2\right] \le \sum\limits_{j\in[n]}\lones{\aj}^2.
	\end{equation*}
	
	Now, we show the variance bound on the $x$ block for distribution~\eqref{eq:l2-improved-prob2-def}:
	\begin{align*}
	\E\left[\ltwo{\tilde{g}\x(z)}^2\right]
	& = 
	\sum_{i\in[m],j\in[n]} p_{ij}(z)\cdot \left(\frac{A_{ij}[z\y]_i}{p_{ij}(z)}\right)^2 = \sum\limits_{i\in[m],j\in[n]}\frac{A_{ij}^2[z\y]_i^2}{p_{ij}(z)}\\
	& = \sum\limits_{i\in[m],j\in[n]}|A_{ij}|\lones{\ai}\ltwo{z\y}^2 \le \sum\limits_{i\in[m]}\lones{\ai}^2,
	\end{align*}
	and a similar bound holds on the $y$ block.
\end{proof}
We remark that using the oblivious distribution 
\eqref{eq:l2-improved-prob1-def} saves a logarithmic factor in the runtime 
compared to the dynamic distribution, so for the implementation of all of our $\ell_2$-$\ell_2$ 
algorithms we will use the oblivious distribution. 

\subsubsection{Implementation details}
\label{ssec:implementsubltwo}
In this section, we discuss the details of how to leverage the $\IM_2$ data structure to implement the iterations of our algorithm. The algorithm we analyze is Algorithm~\ref{alg:sublinear}, using the local estimator defined in~\eqref{eq:estimate-l1}, and the distribution~(\ref{eq:l2-improved-prob1-def}). We choose
\[\eta = \frac{\epsilon}{9\lttcop^2} \text{ and } T = \left\lceil\frac{6\Theta}{\eta\eps}\right\rceil \ge\frac{54\lttcop^2}{\epsilon^2}.\]
Lemma~\ref{lem:est-prop-l2-improved} implies that our estimator satisfies the remaining requirements for Proposition~\ref{prop:sublinear}, giving the duality gap guarantee in $T$ iterations. In order to give a runtime bound, we claim that each iteration can be implemented in constant time, with $O(m + n)$ additional runtime.

\paragraph{Data structure initializations and invariants.}
At the start of the algorithm, we spend $O(m + n)$ time initializing data structures via $\IMS_2\x.\Initialize(\mathbf{0}_n, b)$, $\IMS_2\y.\Initialize(\mathbf{0}_m, c)$, where $\IMS_2\x, \IMS_2\y$ are instantiations of $\IM_2$ data structures. Throughout, we preserve the invariant that the points maintained by $\IMS_2\x, \IMS_2\y$ correspond to the $x$ and $y$ blocks of the current iterate $z_t$ at iteration $t$ of the algorithm. We note that we instantiate data structures which do not support $\Sample()$.

\paragraph{Iterations.} For simplicity, we only discuss the runtime of updating the $x$ block as the $y$ block follows symmetrically. We divide each iteration into the following substeps, each of which we show run in constant time. We refer to the current iterate by $z = (z\x, z\y)$, and the next iterate by $w = (w\x, w\y)$.\\

\noindent
\emph{Sampling.} Because the distribution is oblivious, sampling both $i$ and $j\mid i$ using precomputed data structures takes constant time. \\

\noindent
\emph{Computing the gradient estimator.} To compute $c \defeq A_{ij}[z\y]_i/p_{ij}$, it suffices to compute $A_{ij}$, $[z\y]_i$, and $p_{ij}$. Using an entry oracle for $A$ obtains $A_{ij}$ in constant time, and calling $\IMS_2\y.\Get(i)$ takes constant time. Computing $p_{ij}$ using the precomputed row norms and the values of $A_{ij}, [z\y]_i$ takes constant time.\\

\noindent
\emph{Performing the update.} For the update corresponding to a proximal step, we have
$$w\x\leftarrow \Pi_{\xset}\left(z\x-\eta \tilde{g}\x(z)\right)=\frac{z\x-\eta \tilde{g}\x(z)}{\max\{\ltwo{z\x-\eta \tilde{g}\x(z)},1\}}. $$

We have computed $\tilde{g}\x(z)$, so to perform this update, we call
\begin{align*}
& \IMS_2\x.\AddSparse(j, -\eta c);\\
& \IMS_2\x.\AddDense(-\eta); \\
& \IMS_2\x.\Scale(\max\{\IMS\x.\GetNorm(), 1\}^{-1});\\
& \IMS_2\x.\SumUp().
\end{align*}
By assumption, each operation takes constant time because we do not support $\Sample$ in our instances of $\IMS_2$, giving the desired iteration complexity. It is clear that at the end of performing these operations, the invariant that $\IMS_2\x$ maintains the $x$ block of the iterate is preserved.

\paragraph{Averaging.} After $T$ iterations, we compute the average point $\bar{z}\x$:
$$
[\bar{z}\x]_j\gets\frac{1}{T}\cdot\IMS_2\x.\GetSum(j),\forall j\in[n].
$$ 
By assumption, this takes $O(n)$ time. 

\subsubsection{Algorithm guarantee}
\begin{theorem}
\label{thm:l2l2-sublinear}
In the $\elltwo$ setup, the implementation in Section~\ref{ssec:implementsubltwo} has runtime
\begin{equation*}
O\left(\frac{\lttcop^2}{\epsilon^2} + m + n\right)
\end{equation*}
and outputs a point 
$\bar{z}\in\zset$ such that
\begin{equation*}
\E\gap(\bar{z})
\le \epsilon.
\end{equation*}
\end{theorem}

\begin{proof}
	
The runtime bound follows from the discussion in Section~\ref{ssec:implementsubltwo}. The correctness follows from Proposition~\ref{prop:sublinear}.
\end{proof}
\begin{remark}Using our $\IM_2$ data structure, the $\elltwo$ algorithm of~\citet{BalamuruganB16} runs in time $O\left(\rcs\|A\|_F^2/\eps^2\right)$. Our runtime universally 
	improves upon it since 
	$$\sum\limits_{i\in[m]}\lones{\ai}^2 + \sum\limits_{j\in[n]}\lones{\aj}^2\le 2\rcs\norm{A}_{\mathrm{F}}^2.$$
\end{remark}

\subsection{$\elltwoone$ sublinear coordinate method}
\label{ssec:l2l1sub}

\paragraph{Assumptions} The algorithm in this section will assume access to every oracle listed in Section~\ref{ssec:accessmodel}. However, for a specific matrix $A$, only one of three sampling distributions will be used in the algorithm; we describe the specific oracle requirements of each distribution following their definition. We use the $\ell_2$-$\ell_1$ local norm setup (Table~\ref{tab:local-norm}). Throughout this section, we will assume that the linear term in~\eqref{eq:estimate-l1} is $g(0) =0$ uniformly. 

Finally, in this section we assume access to a weighted variant of $\IM_2$, which takes a 
nonnegative weight vector $w$ as a static parameter. $\WIM_2$ supports two modified 
operations compared to the data structure $\IM_2$: its $\GetNorm()$ operation returns $\sqrt{\sum_j [w]_j 
	[x]_j^2}$, and its $\Sample()$ returns coordinate $j$ with probability 
proportional to $[w]_j [x]_j^2$ (cf. Section~\ref{ssec:interface-norm}). We give the implementation of this 
extension in Appendix~\ref{app:ds-proofs}.

\subsubsection{Gradient estimator}
For $z\in\ball^n\times\Delta^m$ and desired accuracy $\eps>0$, we specify three distinct choices of sampling distributions $p(z), q(z)$. Each of our distributions induces an estimator with different properties.

The first one is 
\begin{equation}\label{eq:l2l1-prob1-def}
p_{ij}(z)\defeq \frac{|A_{ij}|}{\lones{\ai}}\cdot[z\y]_i
~~\mbox{and}~~\ 
q_{ij}(z)\defeq \frac{A_{ij}^2}{\norm{A}_{\mathrm{F}}^2}.
\end{equation}

The second one is
\begin{equation}\label{eq:l2l1-prob2-def}
p_{ij}(z)\defeq \frac{|A_{ij}|}{\lones{\ai}}\cdot[z\y]_i
~~\mbox{and}~~\ 
q_{ij}(z)\defeq \frac{[z\x]_j^2\cdot \1_{\{A_{ij}\neq0\}}}{\sum_{l \in [n]} \cs_l \cdot [z\x]_l^2}.
\end{equation}

Here, we let $\cs_j \le \rcs$ denote the number of nonzero elements in column $A_{:j}$. The third one is 
\begin{equation}\label{eq:l2l1-prob3-def}
p_{ij}(z)\defeq \frac{|A_{ij}|}{\lones{\ai}}\cdot[z\y]_i
~~\mbox{and}~~\ 
q_{ij}(z)\defeq \frac{\left|A_{ij} \right| \cdot 
[z\x]_j^2 }{\sum_{l \in [n]} \norm{A_{:l}}_1 \cdot [z\x]_l^2 }.
\end{equation}

For $\ltooco, \ltotco,$ and $\ltohco$ to be defined, the estimators induced by these distributions are local estimators whose guarantees depend on these constants respectively. Furthermore, these Lipschitz constants are in general incomparable and depend on specific properties of the matrix. Therefore, we may choose our definition of $\ltoco$ to be the minimum of these constants, by choosing an appropriate estimator. We now state the local properties of each estimator.

\begin{restatable}{lemma}{estpropltwolone}
\label{lem:est-prop1-l2l1}
In the $\elltwoone$ setup, estimator \eqref{eq:estimate-l1} using the sampling distributions in~\eqref{eq:l2l1-prob1-def},~\eqref{eq:l2l1-prob2-def}, or~\eqref{eq:l2l1-prob3-def} is respectively a $L_{\mathsf{co}}^{2, 1, (k)}$-local estimator, for $k \in \{1, 2, 3\}$, and
\begin{align*}
\ltooco &\defeq\sqrt{ \max_{i\in[m]}\lones{\ai}^2 + \norm{A}_\mathrm{F}^2}, \\
\ltotco &\defeq\sqrt{2\rcs\max_{i\in[m]}\ltwo{\ai}^2}, \\
\ltohco &\defeq\sqrt{\max_{i\in[m]}\lones{\ai}^2 + \left(\max_{i \in [m]} \norm{\ai}_1\right)\left(\max_{j \in [n]} \norm{\aj}_1\right) }.
\end{align*}
\end{restatable}

\begin{proof}
	First, we give the proof for the sampling distribution~(\ref{eq:l2l1-prob1-def}). Unbiasedness holds by definition. For the $x$ block, we have the variance bound: 
	\begin{align*}
	\E\left[\ltwo{\tilde{g}\x(z)}^2\right]
	&=
	\sum_{i\in[m],j\in[n]} p_{ij}(z)\cdot \left(\frac{A_{ij}[z\y]_i}{p_{ij}(z)}\right)^2 = \sum\limits_{i\in[m],j\in[n]} |A_{ij}|\lones{\ai}[z\y]_i\\
	& = \sum\limits_{i\in[m]}\lones{\ai}^2[z\y]_i \le \max_{i\in[m]}\lones{\ai}^2.
	\end{align*}
	
	For arbitrary $w\y$, we have the variance bound on the $y$ block:
	\begin{align*}
	\E\left[\norm{\tilde{g}\y(z)}_{w\y}^2\right]
	&=
	\sum_{i\in[m],j\in[n]} q_{ij}(z)\cdot\left( [w\y]_i\cdot \left(\frac{A_{ij}[z\x]_j}{q_{ij}(z)}\right)^2 \right) = \sum\limits_{i\in[m],j\in[n]}[w\y]_i \frac{A_{ij}^2[z\x]_j^2}{q_{ij}(z)}\\
	& = \sum\limits_{i\in[m],j\in[n]}[w\y]_i[z\x]_j^2\norm{A}_\mathrm{F}^2
	\le \norm{A}_\mathrm{F}^2.
	\end{align*}	
	Next, we give the proof for the sampling distribution~(\ref{eq:l2l1-prob2-def}). Unbiasedness holds by definition. By Cauchy-Schwarz and our earlier proof, we have the variance bound for the $x$ block:
	$$\E\left[\ltwo{\tilde{g}\x(z)}^2\right] \leq \max_{i\in[m]}\lones{\ai}^2 \leq \rcs \max_{i \in [m]} \ltwo{\ai}^2.$$
	For arbitrary $w\y$, we have the variance bound on the $y$ block, where $S_i \defeq \left\{j \mid \1_{A_{ij} \neq 0} = 1\right\}$:
	\begin{align*}
	\E\left[\norm{\tilde{g}\y(z)}_{w\y}^2\right]
	&=
	\sum_{i\in[m],j\in S_i} q_{ij}(z)\cdot\left( [w\y]_i\cdot \left(\frac{A_{ij}[z\x]_j}{q_{ij}(z)}\right)^2 \right)= \sum\limits_{i\in[m],j\in S_i}[w\y]_i \frac{A_{ij}^2[z\x]_j^2}{q_{ij}(z)}\\
	& \le \sum\limits_{i\in[m],j\in S_i}[w\y]_i A_{ij}^2\rcs \le \rcs\max_{k\in[m]}\ltwo{A_{k:}}^2.
	\end{align*}	
	Finally, we give the proof for the sampling distribution~(\ref{eq:l2l1-prob3-def}). Unbiasedness and the variance bound for the $x$ block again hold. For arbitrary $w\y$, we have the variance bound on the $y$ block:
	\begin{align*}
	\E\left[\norm{\tilde{g}\y(z)}_{w\y}^2\right]
	&=
	\sum_{i\in[m],j\in [n]} q_{ij}(z)\cdot\left( [w\y]_i\cdot \left(\frac{A_{ij}[z\x]_j}{q_{ij}(z)}\right)^2 \right)= \sum\limits_{i\in[m],j\in [n]}[w\y]_i \frac{A_{ij}^2[z\x]_j^2}{q_{ij}(z)}\\
	& \le \left(\sum\limits_{i\in[m],j\in [n]}[w\y]_i |A_{ij}|\right)\left(\sum_{l \in [n]} \norm{A_{:l}}_1 [z\x]_l^2\right)\\
	&\le \left(\max_{k \in [m]} \norm{A_{k:}}_1\right)\left(\max_{l \in [n]} \norm{A_{:l}}_1\right).
	\end{align*}
\end{proof}

By using the definitions of $\ltooco, \ltotco,$ and $\ltohco$, we define the constant
\begin{equation}\label{eq:l21def}\ltoco \defeq \sqrt{\max_{i\in[m]}\lones{\ai}^2 + \min\left(\norm{A}_{\textrm{F}}^2, \rcs \max_{i \in [m]} \norm{\ai}_2^2, \left(\max_{i \in [m]} \norm{\ai}_1\right)\left(\max_{j \in [n]} \norm{\aj}_1\right)\right)}.\end{equation}\

In particular, by choosing whichever of the distributions~(\ref{eq:l2l1-prob1-def}), (\ref{eq:l2l1-prob2-def}), or (\ref{eq:l2l1-prob3-def}) yields the minimial Lipschitz constant, we may always ensure we have a $\ltoco$-local estimator. We now discuss the specific precomputed quantities each estimator requires, among those listed in Section~\ref{ssec:accessmodel}. All distributions require access to entry queries, $\ell_1$ norms of rows, and $\ell_1$ sampling distributions for rows. 

\begin{itemize}
	\item Using the sampling distribution~\eqref{eq:l2l1-prob1-def} requires additional access to $\ell_2$ sampling distributions for rows and columns and the Frobenius norm of $A$.
	\item Using the sampling distribution~\eqref{eq:l2l1-prob2-def} requires additional access to uniform sampling nonzero entries of columns.
	\item Using the sampling distribution~\eqref{eq:l2l1-prob3-def} requires additional access to $\ell_1$ norms of columns and $\ell_1$ sampling distributions for columns.
\end{itemize}

\subsubsection{Implementation details}
\label{ssec:implementsubltwoone}

In this section, we discuss the details of how to leverage the appropriate $\IM_1$ and $\IM_2$ data structures to implement the iterations of our algorithm. The algorithm we analyze is Algorithm~\ref{alg:sublinear}, using the local estimator defined in \eqref{eq:estimate-l1}, and the best choice of distribution among (\ref{eq:l2l1-prob1-def}), (\ref{eq:l2l1-prob2-def}), (\ref{eq:l2l1-prob3-def}). We choose
$$\eta = \frac{\epsilon}{9\ltocop^2} \text{ and } T=\left\lceil\frac{6\Theta}{\eta\eps}\right\rceil\ge\frac{54\ltocop^2 \log(2m)}{\epsilon^2}.$$
Lemma~\ref{lem:est-prop1-l2l1} implies that our estimator satisfies the remaining requirements for Proposition~\ref{prop:sublinear}, giving the duality gap guarantee in $T$ iterations. In order to give a runtime bound, we claim that each iteration can be implemented in $O(\log mn)$ time, with $O(m + n)$ additional runtime. For simplicity, because most of the algorithm implementation details are exactly same as the discussion of Section~\ref{ssec:implementsublone} for the simplex block $y \in \yset$, and exactly the same as the discussion of Section~\ref{ssec:implementsubltwo} for the ball block $x \in \xset$, we discuss the differences here, namely the implementations of sampling and gradient computation. 

We assume that we have initialized $\IMS_1\y$, an instantiation of $\IM_1$, and $\IMS_2\x$, an instantiation of $\IM_2$. When the choice of distribution is \eqref{eq:l2l1-prob2-def}, we also assume access to $\WIMS_2\x$, an instantiation of $\WIM_2$ initialized with the weight vector of nonzero counts  of columns of the matrix; similarly, for distribution \eqref{eq:l2l1-prob3-def} we instantiate a $\WIM_2$ with the weight vector of $\ell_1$ norms of each column. \\

\noindent
\emph{Sampling.} Recall that
$$p_{ij}(z) \defeq \frac{|A_{ij}|}{\lones{\ai}}\cdot[z\y]_i.$$
We first sample coordinate $i$ via $\IMS_1\y.\Sample()$ in $O(\log m)$, and then sample $j$ using the data structure corresponding to $\ai$ in $O(1)$. Next, to sample from the distribution
$$q_{ij}(z)\defeq \frac{A_{ij}^2}{\norm{A}_{\mathrm{F}}^2}$$
required by \eqref{eq:l2l1-prob1-def}, we can sample a coordinate of the matrix proportional to its square in constant time using our matrix access. To sample from the distribution
$$q_{ij}(z)\defeq\frac{[z\x]_j^2\cdot \1_{\{A_{ij}\neq0\}}}{\sum_{l \in [n]} \cs_l \cdot [z\x]_l^2 }$$
required by \eqref{eq:l2l1-prob2-def}, we first sample coordinate $j$ via $\WIMS\x_2.\Sample()$ in $O(\log n)$, and then uniformly sample a coordinate $i$ amongst the entries of $\aj$ for which the indicator labels as nonzero. Finally, to sample from the distribution
$$q_{ij}(z)\defeq\frac{\left|A_{ij} \right| \cdot 
	[z\x]_j^2 }{\sum_{l \in [n]} \norm{A_{:l}}_1 \cdot [z\x]_l^2 }$$
required by \eqref{eq:l2l1-prob3-def}, we sample coordinate $j$ via $\WIMS_2\x.\Sample()$, and then sample a coordinate $i$ proportional to its absolute value using a column sampling oracle. \\

\noindent
\emph{Computing the gradient estimator.} By the proofs of Theorem~\ref{thm:l1l1-sublinear} and Theorem~\ref{thm:l2l2-sublinear}, it suffices to compute $p_{i\x j\x}, q_{i\y j\y}$ in constant time. Calling $\IMS_2\x.\Get(j)$, $\IMS_1\y.\Get(i)$, $\IMS_2\x.\GetNorm()$, and $\WIMS_2\x.\GetNorm()$ when appropriate, and using access to precomputation allows us to obtain all relevant quantities for the computations in $O(1)$.

\subsubsection{Algorithm guarantee}
\begin{theorem}
\label{thm:l2l1-sublinear}
In the $\elltwoone$ setup, the implementation in Section~\ref{ssec:implementsubltwoone} has runtime 
\[O\left(\frac{\ltocop^2\log m\log(mn)}{\epsilon^2} + m + n\right)\]
and outputs a point 
$\bar{z}\in\zset$ such that
\begin{equation*}
\E\gap(\bar{z})
\le \epsilon.
\end{equation*}
\end{theorem}

\begin{proof}
The runtime bound follows from the discussion in Section~\ref{ssec:implementsubltwoone}. The correctness follows from Proposition~\ref{prop:sublinear}.
\end{proof}

\begin{remark}
Using our $\IM_1$ and $\IM_2$ data structures, the $\elltwoone$ algorithm of Clarkson et al. \cite{ClarksonHW10} runs in time $O(\rcs\max_{i \in [m]}\|\ai\|_{2}^2\log^2(mn)/\eps^2)$. By noting the definition of $L_{\mathsf{co}}^{2, 1, (2)}$, our runtime universally improves upon it since $\ltocop^2\le2\rcs\max_{i \in [m]}\|\ai\|_{2}^2$.
\end{remark} %
\section{Deferred proofs for variance-reduced methods}
\label{app:vr-proofs}

\subsection{Helper proofs}\label{ssec:helpervr}

\restatelocalnorms*
\begin{proof}
	Let $\gamma\in\R^m$. 
	Note that for every $\tau \in [0,1]$ (with elementwise multiplication, division and square root),
	$\inner{\gamma}{y-y'}=\inner{\gamma\sqrt{(1-\tau)y + \tau 
			y'}}{\frac{y-y'}{\sqrt{(1-\tau)y + 
				\tau y'}}}$. 
	Therefore, 
	using $2\inner{u}{w} \le  \ltwo{u}^2 + \ltwo{w}^2$, we have for every  
	$\tau \in [0,1]$, 
	\begin{equation*}
	2\inner{\gamma}{y-y'} \le 
	\sum_{i\in[m]} \left((1-\tau)[y]_i + \tau [y']_i\right) [\gamma]_i^2 + 
	\sum_{i\in[m]} \frac{([y]_i-[y']_i)^2}{(1-\tau)[y]_i + \tau [y']_i}.
	\end{equation*}
	Applying the double integral $\int_{0}^1 dt \int_{0}^{t} d\tau$ to both 
	sides of the inequality, and using 
	$\int_{0}^1 dt \int_{0}^{t} 1\cdot d\tau = \frac{1}{2}$ and
	$\int_{0}^1 dt \int_{0}^{t} \tau\cdot d\tau = \frac{1}{6}$ 
	gives
	\begin{equation*}
	\inner{\gamma}{y-y'} \le 
	\sum_{i\in[m]} \left(\frac{1}{3}[y]_i + \frac{1}{6} [y']_i\right) [\gamma]_i^2 + 
	\int_{0}^1 dt \int_{0}^{t}  \sum_{i\in[m]} 
	\frac{([y]_i-[y']_i)^2}{(1-\tau)[y]_i + \tau [y']_i}d\tau.
	\end{equation*}
	Identifying the double integral with the 
	expression
	\begin{equation}\label{eq:kl-div-formula}
	V_{y}(y') = \sum_{i\in[m]} \left(y_i' \log\frac{y_i'}{y_i}  + y_i - y_i'
	\right)= 
	\int_{0}^1 dt \int_{0}^{t}   \sum_{i\in[m]} 
	\frac{(y_i-y'_i)^2}{(1-\tau)y_i + \tau y'_i}d\tau.
	\end{equation}
	for the	divergence induced by entropy, %
	the result follows by choosing $[\gamma]_i = 
	\frac{[y]_i - [y']_i}{\frac{2}{3}[y]_i + \frac{1}{3}[y']_i}$.
\end{proof}

\restatestableentropy*
\begin{proof}
	Letting $\tx$ be the point inducing $x'$ in Definition~\ref{def:betastable}, we have
	\begin{align*}
	\sum_{j \in [n]} x'_j \log x'_j - \sum_{j \in [n]} x_j \log x_j &= \left(\sum_{j \in [n]} x'_j \log x'_j - \sum_{j \in [n]} \tx_j \log \tx_j\right) \\
	&+ \left(\sum_{j\in[n]} \tilde{x}_j\log\tilde{x}_j-\sum_{j\in[n]}x_j\log x_j\right).
	\end{align*}
	
	We bound these two terms separately. For the first term, let $\norm{\tx}_1 = 1 + b$, for some $b \le \beta$; we see that entrywise, $(1 + b)x'_j = \tx_j$. For each $j \in [n]$,
	\begin{align*}x'_j \log x'_j - \tx_j \log \tx_j &= x'_j \log x'_j - (1 + b)x'_j \log\left((1 + b)x'_j\right) \\
	&= bx'_j \log \frac{1}{x'_j} - (1 + b)x'_j \log(1 + b) \\
	&\le bx'_j\log\frac{1}{x'_j} \le \frac{\beta}{e}.\end{align*}
	The first inequality was due to nonnegativity of $(1 + b)\log(1 + b)$ and $x'_j$, and the second was due to the maximum value of the scalar function $z\log \frac{1}{z}$ over the nonnegative reals being $1/e$. Summing over all coordinates yields that the first term is bounded by $\beta n/e$.
	
	For the second term, we have by integration that entrywise
	\begin{align*}
	\tx_j \log \tx_j - x_j \log x_j & = \int_{\alpha=0}^1 (1+\log(x_j+\alpha(\tx_j-x_j)))(\tx_j-x_j)d\alpha \\
	&\le \int_{\alpha=0}^1 (1+\log(\tx_j))(\tx_j-x_j)d\alpha\\
	& \le \int_{\alpha=0}^1 \tx_j(\tx_j-x_j)d\alpha \le(1+\beta)|\tx_j-x_j|.
	\end{align*}
	The first inequality is by $\tx_j\ge x_j$ for all $j\in[n]$ and $\log(x)$ is monotone in $x>0$; the second is by $\log(x)\le x-1$ for all $x>0$; the third again uses $\tx_j\ge x_j$ and that $\tilde{x}_j\le\lones{\tx}\le 1+\beta$, and the second condition in Definition~\ref{def:betastable}. Finally, combining yields the desired
	\[
	\sum_{j \in [n]} x'_j \log x'_j - \sum_{j \in [n]} x_j \log x_j \le \frac{\beta n}{e}+\beta(1+\beta).
	\]
\end{proof}

\subsection{$\elltwo$ variance-reduced coordinate method}
\label{ssec:vr-l2}

\paragraph{Assumptions.} 
As in Section~\ref{ssec:l2l2sub}, the algorithm in this section will assume access to entry queries, $\ell_1$ norms of rows and columns, and $\ell_1$ sampling distributions for rows and columns, and the ability to sample a row or column proportional to its squared $\ell_1$ norm. We use the $\ell_2$-$\ell_2$ local norm setup (cf. Table~\ref{tab:local-norm}). Again, we define
\begin{equation*}\lttco \defeq  \sqrt{\sum\limits_{i\in[m]}\lones{\ai}^2+\sum\limits_{j\in[n]}\lones{\aj}^2}.\end{equation*}

\subsubsection{Gradient estimator}\label{ssec:l2-est}
Given reference point $w_0 \in \ball^n \times \ball^m$, for $z\in\ball^n\times\ball^m$, we specify two distinct sampling distributions $p(z; w_0), q(z; w_0)$ which obtain the optimal Lipschitz constant. The first one is an oblivious distribution:
\begin{equation}\label{eq:vr-l2-prob-def}
p_{ij}(z; w_0)\defeq \frac{\lones{\ai}^2}{\sum_{k\in[m]}\lones{\ai[k]}^2}\cdot\frac{|A_{ij}|}{\lones{\ai}}
~~\mbox{and}~~\ 
q_{ij}(z; w_0)\defeq \frac{\lones{\aj}^2}{\sum_{k\in[n]}\lones{\aj[k]}^2}\cdot\frac{|A_{ij}|}{\lones{\aj}}.
\end{equation}
The second one is a dynamic distribution:
\begin{equation}\label{eq:vr-l2-prob-def-2}
p_{ij}(z; w_0)\defeq \frac{[w_0\y - z\y]_i^2}{\ltwo{w_0\y - z\y}^2}\cdot\frac{|A_{ij}|}{\lones{\ai}}
~~\mbox{and}~~\ 
q_{ij}(z; w_0)\defeq \frac{[w_0\x - z\x]_j^2}{\ltwo{w_0\x - z\x}^2}\cdot\frac{|A_{ij}|}{\lones{\aj}}.
\end{equation}
We now state the local properties of each estimator.

\begin{restatable}{lemma}{restateestpropvrltwo}
\label{lem:est-prop-vr-l2}
In the $\elltwo$ setup, estimator \eqref{eq:estimate-vr} using the sampling distribution in \eqref{eq:vr-l2-prob-def} or~\eqref{eq:vr-l2-prob-def-2} is a $\sqrt{2}\lttco$-centered-local estimator.
\end{restatable}

\begin{proof}
	Unbiasedness holds by definition in both cases. We first show the variance bound on the $x$ block for distribution~\eqref{eq:vr-l2-prob-def}:
	\begin{align*}
	\E\left[\ltwo{\tilde{g}\x_{w_0}(z)-g\x(w_0)}^2\right]
	& =  
	\sum_{i\in[m],j\in[n]} p_{ij}(z;w_0)\cdot \left(\frac{A_{ij}[z\y - w_0\y]_i}{p_{ij}(z;w_0)}\right)^2 = \sum\limits_{i\in[m],j\in[n]}\frac{A_{ij}^2[z\y - w_0\y]_i^2}{p_{ij}(z;w_0)}\\
	& = \sum\limits_{i\in[m],j\in[n]}\frac{|A_{ij}|}{\lones{\ai}}[z\y - w_0\y]_i^2\cdot\left(\sum\limits_{k\in[m]}\lones{\ai[k]}^2\right) \\
	& = \left(\sum\limits_{i\in[m]}\lones{\ai}^2\right)\norm{z\y - w_0\y}_2^2.
	\end{align*}	
	Similarly, we have
	\begin{equation*}
	\E\left[\ltwo{\tilde{g}\y_{w_0}(z)-g\y(w_0)}^2\right] \le \left(\sum\limits_{j\in[n]}\lones{\aj}^2\right)\norm{z\x - w_0\x}_2^2.
	\end{equation*}
	Combining these and using $\norm{z\x - w_0\x}_2^2 + \norm{z\y - w_0\y}_2^2 = 2V_{w_0}(z)$ yields the desired variance bound. Now, we show the variance bound on the $x$ block for distribution~\eqref{eq:vr-l2-prob-def-2}:
	\begin{align*}
	\E\left[\ltwo{\tilde{g}\x_{w_0}(z)-g\x(w_0)}^2\right]
	& =  
	\sum_{i\in[m],j\in[n]} p_{ij}(z;w_0)\cdot \left(\frac{A_{ij}[z\y - w_0\y]_i}{p_{ij}(z;w_0)}\right)^2 = \sum\limits_{i\in[m],j\in[n]}\frac{A_{ij}^2[z\y - w_0\y]_i^2}{p_{ij}(z;w_0)}\\
	& = \sum\limits_{i\in[m],j\in[n]}|A_{ij}|\lones{\ai}\ltwo{z\y - w_0\y}^2\\
	& = \left(\sum\limits_{i\in[m]}\lones{\ai}^2\right)\norm{z\y - w_0\y}_2^2.
	\end{align*}
	and a similar bound holds on the $y$ block.
\end{proof}

Again, for algorithmic considerations (i.e.\ an additional logarithmic factor in the complexity of sampling from \eqref{eq:vr-l2-prob-def-2}), we will only discuss using the oblivious distribution \eqref{eq:vr-l2-prob-def} in our algorithm.

\subsubsection{Implementation details}
\label{ssec:implementvrltwo}
In this section, we discuss the details of how to leverage the $\IM_2$ data structure to implement the iterations of our algorithm. The algorithm we analyze is Algorithm~\ref{alg:outerloop} with $K = \alpha\Theta/\eps$, using Algorithm~\ref{alg:innerloop-approx} as an $(\alpha, 0)$-relaxed proximal oracle. In the implementation of Algorithm~\ref{alg:innerloop-approx}, we use the centered-local gradient estimator defined in \eqref{eq:vr-l2-prob-def}. For each use of Algorithm~\ref{alg:innerloop-approx}, we choose
\begin{align}
\eta =
\frac{\alpha}{20\lttcop^2} \text{ and } T =   
\left\lceil\frac{6}{\eta\alpha}\right\rceil\ge\frac{120\lttcop^2}{\alpha^2}.
\end{align}
Our discussion will follow in three steps: first, we discuss the complexity of all executions in Algorithm~\ref{alg:outerloop} other than the calls to the oracles, as well as the initialization procedure for each inner loop. Next, we discuss the complexity of each iteration of Algorithm~\ref{alg:innerloop-approx}. Finally, we discuss the complexity of computing the average iterate in each run of Algorithm~\ref{alg:innerloop-approx}.
For simplicity, when discussing Algorithm~\ref{alg:innerloop-approx}, we will only discuss implementation of the $x$-block, and the $y$-block will follow symmetrically. Altogether, the guarantees of Proposition~\ref{prop:outerloopproof} and Proposition~\ref{prop:innerloopproof} imply that if the guarantees required by the algorithm hold, the expected gap of the output is bounded by $\eps$.

\paragraph{Outer loop extragradient steps and inner loop data structures.} Overall, we execute $K = \alpha\Theta/\eps$ iterations of Algorithm~\ref{alg:outerloop}, and let $\vepsout=\vepsi=0$ to obtain the desired gap, where $\Theta = 1$ in the $\elltwo$ setup. We spend $O(\nnz)$ time executing each extragradient step in Algorithm~\ref{alg:outerloop} exactly, where the dominant term in the runtime is the computation of each $g(z_{k - 1/2})$, for $k \in [K]$. Also, we can maintain the average point $\bar{z}$ throughout the duration of the algorithm, in $O(m + n)$ time per iteration. At the beginning of each inner loop, we initialize a data structure $\IMS_2\x$ which does not support sampling, an instance of $\IM_2$, with $\IMS_2\x.\Initialize(w_0\x, v)$, for
\[v = (1 - \kappa) w_0\x - \eta\kappa g\x(w_0),\] 
where $\kappa \defeq \frac{1}{1 + \eta\alpha/2}$. The inner loop will preserve the invariant that the point maintained by $\IMS_2\x$ is the $x$ block of the current inner loop iterate $w_t$ in each iteration $t$. To motivate this initialization, we recall the form of the updates,
\begin{equation}\label{eq:updateformltwo}w_{t + 1}\x \gets \Pi_{\xset}\left( \kappa\left(w_t\x + \left(\frac{1}{\kappa} - 1\right)w_0\x - \eta\tilde{g}\x_{w_0}(w_t) \right)\right),\end{equation}
where $\Pi_{\xset}(w) = \frac{w}{\max\{1, \norm{w}_2\}}$, and the fixed dense part of $\tilde{g}\x_{w_0}(w_t)$ is $g\x(w_0)$. Therefore, in the following discussion we will be able to maintain this difference via a scaling by $\kappa$, an appropriate addition of the scaled dense vector, and a sparse update.

Finally, we also store the vector $w_0$ in full, supporting entry queries.

\paragraph{Inner loop iterations.} Each inner loop iteration consists of sampling indices for the computation of $\tilde{g}_{w_0}$, computing the sparse part of $\tilde{g}_{w_0}$, and performing the update to the iterate. We show that we can run each substep in constant time. Then, this implies that the total complexity of the inner loop, other than initializing the data structures and outputting the average iterate, is
\[O(T) = O\left(\frac{\lttcop^2}{\alpha^2}\right).\]
We discuss how to make appropriate modifications to the $x$-block. For simplicity we denote our current iterate as $z$, and the next iterate as $w$. Recall that the distribution is given by 
\[p_{ij}(z; w_0)\defeq \frac{\lones{\ai}^2}{\sum_{k\in[m]}\lones{\ai[k]}^2}\cdot\frac{|A_{ij}|}{\lones{\ai}}.\]

\noindent
\emph{Sampling.} By using precomputed distributions, we can sample $i \propto \lone{\ai}^2$ and then $j \mid i \propto |A_{ij}|$ in constant time.\\

\noindent
\emph{Computing the gradient estimator.} Computing the sparse component of the gradient estimator~\ref{eq:estimate-vr} requires computing $A_{ij}$, $[z\y - w_0\y]_i$, and $p_{ij}(z; w_0)$. Using appropriate use of precomputed access to entries and row norms (it is clear we may pay $O(m + n)$ at the beginning of the algorithm to store the sum $\sum_{k \in [m]} \lone{A_{k:}}^2$), entry $[w_0\y]_i$, and $\IMS_2\y.\Get(i)$ allows us to perform the required computation of the sparse component
$$c \defeq [\tilde{g}\x_{w_0}(z) - g(w_0)]_j$$ 
in constant time, by assumption.\\

\noindent
\emph{Performing the update.} In order to perform the update, we recall the form of the update given by \eqref{eq:updateformltwo}. Thus, it suffices to call
\begin{align*}
& \IMS_2\x.\Scale(\kappa);\\
& \IMS_2\x.\AddDense(1);\\
& \IMS_2\x.\AddSparse(j,-\kappa\eta c);\\
& \IMS_2\x.\Scale(\max\{\IMS_2\x.\GetNorm(), 1\}^{-1});\\
& \IMS_2\x.\SumUp()
\end{align*}
By assumption, each operation takes constant time. By the discussion in the data structure initialization section, it is clear that we preserve the invariant that the point maintained by $\IMS_2\x$ is the $x$ block of the current iterate.

\paragraph{Average iterate computation.}
At the end of each run of Algorithm~\ref{alg:innerloop-approx}, we spend $O(n)$ time computing and returning the average iterate via appropriate calls to $\IMS_2\x.\GetSum(j)$ for each $j \in [n]$, and scaling by $1/T$. This operation is asymptotically dominated by the $O(\nnz(A))$ cost of the extragradient step.

\subsubsection{Algorithm guarantee}

\begin{theorem}
\label{thm:l2l2}
In the $\elltwo$ setup, the implementation in Section~\ref{ssec:implementvrltwo} with the optimal choice of $\alpha = \max\{\epsilon, \lttco \sqrt{1/\nnz}\}$ has runtime
\[O\left(\left(\nnz + \frac{\lttcop^2}{\alpha^2}\right)\frac{\alpha}{\eps}\right) = O\left(\nnz + \frac{\sqrt{\nnz} \lttco}{\epsilon} \right)\]
and outputs a point $\bar{z}\in\zset$ such that
\begin{equation*}
\begin{aligned}
& \E\gap(\bar{z}) \le \epsilon.
\end{aligned}
\end{equation*}
\end{theorem}
\begin{proof}
The correctness of the algorithm is given by the discussion in Section~\ref{ssec:implementvrltwo} and the guarantees of Proposition~\ref{prop:outerloopproof} and Proposition~\ref{prop:innerloopproof}. The runtime bound is given by the discussion in Section~\ref{ssec:implementvrltwo}, and the optimal choice of $\alpha$ is clear.
\end{proof}

To better understand the strengths of our runtime guarantee, Proposition~\ref{prop:reg-improve} shows that Theorem~\ref{thm:l2l2} implies a universal improvement for $\elltwo$ games compared to accelerated gradient descent for matrices $A$ with nonnegative entries (or more generally, for $A$ with $\norm{|A|}_{\text{op}} = O(\norm{A}_{\text{op}})$).

\begin{proposition}\label{prop:reg-improve}
	For any $A\in\R^{m\times n}$, we have 
	\[
	\lttco\defeq \max\left\{\sqrt{\sum_i\lones{\ai}^2}, \sqrt{\sum_j\lones{\aj}^2}\right\}\le \sqrt{m+n}\cdot\opnorm{ |A| }.
	\]
\end{proposition}
\begin{proof}
Denote $\boldsymbol{1}_k$ as the all $1$ vector in $\R^k$. We have the following sequence of inequalities:
\[\sqrt{\sum_{i\in[m]}\lones{\ai}^2}= \ltwo{|A|^\top \boldsymbol{1}_m} = \max_{x\in\ball^n}\boldsymbol{1}_m^\top |A|x \le \ltwo{\boldsymbol{1}_m}\max_{x\in\ball^n}\ltwo{|A|x}\le \sqrt{m}\opnorm{|A|}.\]

Similarly, bounding $\max_{y\in\ball^n}y^\top |A| \boldsymbol{1}_n$ implies $\sqrt{\sum_{j\in[n]}\lones{\aj}^2}\le \sqrt{n}\opnorm{|A|}$. Taking a maximum and using $\max\{\sqrt{m},\sqrt{n}\}\le\sqrt{m+n}$ implies the result.
\end{proof}

\begin{remark}

For matrix $A\in\R^{m\times n}$, combining the guarantees of Theorem~\ref{thm:l2l2} with the bound from Proposition~\ref{prop:reg-improve} implies a runtime bounded by
\[
O\left(\nnz+ \frac{\sqrt{\nnz\cdot (m+n)} \opnorm{|A|}}{\epsilon} \right).
\]
Whenever $\opnorm{A}\ge\opnorm{|A|}$, this is an improvement by a factor of $\sqrt{\nnz/(m+n)}$ compared to the accelerated full-gradient method (c.f. Table~\ref{table:runtimes}), which obtains a runtime of $O(\nnz\cdot\opnorm{A}/\eps)$. This applies without any sparsity or numerical sparsity assumptions, and is the same speedup factor as we obtained for $\ellone$ and $\elltwoone$ games using a variance reduction framework with row and column based gradient estimators in~\citet{CarmonJST19}. The $\elltwo$ variance reduction algorithms of \citet{CarmonJST19} and \citet{BalamuruganB16} do not offer such improvements, and our improvement stems from our coordinate-based gradient estimators and our data structure design.
\end{remark}

\subsection{$\elltwoone$ variance-reduced coordinate method}
\label{ssec:vr-l2l1}

\paragraph{Assumptions.}
The algorithm in this section will assume access to entry queries, $\ell_1$ norms of rows, $\ell_2$ sampling distributions for rows and columns, and the Frobenius norm of $A$. We use the $\ell_2$-$\ell_1$ local norm setup (cf.\ Table~\ref{tab:local-norm}). Again, we define
\begin{align}\ltooco &\defeq\sqrt{ \max_{i\in[m]}\lones{\ai}^2 + \norm{A}_\mathrm{F}^2},\label{eq:def-ltooco1}
\\ \ltotco &\defeq\sqrt{2\rcs\max_{i\in[m]}\ltwo{\ai}^2},\label{eq:def-ltooco2}
\\ \ltohco &\defeq\sqrt{\max_{i\in[m]}\lones{\ai}^2 + \left(\max_{i \in [m]} \norm{\ai}_1\right)\left(\max_{j \in [n]} \norm{\aj}_1\right) }.\label{eq:def-ltooco3}\end{align}
Finally, in this section we assume access to a centered variant of $\WIM_2$, which takes a 
point $x_0$ as a static parameter, where $x_0$ is in the space as the iterates $x$ maintained. $\CIM_2$ supports two additional
operations compared to the data structure $\WIM_2$: $\Sample()$ returns coordinate $j$ with probability 
proportional to $[w]_j[x - x_0]_j^2$ (cf.\ Section~\ref{ssec:interface-norm}) in $O(\log n)$ time, and we may query $\norm{x - x_0}_w^2$ in constant time, where $w$ is a specified weight vector. We give the implementation of this 
extension in Appendix~\ref{app:ds-proofs}.
\subsubsection{Gradient estimator}
\label{ssec:l2l1-est}

Given reference point $w_0 \in \ball^n \times \Delta^m$, for $z\in\ball^n \times \Delta^m$ and a parameter $\alpha>0$, as in Section~\ref{ssec:l2l1sub}, we specify three distinct choices of sampling distributions $p(z; w_0), q(z; w_0)$.

The first one is
\begin{equation}\label{eq:vr-l2l1-prob-def-1}
\begin{aligned}
p_{ij}(z;w_0) \defeq \frac{[z\y]_i+2[w_0\y]_i}{3}\cdot\frac{|A_{ij}|}{\lones{\ai}}
~~\mbox{and}~~\ 
q_{ij}(z;w_0) \defeq \frac{A_{ij}^2}{\norms{A}_\mathrm{F}^2}.
\end{aligned}
\end{equation}

The second one is
\begin{equation}\label{eq:vr-l2l1-prob-def-2}
\begin{aligned}
p_{ij}(z;w_0) \defeq \frac{[z\y]_i+2[w_0\y]_i}{3}\cdot\frac{|A_{ij}|}{\lones{\ai}}
~~\mbox{and}~~\ 
q_{ij}(z)\defeq \frac{[z\x - w_0\x]_j^2 \cdot \1_{\{A_{ij}\neq0\}}}{\sum_{l \in [n]} \cs_l \cdot [z\x - w_0\x]_l^2}.
\end{aligned}
\end{equation}

As in Section~\ref{ssec:l2l1sub}, $\cs_j$ is the number of nonzeros of $A_{:j}$. The third one is
\begin{equation}\label{eq:vr-l2l1-prob-def-3}
\begin{aligned}
p_{ij}(z;w_0) \defeq \frac{[z\y]_i+2[w_0\y]_i}{3}\cdot\frac{|A_{ij}|}{\lones{\ai}}
~~\mbox{and}~~\ 
q_{ij}(z)\defeq \frac{\left|A_{ij} \right| \cdot 
	[z\x- w_0\x]_j^2 }{\sum_{l \in [n]} \norm{A_{:l}}_1 \cdot [z\x - w_0\x]_l^2 }.
\end{aligned}
\end{equation}
We now state the local properties of each estimator.

\begin{restatable}{lemma}{restateestpropvrltwolone}
\label{lem:est-prop-vr-l2l1}
In the $\elltwoone$ setup, estimator \eqref{eq:estimate-vr} using the sampling distributions  in~\eqref{eq:vr-l2l1-prob-def-1},~\eqref{eq:vr-l2l1-prob-def-2}, or~\eqref{eq:vr-l2l1-prob-def-3} is respectively a $\sqrt{2}L_{\mathsf{co}}^{2, 1, (k)}$-centered-local estimator, for $k \in \{1, 2, 3\}$.
\end{restatable}

\begin{proof}
First, we give the proof for the sampling distribution \eqref{eq:vr-l2l1-prob-def-1}. Unbiasedness holds by definition. For the $x$ block, we have the variance bound:
	\begin{align*}
	\E\left[\ltwo{\tilde{g}\x_{w_0}(z)-g\x(w_0)}^2\right]
	& =
	\sum_{i\in[m],j\in[n]} p_{ij}(z;w_0) \left(\frac{A_{ij}[z\y-w_0\y]_i}{p_{ij}(z;w_0)}\right)^2  = \sum\limits_{i\in[m],j\in[n]} \frac{A_{ij}^2[z\y-w_0\y]_i^2}{p_{ij}(z;w_0)}\\
	& \le 2\max_{i\in[m]} \lone{\ai}^2 V_{w_0\y}(z\y),
	\end{align*}
	where in the last inequality we used Lemma~\ref{lem:local-norms}.
	
	For arbitrary $w\y$, we have the variance bound on the $y$ block:
	\begin{align*}
	\E\left[\norm{\tilde{g}\y_{w_0}(z)-g\y(w_0)}_{w\y}^2\right]
	&= \sum\limits_{i\in[m],j\in[n]}[w\y]_i \frac{A_{ij}^2[z\x-w_0\x]_j^2}{q_{ij}(z;w_0)}\\
	& = \sum\limits_{i\in[m],j\in[n]}[w\y]_i[z\x - w\x_0]_j^2\norm{A}_\mathrm{F}^2 \le 2\norm{A}_\mathrm{F}^2V_{w_0\x}(z\x).
	\end{align*}
	Combining these and using $$\norm{\tilde{g}_{w_0}(z)-g(w_0)}^2_w \defeq\ltwo{\tilde{g}_{w_0}(z)\x-g(w_0)\x}^2+\norm{\tilde{g}_{w_0}(z)\y-g(w_0)\y}^2_{w\y}$$ yields the desired variance bound.	For the remaining two distributions, the same argument demonstrates unbiasedness and the variance bound for the $x$ block. For sampling distribution \eqref{eq:vr-l2l1-prob-def-2} and arbitrary $w\y$, we have the variance bound on the $y$ block:
	\begin{align*}
	\E\left[\norm{\tilde{g}\y_{w_0}(z)-g\y(w_0)}_{w\y}^2\right]
	&= \sum\limits_{i\in[m],j\in[n]}[w\y]_i \frac{A_{ij}^2[z\x-w_0\x]_j^2}{q_{ij}(z;w_0)}\\
	&\le \left(\sum\limits_{i\in[m],j\in[n]}[w\y]_i A_{ij}^2\right)\left(\rcs \sum_{j \in [n]} [z\x-w_0\x]_j^2\right)\\
	&\le 2\rcs \max_{i \in [m]} \norm{\ai}_2^2V_{w_0\x}(z\x).
	\end{align*}
	Finally, for sampling distribution \eqref{eq:vr-l2l1-prob-def-3}, we have the variance bound on the $y$ block:
	\begin{align*}
	\E\left[\norm{\tilde{g}\y_{w_0}(z)-g\y(w_0)}_{w\y}^2\right]
	&= \sum\limits_{i\in[m],j\in[n]}[w\y]_i \frac{A_{ij}^2[z\x-w_0\x]_j^2}{q_{ij}(z;w_0)}\\
	&\le \left(\sum\limits_{i\in[m],j\in[n]}[w\y]_i |A_{ij}|\right)\left(\sum_{l \in [n]} \norm{A_{:l}}_1 \cdot [z\x - w_0\x]_l^2\right)  \\
	&\le 2\left(\max_{i \in [m]} \norm{\ai}_1\right)\left(\max_{j \in [n]} \norm{\aj}_1\right)V_{w_0\x}(z\x).
	\end{align*}
\end{proof}

Finally, as in Section~\ref{ssec:l2l1sub}, we define the constant
\[\ltoco \defeq \sqrt{\max_{i\in[m]}\lones{\ai}^2 + \min\left(\norm{A}_{\textrm{F}}^2, \rcs \max_{i \in [m]} \norm{\ai}_2^2, \left(\max_{i \in [m]} \norm{\ai}_1\right)\left(\max_{j \in [n]} \norm{\aj}_1\right)\right)},\]
and note that Lemma~\ref{lem:est-prop-vr-l2l1} implies that we can obtain a $\sqrt{2}\ltoco$-centered-local estimator by appropriately choosing a sampling distribution depending on the minimizing parameter. 

\subsubsection{Implementation details}
\label{ssec:implementvrltwoone}

The algorithm we analyze is Algorithm~\ref{alg:outerloop} with $K = 3\alpha\Theta/\eps$, $\vepsout=2\eps/3$ using Algorithm~\ref{alg:innerloop-approx} as an $(\alpha, \vepsi = \eps/3)$-relaxed proximal oracle with $\epsaprx=\eps/18$. In the implementation of Algorithm~\ref{alg:outerloop}, we again apply the $\truncate(\cdot, \delta)$ operation to each iterate $z_k^{\star}$, where the $\truncate$ operation only affects the $y$ block; choosing $\delta = \frac{\vepsout - \vepsi}{\alpha m}$ suffices for its guarantees (see Section~\ref{ssec:implementvrlone} for the relevant discussion). In the implementation of Algorithm~\ref{alg:innerloop-approx}, we use the centered-local gradient estimator defined in \eqref{eq:estimate-vr}, using the sampling distribution amongst~\eqref{eq:vr-l2l1-prob-def-1},~\eqref{eq:vr-l2l1-prob-def-2}, or~\eqref{eq:vr-l2l1-prob-def-3} which attains the variance bound $\ltoco$. For each use of Algorithm~\ref{alg:innerloop-approx}, we choose 
\begin{equation*}\label{eq:etatchoices}
\begin{aligned}\eta &=
 \frac{\alpha}{20\ltocop^2}
 \text{ and } T = 
\left\lceil\frac{6}{\eta\alpha}\right\rceil=\frac{120\ltocop^2}{\alpha^2}.
\end{aligned}\end{equation*}
For simplicity, because most of the algorithm implementation details are exactly the same as the discussion of Section~\ref{ssec:implementvrlone} for the simplex block $y \in \yset$, and exactly the same as the discussion of Section~\ref{ssec:implementvrltwo} for the ball block $x \in \xset$, we discuss the differences here. 

\paragraph{Outer loop extragradient steps.} We execute $3\alpha\log(2m)/\eps$ iterations of Algorithm~\ref{alg:outerloop} to obtain the desired gap. We spend $O(\nnz)$ time executing each extragradient step exactly, and then $O(m + n)$ time applying the $\truncate$ operation and maintaining the average point $\bar{z}$. When we initialize the inner loop, we also create a data structure supporting sampling from $w_0\y$ in constant time.

\paragraph{Data structure initializations and invariants.}

On the simplex block, we follow the strategy outlined in Section~\ref{ssec:implementvrlone}. We initialize our simplex maintenance data structure $\AEMS\y(w_0\y, v, \kappa, \tveps)$ with parameters
\[\kappa \defeq \frac{1}{1 + \eta\alpha/2},\; v \defeq (1 - \kappa)\log w_0\y - \eta\kappa g\y(w_0),\; \tveps \defeq (m + n)^{-8}.\] 
We will again maintain the invariant that the data structures maintain ``exact'' and ``approximate'' points corresponding to the iterates of our algorithm. The correctness of this setting with respect to the requirements of Proposition~\ref{prop:innerloopproof}, i.e.\ the approximation conditions in Line~\ref{line:inner-error-hat},~\ref{line:inner-error-star} and \ref{line:inner-average} in Algorithm~\ref{alg:innerloop-approx}, follows from the discussion of Section~\ref{ssec:implementvrlone}; we note that the condition $\min_j[w_0\x]_j \ge (m + n)^{-5} = \lambda$ again holds, and that $1 - \kappa \ge (m + n)^{-8}$. Thus, for the parameter $\omega$ used in the interface of $\AEM$, we have \[\log(\omega)=\log\left(\max\left(\frac{1}{1-\kappa},\frac{m}{\lambda \tveps}\right)\right)=O(\log(mn)).\]

On the ball block, we follow the strategy outlined in Section~\ref{ssec:implementvrltwo}, but instead of using an $\IM_2$ on the $x$-block, we use $\CIMS_2\x$, an instance of $\CIM_2$ data structure initialized with the point $w_0\x$, supporting the required sampling operation. For the sampling distribution \eqref{eq:vr-l2l1-prob-def-2}, we use the weight vector of column nonzero counts, and for \eqref{eq:vr-l2l1-prob-def-3} we use the weight vector of column $\ell_1$ norms. Overall, the complexity of the initializations on both blocks is bounded by $O(n+m\log^2(m)\log^2(mn))$. 

\paragraph{Inner loop iterations.}

We discuss how to sample from each of the distributions~\eqref{eq:vr-l2l1-prob-def-1},~\eqref{eq:vr-l2l1-prob-def-2}, and~\eqref{eq:vr-l2l1-prob-def-3} in  $O(\log(m)\log(mn))$. Combining with the discussions of implementing the inner loop in Sections~\ref{ssec:implementvrlone} and~\ref{ssec:implementvrltwo}, the total complexity of the inner loop, other than outputting the average iterate, is
\begin{align*}O\left(T\log^2(m)\log^2(mn) + \nnz + m\log(m)\log^2(mn)\right) \\
= O\left(\frac{\ltoocop^2\log^2(m)\log^2(mn)}{\alpha^2}+ \nnz + m\log(m)\log^2(mn)\right).\end{align*}
As in the variance-reduced $\ellone$ setting, the dominant term in the runtime is the complexity of calling $\AEMS\y.\AddSparse$ in each iteration. Recall that the distribution $p$ in every case is given by
\begin{equation*}
\begin{aligned}
p_{ij}(z;w_0) \defeq \frac{[z\y]_i+2[w_0\y]_i}{3}\cdot\frac{|A_{ij}|}{\lones{\ai}}
\end{aligned}
\end{equation*}
With probability $2/3$ we sample a coordinate $i$ from the precomputed data structure for sampling from $w_0\y$, and otherwise we sample $i$ via $\AEMS\y.\Sample()$. Then, we sample an entry $j$ proportional to its magnitude from the $\ell_1$ sampling oracle for $\ai$ in constant time. The runtime is dominated by $O(\log(m)\log(mn))$.

To sample from the distribution $q$ in \eqref{eq:vr-l2l1-prob-def-1}, we follow the outline in Section~\ref{ssec:vr-l2l1}. Similarly, for sampling from distributions \eqref{eq:vr-l2l1-prob-def-2} and \eqref{eq:vr-l2l1-prob-def-3}, we follow the outline in Section~\ref{ssec:vr-l2l1} but replace all calls to an $\IM$ instance with a call to $\CIMS_2\x$ initialized with an appropriate weight vector. In all cases, the runtime is $O(\log m)$ which does not dominate the iteration complexity. 

Finally, it is clear from discussions in previous sections that the iterate maintenance invariants of our data structures are preserved by the updates used in this implementation.

\subsubsection{Algorithm guarantee}

\begin{theorem}
	\label{thm:l2l1}
	In the $\elltwoone$ setup, let $\tnnz\defeq \nnz + m\log(m)\log^2(mn)$. The implementation in Section~\ref{ssec:implementvrltwoone} with the optimal choice of $\alpha = \max\left(\epsilon/3, \ltoco\log(m)\log\left(mn\right)/\sqrt{\tnnz}\right)$ has runtime
	\begin{align*}
	O\left(\left(\tnnz + \frac{\ltocop^2\log^2(m)\log^2(mn)}{\alpha^2}\right)\frac{\alpha\log(m)}{\eps}\right) 
	= O\left(\tnnz + \frac{\sqrt{\tnnz} \ltoco \log(mn)\log^2(m)}{\epsilon} \right)
	\end{align*}
	and outputs a point $\bar{z} \in \zset$ such that
\begin{equation*}
\begin{aligned}
\E\,\left[\gap(z) \right]\le \epsilon.
\end{aligned}
\end{equation*}
\end{theorem}
\begin{proof}
The correctness of the algorithm is given by the discussion in Section~\ref{ssec:implementvrltwoone} and the guarantees of Proposition~\ref{prop:outerloopproof} with $K=3\alpha\Theta/\eps$, $\vepsout=2\eps/3$, $\vepsi=\eps/3$, Proposition~\ref{prop:innerloopproof} with $\epsaprx=\eps/18$ and data structure $\AEM$ with our choice of \[ \tilde{\veps}\defeq (m + n)^{-8}\] to meet the approximation conditions in Line~\ref{line:inner-error-hat},~\ref{line:inner-error-star} and \ref{line:inner-average} in Algorithm~\ref{alg:innerloop-approx}. The runtime bound is given by the discussion in Section~\ref{ssec:implementvrltwoone}, and the optimal choice of $\alpha$ is clear.
\end{proof}

\section{Additional results on variance-reduced methods}
\label{app:rcs-vr}

\subsection{Row-column sparsity variance-reduced methods}\label{app:rcs}

By instantiating relaxed proximal oracles with row-column based gradient estimators developed in~\citep{CarmonJST19}, implemented with the data structures we develop in Section~\ref{sec:ds}, we obtain the improved complexities as stated in Table~\ref{table:runtimes}. Namely, up to logarithmic factors, we generically replace a dependence on $O(m+n)$ with $O(\rcs)$, where $\rcs$ is defined as the maximum number of nonzero entries for any row or column. In this section, we give implementation details.

The estimators $\tilde{g}_{w_0}$ of \cite{CarmonJST19}, parameterized by reference point $w_0$, sample a full column or row of the matrix (rather than a coordinate). To compute $\tilde{g}_{w_0}(z)$ we sample $i\sim p(z)$ and $j\sim q(z)$ 
independently according to a specified distribution depending on the setup, and use the estimator
\begin{equation}
\label{eq:tgdef-rcs}
\begin{aligned}
\tilde{g}_{w_0}(z)& \defeq \left(A^\top 
w_0\y+\ai\frac{[z\y]_i-[w_0\y]_i}{p_i(w)},
-Aw_0\x-\aj\frac{[z\x]_j-[w_0\x]_j}{q_j(w)}\right),
\end{aligned}
\end{equation}

The key difference between this estimator with that of Section~\ref{ssec:vr-estimator} is that its difference with $g(w_0)$ is $O(\rcs)$-sparse rather than $O(1)$-sparse, requiring $\MultSparse$ steps with $O(\rcs)$-sparse vectors. In all other respects, the implementation details are exactly the same as those in Section~\ref{ssec:vr-l1} and Appendix~\ref{app:vr-proofs}, so we omit them for brevity. We now state our sampling distributions used with the estimator form \eqref{eq:tgdef-rcs}, and the corresponding centered local variance bounds.

In the $\ellone$ setup, we use the sampling distribution (from reference point $w_0 \in \Delta^m \times \Delta^n$)
\begin{equation}\label{eq:l1-prob-def-rcs}
p_i(z)\defeq\frac{[z\y]_i+2[w_0\y]_i}{3}~~\mbox{and}~~\ 
q_j(z)\defeq\frac{[z\x]_j+2[w_0\x]_j}{3}.
\end{equation}
\begin{lemma}
In the $\ellone$ setup, gradient estimator \eqref{eq:tgdef-rcs} using the sampling distribution in \eqref{eq:l1-prob-def-rcs} is a $\sqrt{2}\norm{A}_{\max}$-centered-local estimator.
\end{lemma}
\begin{proof}
Unbiasedness holds by definition. For the variance bound, it suffices to show that 
\[\E \norm{\tilde{g}_{w_0}(z) - g(w_0)}_\infty^2 \le 2\norm{A}_{\max}^2 V_{w_0\x}(z\x);\]
clearly this implies the weaker relative variance bound statement (along with an analogous bound on the $y$ block). To this end, we have
\begin{align*}
\E\norm{\tilde{g}_{w_0}(z) - g(w_0)}_\infty^2 \le \sum_{i \in [m]} \frac{\norm{\ai}_\infty^2[z\y - w_0\y]^2_i}{p_i(z)} \le 2\norm{A}_{\max}^2V_{w_0\x}(z\x),
\end{align*}
where the last inequality used Lemma~\ref{lem:local-norms}. 
\end{proof}
In the $\elltwo$ setup, we use the oblivious sampling distribution
\begin{equation}
\label{eq:l2l2-probs-rcs}
\begin{aligned}
p_i=\frac{\ltwo{\ai}^2}{\lfro{A}^2}
 ~~\mbox{and}~~
 q_j=\frac{\ltwo{\aj}^2}{\lfro{A}^2}.
\end{aligned}
\end{equation}

We proved that gradient estimator \eqref{eq:tgdef-rcs} using the sampling distribution in \eqref{eq:l2l2-probs-rcs} admits a $\lfro{A}$-centered estimator in \cite{CarmonJST19}, which is an equivalent definition to Definition~\ref{def:vr} in the $\elltwo$ setup. In the $\elltwoone$ setup, we use the sampling distribution (from reference point $w_0 \in \ball^n \times \Delta^m$)
\begin{equation}
\label{eq:tgdef-l2-probs-rcs}
\begin{aligned}
 p_i(z)=\frac{[z\y]_i+2[w_0\y]_i}{3}~~\mbox{and}~~
 \ q_j(z)=\frac{([z\x]_j-[w_0\x]_j)^2}{\left\Vert z\x - 
 w_0\x\right\Vert^2_{2}}.
\end{aligned}
\end{equation}

\begin{restatable}{lemma}{restateClipped}
	\label{lem:l2-gradient-est-dynamic}
	In the $\elltwoone$ setup, gradient estimator~\eqref{eq:tgdef-rcs} using the sampling distribution in~\eqref{eq:tgdef-l2-probs-rcs} is a $\sqrt{2}L$-centered-local estimator with $L = 
	\max_{i \in [m]}\ltwo{\ai} = \norm{A}_{2\rightarrow\infty}$.
\end{restatable}
\begin{proof}
Unbiasedness holds by definition. For the variance bound, we first note
\begin{align*}
	\E\left[\norm{\tilde{g}\x_{w_0}(z)-g\x(w_0)}_{2}^2\right] & \le \sum_{i\in[m]}\ltwo{\ai}^2\frac{\left([z\y]_i-[w_0\y]_i\right)^2}{\tfrac{1}{3}[z\y]_i+\tfrac{2}{3}[w_0\y]_i}\le \max_{i\in[m]}\ltwo{\ai}^2\left(\sum_{i\in[m]}\frac{\left([z\y]_i-[w_0\y]_i\right)^2}{\tfrac{1}{3}[z\y]_i+\tfrac{2}{3}[w_0\y]_i}\right)\\
	& \le 2\max_{i\in[m]}\ltwo{\ai}^2V_{w_0\y}(z\y),
\end{align*}
where for the last inequality we use Lemma~\ref{lem:local-norms}.
On the other block, we have 
\[
\max_{i\in[m]}\E \left[\tilde{g}_{w_0}\y(w) - 
		g\y(w_0)\right]_i^2\le \max_{i\in[m]}\sum_{j\in[n]}\frac{A_{ij}^2[w\x-w_0\x]_j^2}{q_j(w)} =2\max_{i\in[m]}\ltwo{\ai}^2V_{w_0\x}(w\x).
\]
Summing these two bounds concludes the proof.
\end{proof}

\subsection{Extensions with composite terms}
\label{ssec:comp}

In this section, we give a brief discussion of how to change Proposition~\ref{prop:innerloopproof} and implementations of the procedures in Sections~\ref{sec:sublinear} and~\ref{sec:vr} to handle modified regularization in the context of Proposition~\ref{prop:outerloopproof-sm}, and composite regularization terms in the objective in the methods of Section~\ref{sec:app}. Specifically we consider a composite optimization problem of the form:
\[
\min_{x\in\xset}\max_{y\in\yset}y^\top Ax+\mu\x\phi(x)-\mu\y\psi(y)\text{ where }\phi= 
	V\x_{x'}
\text{ and }\psi=
	V\y_{y'}.
\]

For simplicity of notation we define $\Upsilon(x,y)\defeq \mu\x\phi(x)+\mu\y\psi(y)$. We remark that $x'=0$ recovers the case of $\phi= r\x$ when $\xset=\ball^n$, and $x'=\tfrac{1}{n}\1$ recovers the case of $\phi=r\x$ when $\xset=\Delta^n$ (similarly setting $y'$ allows us to recover this for the $y$ block).

\subsubsection{Changes to inner loop}

In this section, we first discuss the necessary changes to Algorithm~\ref{alg:innerloop-approx} and Proposition~\ref{prop:innerloopproof}. For simplicity of notation, we denote $\rho\defeq\sqrt{\mu\x/\mu\y}$, $\hat{V}\x \defeq \rho V\x$, $\hat{V}\y\defeq\tfrac{1}{\rho}V\y$, $\hat{V} \defeq \hat{V}\x+\hat{V}\y$.

\begin{algorithm}
	\label{alg:innerloop-approx-comp}
	\DontPrintSemicolon
	\KwInput{Initial $w_0\in\zset$, $(L, 
	\alpha)$-centered-local gradient estimator $\tilde{g}_{w_0}$,
	oracle quality $\alpha>0$}
	\Parameter{Step size $\eta$, number of iterations $T$, approximation tolerance $\epsaprx$}
	\KwOutput{Point $\tilde{w}$ satisfying Definition~\ref{def:alphaprox}}
	\For{$t = 1, \ldots, T$}
	{
	$\hat{w}_{t - 1}\approx w_{t - 1}$ satisfying $\hat{V}_{w_0}(\hat{w}_{t-1})-\hat{V}_{w_0}(w_{t-1})\le\tfrac{\epsaprx}{\alpha}$ and $\norm{\hat{w}_{t - 1}-w_{t - 1}}\le\tfrac{\epsaprx}{LD}$\;\label{line:inner-error-hat-comp}
	$w_t^\star \leftarrow 
				\argmin\left\{\inner{\clip(\eta \tilde{g}_{w_0}(\hat{w}_{t - 
						1}) - \eta g(w_0))}{w} + \eta \Upsilon(w)+ \frac{\eta\alpha}{2}\hat{V}_{w_0}(w) + \hat{V}_{w_{t - 1}}(w) 
		\right\}$\; 
	$w_t\approx w_t^\star$ satisfying $\max_u \left[\hat{V}_{w_t}(u)-\hat{V}_{w_t^\star}(u)\right]\le \tfrac{\epsaprx}{1+\sqrt{\mu\x\mu\y}}$, $\hat{V}_{w_0}(w_t)-\hat{V}_{w_0}(w_t^\star)\le\tfrac{\epsaprx}{\alpha}$, and $\hat{V}_{z'}(w_t)-\hat{V}_{z'}(w_t^\star)\le\tfrac{\epsaprx}{\sqrt{\mu\x\mu\y}}$\;\label{line:inner-error-star-comp}
	}
	\Return 
	$\tilde{w}\approx\frac{1}{T}\sum_{t=1}^T w_t$ satisfying $\norm{\tilde{w} - \frac{1}{T}\sum_{t=1}^T w_t} \le \tfrac{\epsaprx}{LD}$, $\max_u\left[\hat{V}_{\tilde{w}}(u)- \hat{V}_{\bar{w}} (u)\right]\le \tfrac{\epsaprx}{\sqrt{\mu\x\mu\y}}$, $\hat{V}_{w'}(\tilde{w})- \hat{V}_{w'} (\bar{w})\le \tfrac{\epsaprx}{\sqrt{\mu\x\mu\y}}$, and $\norm{w_t-w_t^\star}\le \tfrac{\epsaprx}{2LD}$\label{line:inner-average-comp}
	\caption{$\InnerLoopApprox(w_0, \tilde{g}_{w_0}, \epsaprx)$}
\end{algorithm}

\begin{restatable}{corollary}{restateInnerLoop}
\label{prop:innerloopproof-comp}
	Let ($\zset$, $\norm{\cdot}_{\cdot}$, 
	$r$, $\Theta$, $\clip$) be any local norm setup.
	Let $w_0 \in \zset$, $\vepsi>0$, and 
	$\tilde{g}_{w_0}$ be an $L$-centered-local estimator for some $L \ge \alpha \ge \vepsi$. Assume the problem has bounded domain size $\max_{z\in\zset}\|z\|\le D$, $g$ is $L$-Lipschitz, i.e.\ $\norm{g(z)-g(z')}_* \le L \norm{z-z'}$, that $g$ is $LD$-bounded, i.e. $\max_{z\in\zset}\norm{g(z)}_*\le LD$, and $\hat{w}_0 = w_0$. Then, for $\eta = \frac{\alpha}{10L^2}$, 
	$T \geq \frac{8}{\eta\alpha} \ge\frac{60L^2}{\alpha^2}$, $\epsaprx = \frac{\vepsi}{10}$,
	Algorithm~\ref{alg:innerloop-approx-comp} outputs a point $\hat{w} \in 
	\zset$ such that
	\begin{equation}\label{eq:innerloop-guarantee}
	\Ex{}\max\limits_{u\in\zset}
	\left[\inner{g(\tilde{w})+\nabla \Upsilon (\tilde{w})}{\tilde{w} - u} - \alpha V_{w_0}(u)\right]
	\le \vepsi,
	\end{equation}
	i.e. Algorithm~\ref{alg:innerloop-approx-comp} is an 
	$(\alpha,\vepsi)$-relaxed proximal oracle.
\end{restatable}
\begin{psketch}
	
Note that the only change is in the definition of the regularized mirror descent step with extra composite terms
\[w_t^\star \leftarrow 
				\argmin\left\{\inner{\clip(\eta \tilde{g}_{w_0}(\hat{w}_{t - 
						1}) - \eta g(w_0))}{w} + \eta \Upsilon(w)+ \frac{\eta\alpha}{2}\hat{V}_{w_0}(w) + \hat{V}_{w_{t - 1}}(w) 
				\right\}.\]
Denote $\grad\Upsilon(w)=(\mu\x\nabla \phi(w\x),\mu\y\nabla \psi(w\y))$, so that for the final regret bound there are two additional error terms. The first term comes from the error in regularized mirror descent steps via (denoting $z'=(x',y')$)
\begin{align*}
	& \frac{1}{T}\sum_{t\in[T]}\left[-\langle\grad\Upsilon(w_t^\star),w_t^\star-u\rangle+\langle\grad\Upsilon(w_t),w_t-u\rangle\right]\\
\le & \frac{\sqrt{\mu\x\mu\y}}{T}\sum_{t\in[T]}\left(\hat{V}_{z'}(w_t)-\hat{V}_{z'}(w_t^\star)+\hat{V}_{w_t}(u)-\hat{V}_{w_t^\star}(u)\right)\le 2\epsaprx
\end{align*}
following the approximation guarantee in Line~\ref{line:inner-error-star-comp}. The other term comes from averaging error. Denote the true average iterate by $\bar{w}\defeq\frac{1}{T}\sum_{t\in[T]}w_t$. We have $\forall u\in\zset$,
\begin{align*}
\inner{g(\tilde{w})}{\tilde{w} - u} - \frac{1}{T}\sum_{t\in[T]}\inner{g(w_t)}{w_t-u}& = -\inner{g(\tilde{w})}{u} - \inner{g(\bar{w})}{\bar{w}-u}\\
&= \inner{g(\bar{w})-g(\tilde{w})}{u} \le \epsaprx,
\end{align*}
and also 
\begin{align*}
\inner{\grad \Upsilon (\tilde{w})}{\tilde{w} - u} & = \langle\grad \Upsilon(\tilde{w})-\grad \Upsilon(\bar{w}),\tilde{w}-u\rangle + \langle\nabla \Upsilon(\bar{w}),\tilde{w}-\bar{w}\rangle + \langle\nabla \Upsilon(\bar{w}),\bar{w}-u\rangle\\
& \stackrel{(i)}{=} \sqrt{\mu\x\mu\y}\left(- \hat{V}_{\bar{w}} (u)+\hat{V}_{\tilde{w}}(u)+\hat{V}_{\bar{w}}(\tilde{w})\right) +\langle\nabla \Upsilon(\bar{w}),\tilde{w}-\bar{w}\rangle+\langle\nabla \Upsilon(\bar{w}),\bar{w}-u\rangle\\
& \stackrel{(ii)}{=} \sqrt{\mu\x\mu\y}\left(- \hat{V}_{\bar{w}} (u)+\hat{V}_{\tilde{w}}(u)+\hat{V}_{w'}(\tilde{w}) - \hat{V}_{w'}(\bar{w})\right)+\langle\nabla \Upsilon(\bar{w}),\bar{w}-u\rangle,\\
& \stackrel{(iii)}{\le} 2\epsaprx+\langle\nabla \Upsilon(\bar{w}),\bar{w}-u\rangle\\
& \stackrel{(iv)}{\le} 2\epsaprx+\frac{1}{T}\sum_{t\in[T]}\langle\nabla \Upsilon(w_t),w_t-u\rangle.
\end{align*}
where we use $(i)$ the three-point property of Bregman divergence, $(ii)$ the fact that $\hat{V}_{\bar{w}}(\tilde{w}) +\langle\nabla \Upsilon(\bar{w}),\tilde{w}-\bar{w}\rangle = \hat{V}_{w'}(\tilde{w}) - \hat{V}_{w'}(\bar{w})$ again by the three-point property, $(iii)$ the approximation guarantee of Line~\ref{line:inner-average-comp}, and $(iv)$ the fact that $\langle\nabla \Upsilon(w),w-u\rangle$ is convex in $w$ for our choices of $\Upsilon$. Hence incorporating the above extra error terms into the regret bound yields the conclusion, as $10\epsaprx = \vepsi$ by our choice of $\epsaprx$.
\end{psketch}

\subsubsection{Changes to implementation}\label{ssec:imp-comp}

Broadly speaking, all of these modifications can easily be handled via appropriate changes to the initial data given to our data structures $\CIM_2$ and $\AEM$. We discuss general formulations of iterations with these modifications in both simplices and Euclidean balls, and provide appropriate modifications to the inital data given to our data structures. Finally, it is simple to check that all relevant parameters are still bounded by a polynomial in the dimensions of variables, so no additional cost due to the data structure is incurred. For simplicity here we only considerfor the $x$-block when $\phi\x(x)=\mu r(x)$ and remark that the case when $\phi\x(x)=\mu V_{x'}(x)$ for some $x'$ follows similarly.

\paragraph{$\ell_1$ domains.} For this section, define a domain $\xset = \Delta^n$, let $r(x) = \sum_{j \in [n]} x_j \log x_j$ be entropy, and let $\mu$, $\alpha$, $\eta$, $\rho$ be nonnegative scalar parameters. Consider a sequence of iterates of the form
\[x_{t + 1} \gets \argmin_{x \in \xset}\inner{\tilde{g}_{x_0}(x_t)}{x} + \mu r(x) + \frac{\alpha \rho}{2} V_{x_0}(x) + \frac{\rho}{\eta} V_{x_t}(x).\]
This update sequence, for the form of gradient estimator
\[\tilde{g}_{x_0}(x) = g(x_0) + b + g'(x),\]
where $g'(x)$ is a vector with suitable sparsity assumptions depending on the point $x$, and $b$ is some fixed vector, generalizes all of the settings described above used in our various relaxed proximal oracle implementations. Optimality conditions imply that the update may be rewritten as
\[x_{t + 1} \gets \Pi_{\Delta}\left(\exp\left(\frac{\frac{\rho}{\eta}\log x_t + 
\frac{\alpha\rho}{2}\log x_0 - g(x_0) - b - g'(x_t)}{\mu + 
\frac{\alpha\rho}{2} + \frac{\rho}{\eta}}\right)\right).\]
Thus, initializing an $\AEM$ instance with
\[\kappa = \frac{1}{\frac{\mu\eta}{\rho} + \frac{\alpha\eta}{2} + 1},\; v = \frac{\frac{\alpha\rho}{2}\log x_0 - g(x_0) - b}{\mu + \frac{\alpha\rho}{2} + \frac{\rho}{\eta}}\]
enables $\DenseStep$ to propagate the necessary changes to the iterate; we propagate changes due to $g'(x_t)$ via $\AddSparse$ and the appropriate sparsity assumptions.

\paragraph{$\ell_2$ domains.} For this section, define a domain $\xset = \ball^n$, let $r(x) = \half\norm{x}_2^2$ be entropy, and let $\mu$, $\alpha$, $\eta$, $\rho$ be nonnegative scalar parameters. Consider a sequence of iterates of the form
\[x_{t + 1} \gets \argmin_{x \in \xset}\inner{\tilde{g}_{x_0}(x_t)}{x} + \mu r(x) + \frac{\alpha \rho}{2} V_{x_0}(x) + \frac{\rho}{\eta} V_{x_t}(x).\]
This update sequence, for the form of gradient estimator
\[\tilde{g}_{x_0}(x) = g(x_0) + b + g'(x),\]
where $g'(x)$ is a vector with suitable sparsity assumptions depending on the point $x$, and $b$ is some fixed vector, generalizes all of the settings described above used in our various relaxed proximal oracle implementations. Optimality conditions imply that the update may be rewritten as
\[x_{t + 1} \gets \Pi_{\ball^n}\left(\frac{\frac{\rho}{\eta}x_t + \frac{\alpha\rho}{2}x_0 - g(x_0) - b - g'(x_t)}{\mu + \frac{\alpha\rho}{2} + \frac{\rho}{\eta}}\right).\]
Thus, initializing an $\CIM$ instance with
\[v = \frac{\frac{\alpha\rho}{2}x_0 - g(x_0) - b}{\mu + \frac{\alpha\rho}{2} + \frac{\rho}{\eta}}\]
enables $\AddDense$, $\Scale$, and $\GetNorm$ to propagate the necessary changes to the iterate; we propagate changes due to $g'(x_t)$ via $\AddSparse$ and the appropriate sparsity assumptions.

\section{Deferred proofs from Section~\ref{sec:app}}
\label{sec:deferred_application_proofs}

\subsection{Proofs from Section~\ref{app:maxIB}}\label{app:maxIBproofs}

\begin{proof}[Proof of Lemma~\ref{lem:maxIB-preprocess}]
We consider the following $(\mu, \mu)$-strongly monotone problem, for various levels of $\mu$:
\[\max_{x\in\ball^n}\min_{y\in\Delta_m} f_\mu(x,y)\defeq y^\top \tilde{A}x+y^\top b +\mu\sum_{i\in[m]}[y]_i\log[y]_i-\frac{\mu}{2}\ltwo{x}^2.\]
We claim we can implement an $(\alpha, \veps)$-relaxed proximal oracle for this problem in time
\[\Otil{\frac{\ltocop^2}{\alpha^2}}.\]
The oracle is a composite implementation of Algorithm~\ref{alg:innerloop-approx} as in Algorithm~\ref{alg:innerloop-approx-comp}, using the estimator of Appendix~\ref{ssec:vr-l2l1}.
By an application of Proposition~\ref{prop:outerloopproof-sm}, the overall complexity of solving this problem is (by choosing the optimal $\alpha$, and overloading the constant $\ltoco$ to be with respect to $\tilde{A}$):
\[\Otil{\left(\nnz + \frac{\ltocop^2}{\alpha^2}\right)\frac{\alpha}{\mu}} = \Otil{\nnz + \frac{\sqrt{\nnz} \cdot \ltoco}{\mu}}.\]
By conducting a line search over the parameter $\mu$ via repeatedly halving, the total cost of solving each of these problems is dominated by the last setting, wherein $\mu = \Theta(r^*/ \log m)$, and $R/\mu = \Otil{\rho}$; here, we recall that we rescaled $\tilde{A}$ so that $\ltoco = O(R)$. We defer details of the line search procedure to Lemma C.3 of \citet{AllenLO16}.
\end{proof}

\begin{proof}[Proof of Theorem~\ref{thm:MaxIB}]
	We solve the problem \eqref{eq:MaxIB} to duality gap $\eps \hat{r}/8 \le \eps r^*$, using the algorithm of Appendix~\ref{ssec:vr-l2l1} for $\elltwoone$ games. The complexity of this algorithm is (choosing $\alpha$ optimally)
	\[\Otil{\left(\nnz(\tilde{A}) + \frac{\ltocop^2}{\alpha^2}\right) \cdot \frac{\alpha}{\eps\hat{r}}} = \Otil{\nnz + \frac{\rho\sqrt{\nnz } \cdot \ltoco}{\eps}},\]
	as claimed. Here, we used that $\tilde{A}$ is a rescaling of $A$ by $2R$, and $\hat{r}$ is a constant multiplicative approximation of $r$. The approximate solution $(x^*_{\eps'},y^*_{\eps'})$  obtains the requisite duality gap in expectation; Markov's inequality implies that with logarithmic overhead in the runtime, we can obtain a pair of points satisfying with high probability
		\[\max_x f(x,y^*_{\eps'})-\min_y f(x^*_{\eps'},y)=\max_x f(x,y^*_{\eps'})-f(x^*,y^*)+f(x^*,y^*)-\min_y f(x^*_{\eps'},y)\le\eps'.\]
	
Because $y^*$ is the best response to $x^*$, we have $f(x^*, y^*_{\eps'}) \ge f(x^*, y^*)$, which implies
\[\max_x f(x,y^*_{\eps'})-f(x^*,y^*)=\max_x f(x,y^*_{\eps'})-f(x^*,y^*_{\eps'})+f(x^*,y^*_{\eps'})-f(x^*,y^*)\ge 0.\] 
Combining yields $f(x^*,y^*)-\min_y f(x^*_{\eps'},y)\le\eps'\le \eps r^*$, so since $f(x^*, y^*) = r^*$, rearranging implies $\min_y f(x^*_{\eps'},y)\ge r^*-\eps'\ge (1-\eps)r^*$. Thus, $x^*_{\eps'}$ is an $\eps$-approximate solution for Max-IB. 
\end{proof}

\subsection{Proofs from Section~\ref{app:minEB}}
\label{sec:proofs_from_62}

\begin{proof}[Proof of Lemma~\ref{lem:minEB-reg}]
If $(x', y')$ is an approximately optimal solution with duality gap $\eps/16$ for (\ref{eq:minEB-adapted}), by definition
\[	\max_{y \in \Delta^m} f_{\eps'}(x', y)-\min_{x \in \R^n} f_{\eps'}(x, y')\le \frac{\eps}{16}.\]

Therefore, the following sequence of inequalities hold:
\begin{align*}
& \max_{y \in \Delta^m} f(x',y)-\min_{x \in \R^n} f(x,y')
=  \left(\max_{y \in \Delta^m} f(x',y)-\max_{y \in \Delta^m} f_{\eps'}(x',y)\right)\\
+&\left(\max_{y \in \Delta^m} f_{\eps'}(x', y)-\min_{x \in \R^n} f_{\eps'}(x, y')\right)+\left(\min_{x \in \R^n} f_{\eps'}(x, y')-\min_{x \in \R^n} f(x, y')\right)\\
\stackrel{(i)}{\le}  &  \left(\max_{y \in \Delta^m} f(x',y)-\max_{y \in \Delta^m} f_{\eps'}(x', y)\right) + \frac{\eps}{16} + \left(\min_{x \in \R^n} f_{\eps'}(x, y')-\min_{x \in \R^n} f(x, y')\right)\\
\stackrel{(ii)}{\le} & \frac{\eps}{32} + \frac{\eps}{16} + \frac{\eps}{32} = \frac{\eps}{8}.
\end{align*}
In $(i)$, we used the fact that the pair $(x', y')$ has good duality gap with respect to $f_{\eps'}$, and in $(ii)$ we used that for the first summand, $f_{\eps'}(x', \cdot)$ approximates $f(x', \cdot)$ to an additive $\eps/32$, and for the third summand, $-\eps' \sum_{i \in [m]} [y']_i \log [y']_i$ is bounded by $\eps/32$, and all other terms cancel.
\end{proof}

\subsection{Proofs from Section~\ref{app:reg}}
\label{ssec:app-reg}

\begin{proof}[Proof of Lemma~\ref{lem:reg-opt-rel}]
	At optimality for~\eqref{def:minimax-reg}, it holds that
\[
	\begin{cases}
	y_{x'}^*=\frac{1}{\beta}(Ax^*_{x'}-b)\\
	x_{x'}^*=x'-\frac{1}{\beta}A^\top y_{x'}^*
	\end{cases}.
\]	
By substituting $y^*_{x'}$ and rearranging terms we get 
\[
\left(I+\frac{1}{\beta^2}A^\top A\right)(x^*_{x'}-x^*)=x'-x^*,
\]
which in turn gives
\[
\ltwo{x^*_{x'}-x^*} = \ltwo{\left(I+\frac{1}{\beta^2}A^\top A\right)^{-1}(x'-x^*)}\le \frac{1}{1+\mu/\beta^2}\ltwo{x'-x^*}.
\]
For the last inequality we use the fact that
\[\norm{I + \frac{1}{\beta^2}A^\top A}_2^{-1} = \lambda_{\min}\left(I + \frac{1}{\beta^2}A^\top A\right)^{-1} = \frac{1}{1 + \mu/\beta^2}, \] 
by the definition of $\mu$ and since $I$ and $A^\top A$ commute.
\end{proof}

\begin{algorithm}[htbp] 
	\DontPrintSemicolon
	\KwInput{Matrix $A\in\R^{m\times n}$ with $i$th row $\ai$ and $j$th 
		column $\aj$, vector $b\in\R^m$, accuracy $\epsilon$}
	\KwOutput{A point $\tilde{x}$ with $\|\tilde{x}-x^*\|_2\le\epsilon$}
	
	\vspace{3pt}
	$L\gets\max\left\{\sqrt{\sum_i\lones{\ai}^2}, \sqrt{\sum_j\lones{\aj}^2}\right\}$,~$\alpha \gets 
	L/\sqrt{\nnz}$,~$\beta=\sqrt{\mu}$,~$\eta \gets 
	\frac{\alpha}{4 L^2}$\;
	$T\gets\ceil{\frac{4}{\eta\alpha}}$,~$K\gets 
	\Otil{\alpha/\beta}$,~$H=\Otil{1}$,~$z^{(0)}=(x^{(0)},y^{(0)})
	\leftarrow (\boldsymbol{0}_n,\boldsymbol{0}_m )$, 
	$(z_0\x, z_0\y) \gets (\boldsymbol{0}_n, 
	\boldsymbol{0}_m)$\;
	\vspace{3pt}
	\For{$h=1,2,\cdots,H$}
	{
		\For{$k=1,\ldots,K$}
		{
			\vspace{3pt}%
			\Comment{\emph{Relaxed oracle query:}}
			
			$\displaystyle (x_0, y_0) \gets (z_{k-1}\x,z_{k-1}\y)$, $(g\x_0, g\y_0) 
			\gets (A^\top y_0+\beta(x_0-x^{(h-1)}), -A x_0+\beta y_0)$\;
			
			\vspace{3pt}
			
			\For{$t=1,\ldots,T$}
			{
				\vspace{3pt}%
				\Comment{\emph{Gradient estimation:}}
				\vspace{-4pt}%
				
				Sample $i\sim p$ where $\displaystyle p_i = \frac{\left( [y_{t-1}]_i 
					- [y_{0}]_i \right)^2}{\ltwo{y_{t-1}-y_0}^2}$\;
				Sample $j\sim q$ where 
				$\displaystyle q_j = \frac{ \left( [x_{t-1}]_j - [x_{0}]_j 
					\right)^2}{\ltwo{x_{t-1}-x_0}^2}$\;
				
				Set $\displaystyle\tilde{g}_{t-1} = g_0 +  \left(
				\ai \frac{[y_{t-1}]_i - [y_{0}]_i}{p_i}, 
				-\aj \frac{[x_{t-1}]_j - [x_{0}]_j}{q_j}\right)$
				\;

				\Comment{\emph{Mirror descent step:}}
				
				$\displaystyle x_t \gets  \frac{1}{1+\eta\alpha/2} 
				\left(  x_{t 
					- 1} + \frac{\eta\alpha}{2} x_0 - \eta \tilde{g}_{t-1}\x 
				\right)$

				$\displaystyle y_t \gets \Pi_{\yset}\left( \frac{1}{1+\eta\alpha/2} 
				\left(  y_{t 
					- 1} + \frac{\eta\alpha}{2} y_0 - \eta \tilde{g}_{t-1}\y 
				\right)\right)$
				\Comment*[f]{$\Pi_{\yset}(v)=\frac{v}{\max\{1,\ltwo{v}\}}$}
				\; 
			}
			$\displaystyle z_{k - 1/2} \gets \frac{1}{T}\sum_{t=1}^T (x_t, y_t)$\;
			
			\Comment{\emph{Extragradient step:}}
			
			$\displaystyle z_k\x \gets  \frac{\alpha}{\alpha+2\beta} z_{k-1}\x + \frac{2\beta}{\alpha+2\beta} z_{k-1/2}\x  
			-\frac{1}{\alpha+2\beta} \left(A^\top z_{k-1/2}\y +\beta(z\x_{k-1/2}-x^{(h-1)})\right)$\;
			$\displaystyle z_k\y \gets \Pi_{\yset}\left( \frac{\alpha}{\alpha+2\beta} z_{k-1}\y + \frac{2\beta}{\alpha+2\beta} z_{k-1/2}\y  
			+\frac{1}{\alpha+2\beta} \left(A z_{k-1/2}\x -\beta z_{k-1/2}\y \right) \right)$ \;
			
		}
		\Comment{\emph{Reshifting the oracle:}}
		$z^{(h)}=(x^{(h)},y^{(h)})\leftarrow z_K=(z_K\x,z_K\y)$\;
	}
	\Return $\tilde{x}\leftarrow x^{(H)}$
	\caption{Coordinate variance reduced method for linear regression}
	\label{alg:reg}
\end{algorithm}

\restateregthm*

\begin{proof}
	
	We first prove correctness. We bound the progress from $x^{(h)}$ to $x^{(h+1)}$, for some $h \in [H]$, by
	\begin{equation}\label{reg:sum}
	\frac{1}{2}\ltwo{x^{(h+1)}-x^*}^2 \le \ltwo{x^{(h+1)}-x_{x^{(h)}}^*}^2+\ltwo{x_{x^{(h)}}^*-x^*}^2	\le 2V_{z^{(h+1)}}(z_{x^{(h)}}^*)+\ltwo{x_{x^{(h)}}^*-x^*}^2.
	\end{equation}
	The first inequality used $\norm{a + b}_2^2 \le 2\norm{a}_2^2 + 2\norm{b}_2^2$, and the second used the definition of the divergence in the $\elltwo$ setup. Next, choosing a sufficiently large value of $K= \Otil{\beta/\mu}$, we use Proposition~\ref{prop:outerloopproof-sm} to obtain a point $z^{(h + 1)}$ satisfying
	\begin{equation}\label{reg:term-1}
	V_{z^{(h+1)}}(z^*_{x^{(h)}})\le\frac{\eps^2}{80}V_{z^{(h)}}(z^*_{x^{(h)}})\le\frac{\eps^2}{40}V_{z^{(h)}}(z^*)+\frac{\eps^2}{40}V_{z^*}(z^*_{x^{(h)}}).
	\end{equation}
	Further, using Lemma~\ref{lem:reg-opt-rel} with $x'=x^{(h)},\beta=\sqrt{\mu}$ yields
	\begin{equation}\label{reg:term-2}
	\ltwo{x^*_{x^{(h)}}-x^*}\le\frac{1}{2}\ltwo{x^{(h)}-x^*}.
	\end{equation}
	Plugging these two bounds into~\eqref{reg:sum}, and using the form of the divergence in the $\elltwo$ setup,
	\begin{equation}
	\begin{aligned}
	\frac{1}{2}\ltwo{x^{(h+1)}-x^*}^2 & \stackrel{\eqref{reg:term-1}}{\le}\frac{\eps^2}{20}V_{z^{(h)}}(z^*) + \frac{\eps^2}{20}V_{z^*}(z^*_{x^{(h)}}) + \ltwo{x^*_{x^{(h)}}-x^*}^2 \\
	& \stackrel{\eqref{reg:term-2}}{\le}\frac{1}{2}\left(\frac{\eps^2}{20} + \frac{\eps^2}{20}+\frac{1}{2}\right)\ltwo{x^{(h)}-x^*}^2 + \frac{\eps^2}{40}\left(\ltwo{y^{(h)}-y^*}^2+\ltwo{y^*_{x^{(h)}}-y^*}^2\right) \\
	&\le \frac{3}{4}\cdot\frac{1}{2}\ltwo{x^{(h)}-x^*}^2+\frac{\eps^2}{5}.
	\end{aligned}
	\end{equation}
	In the last inequality we use the conditions that $\eps\in(0,1)$ and $\yset=\ball^m$. 	
	Recursively applying this bound for $h\in [H]$, and for a sufficiently large value of $H = \Otil{1}$, we have the desired
	\[
	\ltwo{x^{(H)}-x^*}^2\le \left(\frac{3}{4}\right)^{H}\ltwo{x^{(0)}-x^*}^2+\frac{4\eps^2}{5} \le\eps^2.
	\]
	
	To bound the runtime, recall the inner loop runs for $T=O((\lttco)^2/\alpha^2)$ iterations, each costing constant time, and the outer loop runs for $K=\Otil{\alpha/\beta}$ iterations, each  costing $O(T+\nnz)$. Finally, since $H=\Otil{1}$, the overall complexity of the algorithm is 
	\[\Otilb{\left(\nnz + \frac{\lttcop^2 }{\alpha^2}\right)\frac{\alpha}{\beta}}. \]
	Choosing $\alpha=\max\{
	\lttco/\sqrt{\nnz},\beta\}$ optimally and substituting \[\beta=\sqrt{\mu},\; \lttco=\max\left\{\sqrt{\sum_i\lones{\ai}^2}, \sqrt{\sum_j\lones{\aj}^2}\right\},\]
	we have the desired runtime bound on Algorithm~\ref{alg:reg}.
\end{proof}

\section{$\IM_2$: numerical stability and variations}
\label{app:ds-proofs}

\subsection{Numerical stability of $\IM_1$.} %
We discuss the implementation of a numerically stable version of $\IM_1$, and the complexity of its operations, for use in our sublinear algorithms in Section~\ref{ssec:l1l1sub} and Section~\ref{ssec:l2l1sub}. We discuss this implementation for a simplex block, e.g.\ a simplex variable of dimension $n$, as for an $\ell_2$ geometry numerical stability is clear. The main modifications we make are as follow.
\begin{itemize}
	\item We reinitialize the data structure whenever the field $\nu$ grows larger than some fixed polynomial in $n$, or if $n/2$ iterations have passed.
	\item We track the coordinates modified between restarts.
	\item Every time we reinitialize, we maintain the invariant that the multiplicative range of coordinates of $x$ is bounded by a polynomial in $n$, i.e. $\max_j x_j / \min_j x_j$ is bounded by some fixed polynomial in $n$. We will implement this via an explicit truncation, and argue that such an operation gives negligible additive error compared to the accuracy of the algorithm.
	\item We implicitly track the set of truncated coordinates at each data 
	structure restart. We do so by explicitly tracking the set of 
	non-truncated coordinates whenever a truncation operation happens (see 
	the discussion below), in constant amortized time.
\end{itemize}
We now discuss the complexity and implementation of these restarts. First, note that $\nu$ can never decrease by more than a multiplicative polynomial in $n$ between restarts, because of nonnegativity of the exponential, the fact that the original range at the time of the last restart is multiplicatively bounded, and we restart every time half the coordinates have been touched. Thus, the only source of numerical instability comes from when $\nu$ grows by more than a multiplicative polynomial in $n$. Suppose this happens in $\tau$ iterations after the restart. Then,
\begin{itemize}
	\item If $\tau < n/2$, we claim we can implement the restart in 
	$O(\tau)$, so the amortized cost per iteration is $O(1)$. To see this, for 
	every coordinate touched in these $\tau$ iterations, we either keep or 
	explicitly truncate if the coordinate is too small. For every coordinate not 
	touched in these $\tau$ iterations, the relative contribution is at most 
	inverse polynomial in $n$; we truncate all such coordinates. Then, we 
	compute the normalization constant according to all non-truncated 
	coordinates, such that the value of all truncated coordinates is set to a 
	fixed inverse polynomial in $n$. We can implement this by implicitly 
	keeping track of the set of truncated coordinates as well as their 
	contribution to the normalization factor, and explicitly setting their value 
	in the data structure when they are updated by $\AddSparse$. Overall, 
	this does not affect the value of the problem by more than a small 
	multiple of $\eps$, by our assumptions on $\Lrc/\eps$. To see that we 
	can track the non-truncated coordinates explicitly, we note that it is a 
	subset of the at most $\tau$ coordinates that were touched, so this can 
	be done in constant amortized time.
	\item If $\tau = n/2$, we claim we can implement the restart in $O(n)$, 
	so the amortized cost per iteration is $O(1)$. This is clear: we can do so 
	by explicitly recomputing all coordinates, and truncating any coordinates 
	which have become too small.
\end{itemize}
We describe how the data structure implements this through its maintained 
fields: for non-truncated coordinates, we do not do anything other than 
change the scaling factor $\nu$, and for truncated coordinates, we reset 
the values of $u, u'$ in that coordinate appropriately once they have been 
sparsely updated. Overall, this does not affect the amortized runtime of our 
algorithm.

\subsection{$\WIM_2$}
In this section, we give implementation details for a weighted generalization of $\IM_2$, which we will call $\WIM_2$. It is used in Section~\ref{ssec:l2l1sub}, when using the sampling distribution~\eqref{eq:l2l1-prob3-def}. At initialization, $\WIM_2$ is passed an additional parameter $w \in \R_{\geq 0}^n$, a nonnegative weight vector. We let
$$\inner{u}{v}_w \defeq \sum_{j \in [n]} [w]_j [u]_j [v]_j, \norm{v}_w \defeq \sqrt{\inner{v}{v}_w}.$$
$\WIM_2$ supports all the same operations as $\IM_2$, with two differences:

\begin{itemize}
	\item For the current iterate $x$, $\WIM_2.\Norm()$ returns weighted norm $\norm{x}_w$.
	\item For the current iterate $x$, $\WIM_2.\Sample()$ returns a coordinate $j$ with probability proportional to $[w]_j [x]_j^2$.
\end{itemize}

Similarly to $\IM_2$, $\WIM_2$ maintains the following fields.
\begin{itemize}
	\item Scalars $\xi_u$, $\xi_v$, $\sigma_u$, $\sigma_v$, $\iota$, $\nu$
	\item Vectors $u, u', v, w$ 
	\item Precomputed value $\norm{v}_w^2$.
\end{itemize}
We maintain the following invariants on the data structure fields at the end of every operation:
\begin{itemize}
	\item $x = \xi_u u + \xi_v v$, the internal representation of $x$
	\item $s = v + \sigma_u u + \sigma_v v$, the internal representation of running sum $s$
	\item $\iota = \inner{x}{v}_w$, the weighted inner product of the iterate with fixed vector $v$
	\item $\nu = \norm{x}_w$, the weighted norm of the iterate
\end{itemize}

To support sampling, our data structure also maintains a binary tree $\dist_x$ of depth $O(\log n)$. For the node corresponding to $S \subseteq [n]$ (where $S$ may be a singleton), we maintain
\begin{itemize}
	\item $\sum_{j \in S} [w]_j [u]_j^2$, $\sum_{j \in S} [w]_j [u]_j [v]_j$, $\sum_{j \in S} [w]_j [v]_j^2$ 
\end{itemize}

We now give the implementation of the necessary operations for $\WIM_2$, giving additional proofs of correctness when applicable.

\paragraph{Initialization.} 

\begin{itemize}\item$\Initialize(x_0, v, w)$. Runs in time $O(n)$.
	\begin{enumerate}
		\item $(\xi_u, \xi_v, u) \leftarrow (1, 0, x_0)$.
		\item $(\sigma_u, \sigma_v, u') \leftarrow (0, 0, \mathbf{0}_n)$.
		\item $(\iota, \nu) \leftarrow (\inner{x_0}{v}_w, \norm{x_0}_w)$.
		\item Compute and store $\norm{v}_w^2$.
		\item Initialize $\dist_x$, storing the relevant sums in each internal node.
	\end{enumerate}
\end{itemize}

\paragraph{Updates.}

$\Scale(c)$ and $\SumUp()$ follow identically to the analysis of $\IM_2$.

\begin{itemize}
	\item $\AddSparse(j, c)$: $[x]_j \leftarrow [x]_j + c$. Runs in time $O(\log n)$.
	\begin{enumerate}
		\item $u \leftarrow u + \frac{c}{\xi_u} e_j$.
		\item $u' \leftarrow u' - \frac{c\sigma_u}{\xi_u} e_j$.
		\item $\nu \leftarrow \sqrt{\nu^2 + 2c[w]_j[\xi_u u + \xi_v v]_j + c^2[w]_j}$. 
		\item $\iota \leftarrow \iota + c[w]_j [v]_j$.
		\item For internal nodes of $\dist_x$ on the path from leaf $j$ to the root, update $\sum_{j \in S} [w]_j [u]_j^2$, $\sum_{j \in S} [w]_j [u]_j [v]_j$ appropriately.
	\end{enumerate}
	\item $\AddDense(c)$: $x \leftarrow x + cv$. Runs in time $O(1)$.
	\begin{enumerate}
		\item $\xi_v \leftarrow \xi_v + c$.
		\item $\nu \leftarrow \sqrt{\nu^2 + 2c\iota + c^2\norm{v}_w^2}$.
		\item $\iota \leftarrow \iota + c\norm{v}_w^2$.
	\end{enumerate}
\end{itemize}

We demonstrate that the necessary invariants on $\iota, \nu$ are preserved. Regarding correctness of $\AddSparse$, the updates to $u$ and $u'$ are identical to in the analysis of $\IM_2$. Next, because only $[x]_j$ changes, the updates to $\nu, \iota$ are correct respectively by
\begin{align*}
[w]_j \cdot [\xi_u u + \xi_v v + c]_j^2 &= [w]_j \cdot\left([\xi_u u + \xi_v v]_j^2 + 2c[\xi_u u + \xi_v v]_j + c^2\right),\\
[w]_j \cdot \left([\xi_u u + \xi_v v + c]_j\right) \cdot [v]_j &= [w]_j\cdot\left([\xi_u u + \xi_v v]_j \cdot [v]_j + c[v]_j\right).
\end{align*}
Regarding correctness of $\AddDense$,
\begin{align*}
\norm{x + cv}_w^2 &= \nu^2 + 2c\iota + c^2 \norm{v}_w^2, \\
\inner{x + cv}{v}_w &= \iota + c\norm{v}_w^2.
\end{align*}
Here, we used that the invariants $\nu = \norm{x}_w$ and $\iota = \inner{x}{v}_w$ held.
\paragraph{Queries.}
$\Get(j)$ and $\GetSum(j)$ follow identically to the analysis of $\IM_2$.
\begin{itemize}
	\item $\Norm()$: Return $\norm{x}_w$. Runs in time $O(1)$.
	\begin{enumerate}
		\item Return $\nu$.
	\end{enumerate}
\end{itemize}

\paragraph{Sampling.}

To support $\Sample$, we must produce a coordinate $j$ with probability proportional to $[w]_j [x]_j^2$. To do so, we recursively perform the following procedure, where the recursion depth is at most $O(\log n)$, starting at the root node and setting $S = [n]$: the proof of correctness is identical to the proof in the analysis of $\IM_2.\Sample()$.
\begin{enumerate}
	\item Let $S_1, S_2$ be the subsets of coordinates corresponding to the children of the current node.
	\item Using scalars $\xi_u, \xi_v$, and the maintained $\sum_{j \in S_i} [w]_j [u]_j^2$, $\sum_{j \in S_i} [w]_j [u]_j [v]_j$, $\sum_{j \in S_i} [w]_j [v]_j^2$, compute $\sum_{j \in S_i} [w]_j [x]_j^2 = \sum_{j \in S_i} [w]_j [\xi_u u + \xi_v v]_j^2$ for $i \in \{1, 2\}$.
	\item Sample a child $i \in \{1, 2\}$ of the current node proportional to $\sum_{j \in S_i} [w]_j [x]_j^2$ by flipping an appropriately biased coin. Set $S \leftarrow S_i$.
\end{enumerate}

\subsection{$\CIM_2$}
\label{app:CIM}

In this section, we give implementation details for a generalization of $\WIM_2$, which we call $\CIM_2$. It is used in Section~\ref{ssec:vr-l2l1}, when using the sampling distributions~\eqref{eq:vr-l2l1-prob-def-2} and~\eqref{eq:vr-l2l1-prob-def-3}. At initialization, $\CIM_2$ is passed an additional parameter $x_0 \in \R^n$, a reference point. $\CIM_2$ supports all the same operations as $\WIM_2$, with two differences:
\begin{itemize}
	\item For the current iterate $x$, $\CIM_2.\Sample()$ returns a coordinate $j$ with probability proportional to $[w]_j[x - x_0]_j^2$.
	\item $\CIM_2$ supports querying $\norm{x - x_0}_w^2$ in constant time.
\end{itemize}
Because all the other operations, fields, and invariants supported and maintained by the data structure are exactly the same as $\IM_2$, we only discuss the changes made to the binary tree $\dist_{x}$ in this section for brevity. In particular, to support sampling, our data structure also maintains a binary tree $\dist_x$ of depth $O(\log n)$. For the node corresponding to $S \subseteq [n]$ (where $S$ may be a singleton), we maintain
\begin{itemize}
	\item $\sum_{j \in S} [w]_j[u]_j^2$, $\sum_{j \in S} [w]_j[u]_j [v]_j$, $\sum_{j \in S} [w]_j[v]_j^2$ 
	\item $\sum_{j \in S} [w]_j[x_0]_j^2$, $\sum_{j \in S} [w]_j[u]_j [x_0]_j$, $\sum_{j \in S} [w]_j[v]_j [x_0]_j$
\end{itemize}
At initialization, $\CIM_2$ creates this data structure and stores the relevant sums in each internal node. Upon modifications to $u$ due to updates of the form $\AddSparse(j, c)$, $\CIM_2$ propagates the changes along internal nodes of $\dist_x$ on the path from leaf $j$ to the root. Thus, using these maintained values and the stored values $\xi_u, \xi_v$, it is clear that for any appropriate subset $S$, we are able to compute the quantity
\[\sum_{j \in S} [w]_j[\xi_u + \xi_v v - x_0]_j^2 = \sum_{j \in S} [w]_j\left(\xi_u^2 [u]_j^2 + \xi_v^2 [v]_j^2 + 2\xi_u \xi_v [u]_j[v]_j + [x_0]_j^2 + 2\xi_u[u]_j[x_0]_j + 2\xi_v[v]_j[x_0]_j\right)\]
in constant time, admitting the sampling oracle in time $O(\log n)$ by propagating down the tree maintained by $\dist_x$. This proves the desired sampling complexity. Finally, by appropriately querying the stored values in the root node, we can return $\norm{x - x_0}_w^2$ in constant time.

\end{document}